\documentclass[10pt]{article}
\usepackage[T1]{fontenc}
\usepackage[utf8]{inputenc}
\usepackage[a4paper,width=150mm,top=25mm,bottom=25mm]{geometry}

\usepackage{mathpazo,tgpagella}
\usepackage{microtype}
\usepackage{newunicodechar}
\newunicodechar{Š}{\v{S}}
\newunicodechar{č}{\v{c}}
\newunicodechar{ł}{\l}
\newunicodechar{ń}{\'{n}}

\usepackage{macros}
\usepackage{thmtools,thm-restate}
\makeatletter
\let\c@definition\c@theorem
\let\c@remark\c@theorem
\@for\statementcounter:=theorem,proposition,corollary,lemma,claim,theme,fact,observation,definition,remark\do{%
  \expandafter\counterwithout\expandafter{\statementcounter}{section}%
  \expandafter\def\csname the\statementcounter\endcsname{\thesection.\arabic{theorem}}%
  \expandafter\def\csname theH\statementcounter\endcsname{\theHsection.\arabic{theorem}}%
}
\counterwithin{theorem}{section}
\makeatother
\crefname{observation}{Observation}{Observations}
\Crefname{observation}{Observation}{Observations}
\crefname{appendix}{Appendix}{Appendices}
\Crefname{appendix}{Appendix}{Appendices}
\crefname{subappendix}{Appendix}{Appendices}
\Crefname{subappendix}{Appendix}{Appendices}
\crefname{subsubappendix}{Appendix}{Appendices}
\Crefname{subsubappendix}{Appendix}{Appendices}

\newcommand{\e}{\mathrm{e}}
\newcommand{\rk}{\mathsf{rk}}
\renewcommand{\epsilon}{\varepsilon}
\renewcommand{\emptyset}{\varnothing}

\bibliography{refs}

\def\cM{{\mathcal{M}}}
\def\cN{{\mathcal{N}}}
\def\cB{{\mathcal{B}}}
\def\cI{{\mathcal{I}}}

\def\cU{{\mathcal{U}}}
\def\cK{{\mathcal{K}}}
\def\cL{{\mathcal{L}}}
\def\cE{{\mathcal{E}}}
\def\cR{{\mathcal{R}}}
\def\sL{{\mathscr{L}}}

\def\*#1{\boldsymbol{#1}} 
\def\+#1{\mathcal{#1}} 
\def\-#1{\mathrm{#1}} 
\def\=#1{\mathbb{#1}} 
\def\!#1{\mathfrak{#1}} 

\newcommand{\inner}[2]{\left\langle #1, #2\right\rangle}
\newcommand{\tp}[1]{\left(#1\right)}

\newcommand{\eps}{\varepsilon}

\usepackage{xifthen}
\def\oPr{\mathbf{Pr}}
\renewcommand{\Pr}[2][]{ \ifthenelse{\isempty{#1}}
  {\oPr\left[#2\right]}
  {\oPr_{#1}\left[#2\right]} } 

\def\oE{\mathbb{E}}
\renewcommand{\E}[2][]{ \ifthenelse{\isempty{#1}}
  {\oE\left[#2\right]}
  {\oE_{#1}\left[#2\right]} }

\usepackage{xparse}

\def\oVar{\mathbf{Var}}
\renewcommand{\Var}[2][]{ \ifthenelse{\isempty{#1}}
{\oVar\left[#2\right]}
{\oVar_{#1}\left[#2\right]} }

\title{An $\mathsf{FPRAS}$ for Counting Common Bases of Two Matroids} 
\author{Xiaoyu Chen\thanks{Email: \texttt{xiaoyu@mit.edu},
Massachusetts Institute of Technology. Supported by the NSF CAREER grant CCF-2443045, and the Reed Fund at MIT.} \and Kuikui Liu\thanks{Email: \texttt{liukui@mit.edu},
Massachusetts Institute of Technology. Supported by the NSF CAREER grant CCF-2443045, and the Reed Fund at MIT.}}
\date{}

\begin{document}
\maketitle

\begin{abstract}
We design the first fully polynomial-time randomized approximation scheme ($\FPRAS$) for approximately counting and uniformly sampling common bases of two matroids given by their independence oracles. Moreover, our algorithms generalize far beyond this to Hadamard products of two probability measures on the Boolean cube satisfying a simple nonnegative curvature condition. These algorithmic primitives have myriad applications in statistical physics, polyhedral combinatorics, the theory of quantum many-body systems, and beyond.

At a high level, the central difficulty is that the hard constraints imposed by the two input matroids (or the two input measures, more generally) obstruct traditional approaches based on local Markov chains. Relaxing these constraints can restore rapid mixing, but configurations that satisfy the original hard constraints may become exponentially rare if this is done na\"{i}vely. Our main contribution is a novel approach for approximate counting that marries Markov chain Monte Carlo (MCMC) with convex programming to overcome these obstacles. More concretely:

\begin{itemize}[leftmargin=1.5em,itemsep=2pt,topsep=3pt]
    \item We relax the intersection by imposing an overlap penalty on the product measure formed by the two input measures. We prove, via an integrated Bochner-type method, that this ``soft intersection'' satisfies a Poincar\'{e} Inequality uniformly over all external fields.
    \item We use samples from the relaxed distribution to solve a dual maximum entropy convex program inspired by the Gurvits capacity. This program reweights the soft intersection measure using external fields to boost the aggregate probability of configurations satisfying the hard intersection constraint, without changing their relative weights. We bound the success probability directly using the uniform Poincar\'{e} Inequality and smallness of the gradient norm.
\end{itemize}

\vspace{\baselineskip}

\textbf{AI Disclosure} GPT-5.6 Sol Ultra and GPT-6 Astra Ultra were heavily used to develop the ideas in this paper. A more complete discussion is included in \cref{subsec:ack}.
\end{abstract}

\newpage
{\small
\tableofcontents
}
\newpage

\section{Introduction}\label{sec:intro}
Matroids are beautiful combinatorial abstractions of the notion of linear independence in vector spaces \cite{Oxley2011}, and arise in a multitude of applications, from combinatorial optimization to error-correcting codes. They also form the foundation for a theory of discrete convexity \cite{Murota2003}. Formally, they are defined as a family of subsets of a finite ground set $\cU$ possessing a certain local \emph{exchange property}; see \cref{subsec:matroids} for background on matroids. 
The inclusion-wise maximal sets within the family are called the \emph{bases} of the matroid.

Matroid intersection is the problem of finding a common basis of two matroids, say $\cM,\cN$, on a common ground set $\cU$ of size $n$, with an important classical example being that of finding a perfect matching in a bipartite graph. The seminal work of Edmonds~\cite{Edmonds1970,Edmonds1979MatroidIntersection} gave the first polynomial-time algorithm for matroid intersection when the two matroids are given by their independence oracles. More generally, it is well-known that one can optimize any linear functional over the matroid intersection polytope in polynomial-time; see e.g. \cite{Schrijver2003} and references therein. Intersections of two matroids lie right at the boundary of tractable optimization problems, since even deciding whether or not there exists a common basis of three or more matroids is $\NP$-hard.

A longstanding open problem in theoretical computer science has been to design an algorithm that (approximately) counts or uniformly samples a common basis of two matroids. This is motivated in part by a natural philosophical question concerning the relationship between the computational complexity of optimization and the computational complexity of sampling. In this paper, we fill in this gap by designing the first polynomial-time algorithms that achieve the following. Let $\cM,\cN$ be two arbitrary matroids specified to the algorithm via their independence oracles, and let $\cB(\cM)$ (resp. $\cB(\cN)$) denote the collection of bases of $\cM$ (resp. $\cN$), respectively.
\begin{itemize}
    \item[] \textbf{Approximate Counting:} Estimate the number $\abs{\cB(\cM) \cap \cB(\cN)}$ of common bases up to $(1 \pm \epsilon)$-multiplicative error, for any desired target accuracy $0 < \epsilon < 1$. Note that exact counting is $\#\mathsf{P}$-hard and hence, unlikely to admit an efficient algorithm \cite{Valiant1979}.
    \item[] \textbf{Approximate Sampling:} Sample a common basis from the uniform distribution over $\cB(\cM) \cap \cB(\cN)$ up to $\pm\epsilon$-additive total variation error, for any $0 < \epsilon < 1$.
\end{itemize}

Counting and sampling common bases of two matroids has received significant attention in the literature over the past several decades. For the moment, we highlight two landmark results that previously solved important special cases of the problem. Jerrum--Sinclair--Vigoda designed the first polynomial-time algorithms for counting and sampling (edge weighted) perfect matchings in bipartite graphs \cite{JerrumSinclairVigoda2004}. As a result, they also obtained the first $\mathsf{FPRAS}$ for estimating the permanent of any nonnegative matrix. Anari--Liu--Oveis Gharan--Vinzant gave the first polynomial-time algorithms for counting and sampling bases of a single matroid \cite{AnariEtAl2019}. Both results are based on the Markov chain Monte Carlo (MCMC) method.

Our algorithms are also based on MCMC and hence, are randomized. We note that randomization is necessary when the matroids can only be accessed via their independence oracles. Indeed, even for a single matroid $\cM$, no efficient deterministic algorithm can approximate $\abs{\cB(\cM)}$ to within a simply exponential multiplicative factor in the oracle model \cite{AzarBroderFrieze1994}. This lower bound is nearly matched by the deterministic $2^{O(r)}$-approximation algorithm of Anari--Oveis Gharan--Vinzant \cite{AnariOveisGharanVinzant2018}, which also applies to common bases of two matroids.

\begin{definition}[$\mathsf{FPRAS}$]\label{def:fpras}
Let $Z\ge0$ be a quantity determined by an input instance. A
\emph{fully polynomial-time randomized approximation scheme}
($\mathsf{FPRAS}$) for $Z$ is a randomized algorithm that, given the
instance and an accuracy parameter $0<\eps\le1$, returns
$\widehat Z\ge0$ satisfying
\begin{align*}
    \Pr{(1-\eps) \cdot Z\le\widehat Z\le(1+\eps) \cdot Z}&\ge\frac34,
\end{align*}
in time polynomial in the size of the instance and $1/\eps$.\footnote{By the standard median trick, returning the median of $O(\log(1/\delta))$ independent calls boosts the success probability to $1-\delta$ for any $0<\delta<1$.} In the oracle model, each oracle call has unit cost.
\end{definition}

\begin{theorem}[Counting Common Bases]\label{cor:common-bases}
There is an $\mathsf{FPRAS}$ that, given independence oracles for two matroids $\cM,\cN$ on a common ground set $\cU$, estimates the number $\abs{\cB(\cM)\cap\cB(\cN)}$ of common bases.
\end{theorem}
\cref{cor:common-bases} and its generalizations (i.e. \cref{thm:log-concave-convolution,thm:curvature-convolution} discussed below) are fundamental algorithmic primitives that have abundant applications in statistical physics, the theory of quantum many-body systems, high-dimensional statistics, and combinatorics. We refer interested readers to \cref{sec:additional-application} for a more comprehensive discussion.

\begin{remark}
The running time of our algorithm is $\widetilde{O}(n^{13}/\epsilon^{2})$, which certainly isn't optimal. Since our goal is to present our algorithms and proofs in the simplest way possible, we made no attempt to optimize the running time.
\end{remark}
All problems we consider in this paper will be self-reducible in the sense of \cite{JVV1986}. Hence, approximate counting and sampling are equivalent under polynomial-time reductions and we will move fluidly between them. In particular, \cref{cor:common-bases} also implies a polynomial-time algorithm for approximately sampling a uniformly random common basis.

As a special case, \cref{cor:common-bases} gives an $\mathsf{FPRAS}$ for counting
\emph{strong orientations} and \emph{undirected Euler tours},
resolving the approximate counting questions in
\cite[Section~3.5]{Jerrum2024PartialRejection} and
\cite{GeStefankovic2010EulerTours}, respectively; see
\cref{app:strong-orientations,app:euler-tours}.
A variant of \cref{cor:common-bases} for weighted common bases (proved below) applies to \emph{network reliability}, where the goal is to estimate the probability of directed root-to-all reachability, or of an unpaired vertex-disjoint linkage between two terminal sets, under independent edge failures; see \cref{app:rooted-reliability,app:reliability-proof}. The former settles the general directed reachability question in \cite[Section~7]{GuoJerrum2019Reliability}; the latter extends the two-terminal result of \cite{FengFuGuo2026} to multiple disjoint paths. Finally, we mention that \cref{cor:common-bases} also yields a new $\mathsf{FPRAS}$ for perfect matchings in bipartite graphs (and the permanent of nonnegative matrices more generally).

\subsection{Subset Convolution of Log-Concave Measures}

\cref{cor:common-bases} is an important special case of ``intersections'' of multi-affine log-concave measures, which we now describe. Let $\mu : 2^{\cU}\to\R_{\geq0}$ be a nonnegative weight function, which need not be normalized to sum to $1$, and define its multi-affine generating polynomial with variables $\{x_{u} : u \in \cU\}$ by
\begin{align}\label{eq:genpoly}
    g_{\mu}(\bm{x}) \coloneqq \sum_{S \subseteq \cU} \mu(S) \cdot \bm{x}^{S},
\end{align}
where $\bm{x}^{S} \coloneqq \prod_{u \in S} x_{u}$. We write $\wrapb{\bm{x}^{S}}g_{\mu}$ for the coefficient $\mu(S)$ of $\bm{x}^{S}$ in the polynomial $g_{\mu}$.
\begin{definition}[Subset Convolution; \cite{BjorklundEtAl2007}]\label{def:subset-convolution}
Let $\mu,\nu : 2^{\cU} \to \R_{\geq0}$ be two nonnegative weight functions. We define their \emph{subset convolution} by the quantity
\begin{align*}
    (\mu \ast \nu)(T) &\coloneqq \sum_{S \subseteq T} \mu(S) \cdot \nu(T \setminus S)
             =\wrapb{\bm{x}^{T}}\wrapp{g_\mu(\bm{x}) \cdot g_\nu(\bm{x})},
    \qquad \forall T\subseteq\cU.
\end{align*}
\end{definition}

\begin{definition}[Log-Concave Polynomial; \cite{Gurvits2009,AnariOveisGharanVinzant2018,BrandenHuh2020}]\label{def:log-concave}
A nonzero polynomial $g$ with nonnegative coefficients is
\emph{log-concave} on $\mathbb R_{\geq0}^{\cU}$ if
\begin{align*}
    \nabla^2\log g(\bm{x}) \preceq 0, \qquad\forall \bm{x}\in\R_{>0}^{\cU}.
\end{align*}
\end{definition}

We will design an $\mathsf{FPRAS}$ for the subset convolution (and a natural generalization of it) of two nonnegative weight functions $\mu,\nu$ assuming $g_{\mu},g_{\nu}$ are log-concave on $\R_{\geq0}^{\cU}$. To bound the running time, let us also define the \emph{coefficient range} $\cR_\mu$ for a weight function $\mu$ as
\begin{align} \label{eq:coefficient-range}
    \cR_{\mu} \coloneqq \frac{\max_{S \in \supp(\mu)} \mu(S)}{\min_{S \in \supp(\mu)} \mu(S)} \geq 1.
\end{align}
To guarantee that our algorithms run in polynomial-time, we will only require that $\cR_{\mu},\cR_{\nu} \leq \exp\wrapp{\mathsf{poly}(n)}$, which is easily satisfied in most applications of interest. The algorithms below take these parameters as input; upper bounds supplied by the user may be used instead, with the same replacement in the running-time bounds.

\begin{theorem}[$\mathsf{FPRAS}$ for Log-Concave Subset Convolution]\label{thm:log-concave-convolution}
Let $\mu,\nu : 2^{\cU} \to \R_{\geq0}$ be two nonnegative weight functions given by evaluation oracles,\footnote{This means that there is a black-box that returns the weight $\mu(S)$ (resp. $\nu(S)$) on input $S \subseteq \cU$. The number of calls to this black-box is counted in the overall running time of our algorithms.} and suppose $g_\mu,g_\nu$ are log-concave. Assume initial sets $S_{0},T_{0} \subseteq \cU$ satisfying $\mu(S_{0})>0,\nu(T_{0})>0$, upper bounds on $\cR_\mu,\cR_\nu$, and an error parameter $0 < \epsilon \leq 1$ are given. Then there is an $\mathsf{FPRAS}$ for any specified coefficient
\begin{align*}
 Z_{\bm{\alpha}} \coloneqq \wrapb{\bm{x}^{\bm{\alpha}}}\wrapp{g_\mu(\bm{x}) \cdot g_\nu(\bm{x})},
 \qquad \bm{\alpha}\in\{0,1,2\}^{\cU}
\end{align*}
that uses at most
\begin{align*}
    \widetilde{O}\wrapp{\frac{n^{11} \cdot \wrapp{n + \log(\cR_{\mu} \cdot \cR_{\nu})}^{2}}{\epsilon^{2}}}
\end{align*}
arithmetic operations and evaluation oracle calls.
\end{theorem}
\begin{proof}[Proof of \cref{cor:common-bases}]
This is an immediate corollary of \cref{thm:log-concave-convolution} in the case $\bm{\alpha} = \allone$ by setting $\mu(S) = \allone[S \in \cB(\cM)]$ and $\nu(T) = \allone[\cU \setminus T \in \cB(\cN)]$, both of which can be evaluated given independence oracles for $\cM,\cN$. Log-concavity of $g_{\mu},g_{\nu}$ was established in \cite{AnariOveisGharanVinzant2018,BrandenHuh2020}, and we can find initializing sets $S_{0},T_{0}$ by greedy algorithms.
\end{proof}

We highlight three further applications of the more general
\cref{thm:log-concave-convolution} that aren't captured directly by \cref{cor:common-bases}. Each admits an $\mathsf{FPRAS}$;
precise definitions and reductions appear in \cref{sec:additional-application}.
\begin{itemize}
 \item \textbf{Mixed discriminants.} Approximate the mixed discriminant
 of positive semidefinite matrices, a generalization of the permanent
 \cite{CelisEtAl2017}; see \cref{app:determinantal-applications}.
 This resolves the mixed-discriminant question in
 \cite[Section~4, Problem~2]{Gurvits2007MixedVolume}, disproving the
 conjectured nonexistence of an $\mathsf{FPRAS}$. In upcoming work~\cite{ChenLiu2026MixedVolumes}, we use this mixed discriminant algorithm
 as a key ingredient to design an $\mathsf{FPRAS}$ for computing mixed volumes of
 well-presented compact convex sets.
 \item \textbf{The six-vertex model.} Estimate its partition function
 on loopless $4$-regular multigraphs when the nonnegative weights satisfy
 $a\le b+c$, $b\le a+c$, and $c\le a+b$; see \cref{app:six-vertex}.
 For positive weights, this completes the general-graph approximation
 classification left open by \cite[Section~6]{CaiLiuLu2019}.
 \item \textbf{Quantum Max-Cut and the EPR model.} Estimate the optimal
 value of Quantum Max-Cut on bipartite graphs and of the EPR model on
 arbitrary graphs, with nonnegative edge weights; see \cref{app:quantum}.
 This resolves the bipartite complexity question in
 \cite[Open question~3.4]{Gharibian2023QuantumNP} and proves the
 EPR-in-$\mathsf{BPP}$ conjecture of
 \cite[Conjecture~1]{MarwahaSud2026}.
\end{itemize}

\subsection{Counting and Sampling for Hadamard Products via Curvature}
The key input into our approximate counting and sampling algorithms is a Poincar\'{e} Inequality for the Hadamard product of two probability measures $\mu,\nu$ assumed to satisfy suitable ``curvature'' conditions. {For the Hadamard-product dynamics and Poincar\'{e} inequality below, let $\mu,\nu : 2^{\cU} \to \R_{\geq0}$ have a common nonempty support $\Omega$. For weights with different supports, first restrict both to their support intersection and set them to zero outside it.}
\begin{definition}[Hadamard Product]\label{def:hadamard}
Let $\rho = \mu \odot \nu$ denote the probability measure supported on $\Omega$ defined by $(\mu \odot \nu)(S) \propto \mu(S) \cdot \nu(S)$ for all $S \subseteq \cU$.
\end{definition}
Let us now introduce a natural local Markov chain that samples from $\rho = \mu \odot \nu$. {Write $S+u\coloneqq S\cup\{u\}$ and $S-u\coloneqq S\setminus\{u\}$.} For each $u \in \cU$, define the Laplacian $\sL_{u}$ via the following operator on functions $f : 2^{\cU} \to \R$, {evaluated at $S\in\Omega$}:
\begin{align*}
    (\sL_{u}f)(S) = r_{u}(S) \cdot \wrapp{f(S + u) - f(S - u)} \enspace \text{where} \enspace r_{u}(S) \coloneqq \begin{cases}
        -\mu(S + u) / \mu(S - u), &\quad \text{if } u \notin S \\
        \nu(S - u) / \nu(S + u), &\quad \text{if } u \in S
    \end{cases}
\end{align*}
gives the relative birth and death rates. The operator $\sL_{u}$ encodes the updates at coordinate $u$. The full Laplacian is defined as
\begin{align}\label{eq:full-laplacian}
 \sL&\coloneqq\sum_{u\in\cU}\sL_u.
\end{align}
It acts on $S\in\Omega$ with infeasible additions {and deletions} omitted.
\begin{fact}
$\sL_{u}$ is reversible with respect to $\rho$ for all $u \in \cU$, i.e. $\langle f, \sL_{u}g \rangle_{\rho} = \langle \sL_{u}f, g \rangle_{\rho}$ for all functions $f,g : 2^{\cU} \to \R$. Moreover, $\sL_{u}$ is positive semidefinite, i.e. $\langle f, \sL_{u} f \rangle_{\rho} \geq 0$ for all $f : 2^{\cU} \to \R$.
\end{fact}
\begin{remark}[Discrete-time implementation]\label{rem:discrete-chain}
The continuous-time chain with generator $-\sL$ can be simulated by a
simple discrete-time chain. Let $B\ge\max_{S\in\Omega,\,u\in\cU}|r_u(S)|$
be any known bound on the rates. From the current state $S\in\Omega$:
\begin{enumerate}[label=(\arabic*),nosep,leftmargin=2em]
 \item With probability $1/2$, stay at $S$.
 \item Otherwise, choose $u\in\cU$ uniformly at random, and
 \begin{itemize}[nosep,leftmargin=1.5em]
  \item if $u\notin S$ and $S+u\in\Omega$, move to $S+u$ with probability
  $\mu(S+u)/(B\mu(S))$;
  \item if $u\in S$ {and $S-u\in\Omega$}, move to $S-u$ with probability $\nu(S-u)/(B\nu(S))$;
  \item otherwise, stay at $S$.
 \end{itemize}
\end{enumerate}
Each probability in (2) is at most $1$ by the choice of $B$. The
resulting transition matrix is $P=I-\sL/(2nB)$, which is lazy and
reversible with respect to $\rho$, and its spectral gap is that of $\sL$
divided by $2nB$. The Poincar\'e inequality below therefore gives the
mixing bound \cref{eq:poincare-mixing} with $R=nB$.
\end{remark}

Having defined the dynamics, let us now introduce the local ``curvature'' criterion that will be used to deduce a Poincar\'{e} Inequality for $\sL$.

\begin{definition}[(Upward) $\kappa$-Spectral Curvature]\label{def:spectral-curvature}
Fix a parameter $\kappa \in \R$. For a nonnegative weight function $\mu : 2^{\cU} \to \R_{\geq0}$ that is not identically zero, and each $S \in \supp(\mu)$, define the $\cU \times \cU$ \emph{(upward) $\kappa$-curvature matrix}\footnote{Note that the only nonzero entries of $A_{S}^{\mu,\kappa}$ come from elements in $\cU \setminus S$.}
\begin{align*}
    A_{S}^{\mu,\kappa}(u,v) \coloneqq \begin{cases}
        \frac{\mu(S + v)^{2}}{\mu(S)^{2}} - \kappa \cdot \frac{\mu(S + v)}{\mu(S)}, &\quad\text{if } u,v \notin S, u = v \\
        \frac{\mu(S + u)}{\mu(S)} \cdot \frac{\mu(S + v)}{\mu(S)} - \frac{\mu(S + u + v)}{\mu(S)}, &\quad\text{if } u,v \notin S, u \neq v \\
        0, &\quad\text{otherwise}
    \end{cases}
\end{align*}
We say $\mu$ has \emph{(upward) $\kappa$-spectral curvature} if $A_{S}^{\mu,\kappa} \succeq 0$ for every $S \in \supp(\mu)$.
\end{definition}
\begin{remark}
In this paper, we are primarily interested in the $\kappa \geq 0$ regime. However, we note that $\kappa$ can be any real number; see \cref{app:curvature-examples} for examples.
\end{remark}
\begin{remark}
Unlike the influence/correlation/covariance matrices appearing in the theory of spectral independence \cite{ALO2020} and localization schemes \cite{ChenEldan2022}, whose entries are $\#\mathsf{P}$-hard to compute exactly, the entries of $A_{S}^{\mu,\kappa}$ are straightforward to compute for any $S \in \supp(\mu)$. Hence, for many applications, verifying spectral curvature can be done easily by hand.
\end{remark}

The notion of curvature appearing in \cref{def:spectral-curvature} has previously appeared in the literature in various guises; we discuss these connections in greater depth in \cref{subsec:prior-notions}. Before we continue, note that we can also define a ``downward'' analog of $\kappa$-spectral curvature by considering the (upward) $\kappa$-spectral curvature of the \emph{dual $\mu^{*}$}, defined as follows.
\begin{definition}[Dual]\label{def:dual-measure}
Let $\mu : 2^{\cU} \to \R_{\geq0}$ be nonzero. We define its \emph{dual} $\mu^{*} : 2^{\cU} \to \R_{\geq0}$ by $\mu^{*}(S) \coloneqq \mu(\cU \setminus S)$ for all $S \subseteq \cU$.
\end{definition}
We elaborate on the role of the ``downward'' analog of $\kappa$-spectral curvature in \cref{sec:poincare-proof}. Our main technical theorem is the following.

\begin{theorem}[Main Poincar\'{e} Inequality]\label{thm:poincare}
Let $\mu,\nu : 2^{\cU} \to \R_{\geq0}$ have {a common nonempty support $\Omega$, connected under single-element additions and deletions}, and let $\rho = \mu \odot \nu$ be their Hadamard product, which is supported on $\Omega$. If $\mu$ has $\eta$-spectral curvature and $\nu^{*}$ has $\kappa$-spectral curvature, where $\eta,\kappa \in \R$ satisfy $\eta + \kappa > 0$, then $\rho$ obeys the Poincar\'{e} Inequality
\begin{align*}
    (\eta + \kappa) \cdot \Var[\rho]{f} \leq \langle f, \sL f \rangle_{\rho}, \qquad \forall f : \Omega \to \R.
\end{align*}
\end{theorem}

{
A family $\mathcal F\subseteq 2^W$ is \emph{downward closed} if it contains every subset of each of its members, and \emph{upward closed} if it contains every superset contained in $W$ of each of its members.

\begin{remark}
By \cref{lem:positive-curvature-support}, a positive spectral-gap bound from \cref{thm:poincare} forces $\Omega$ to be upward or downward closed after removing coordinates that are constant on $\Omega$.
\end{remark}
}


Through \Cref{thm:poincare}, we can design an $\mathsf{FPRAS}$ for the subset convolution of a pair of downward closed measures that satisfy nonnegative curvature.
\begin{theorem}\label{thm:curvature-convolution}
Suppose $\mu,\nu : 2^{\cU} \to \R_{\geq0}$ have downward closed supports and upward spectral curvatures $\eta,\kappa$, with $\eta+\kappa\ge0$.
Fix an error parameter $0<\epsilon<1$.
Put $\gamma\coloneqq \max\{\eta,\kappa\}$; if $\gamma>0$,
let $\chi\ge1$ be a supplied bound, polynomial in $n$ and $\epsilon^{-1}$, such that, for each
$\omega\in\{\mu,\nu\}$,
\begin{align*}
 \omega(S+u)&\le\chi \cdot \gamma \cdot \omega(S),
 \qquad \forall S\in\supp\omega, \, \forall u\notin S.
\end{align*}
In all other cases, take $\chi=1$ without imposing this bound.
Given the curvature parameters $\eta,\kappa$, value oracles for $\mu,\nu$, upper bounds on $\cR_\mu,\cR_\nu$, and $\epsilon$, there is an
$\mathsf{FPRAS}$ for any specified coefficient
\begin{align*}
 Z_{\bm{\alpha}}&\coloneqq [\bm{x}^{\bm{\alpha}}]\bigl(g_\mu(\bm{x}) \cdot g_\nu(\bm{x})\bigr),
 \qquad \bm{\alpha}\in\{0,1,2\}^{\cU}
\end{align*}
that uses at most
\begin{align*}
 \widetilde{O}\wrapp{\frac{n^{11}(n+\log(\cR_\mu\cR_\nu))^2}{\eps^2}
       +\frac{n^5\chi(n+\log(\cR_\mu\cR_\nu))}{\eps^3}}
\end{align*}
arithmetic operations and evaluation oracle calls.
\end{theorem}

\begin{remark}\label{rmk:curvature-LC-generalization}
\Cref{thm:curvature-convolution} can be viewed as a generalization of
\cref{thm:log-concave-convolution}. For log-concave measures, translating
$g_\mu(\bm{x}),g_\nu(\bm{x})$ to
$g_\mu(\mathbf1+\bm{x}),g_\nu(\mathbf1+\bm{x})$ gives downward-closed measures
with zero curvature by \cref{prop:flc-curvature}.
After reducing to homogeneous polynomials whose degrees sum to $n$
(see \cref{app:convolution-proofs}), these translations preserve the
coefficient of $\bm{x}^{\cU}$. The translated weights, i.e. the coefficients of $g_{\mu}(\allone + \bm{x}), g_{\nu}(\allone + \bm{x})$, are respectively
\[
 \widehat\mu(S)=\sum_{T\supseteq S}\mu(T),
 \qquad \widehat\nu(S)=\sum_{T\supseteq S}\nu(T).
\]
These sums each admit an $\mathsf{FPRAS}$ since $\mu,\nu$ are log-concave~\cite{AnariEtAl2019}.
Using sufficiently accurate estimates of the queried weights and
controlling the accumulated sampling error gives the reduction,
with additional polynomial overhead. We state
\cref{thm:log-concave-convolution} separately to give a direct
algorithm using the original evaluation oracles and a sharper
running-time bound.
\end{remark}

{
\begin{remark}[Optimality of the curvature threshold]\label{rmk:curvature-threshold}
The condition $\eta+\kappa\ge0$ in \cref{thm:curvature-convolution}
is optimal among constant lower bounds on the curvature sum.
For every fixed $\delta>0$, there is no $\mathsf{FPRAS}$ under the
weaker condition $\eta+\kappa>-\delta$, unless
$\mathsf{NP}=\mathsf{RP}$. This holds even when $\mu=\nu$ have
downward-closed support, polynomial-time value oracles, and
$\log(\cR_\mu\cR_\nu)=O(n\log n)$.
The proof is deferred to \cref{app:curvature-threshold}.
\end{remark}
}

\paragraph{Additional Applications of \cref{thm:curvature-convolution}.} As noted in \cref{rmk:curvature-LC-generalization}, \cref{thm:curvature-convolution} generalizes \cref{thm:log-concave-convolution}. However, we highlight a couple interesting applications of the former that aren't captured by the latter. In particular, we show that each of the following approximate counting problems on simple graphs admits an
$\mathsf{FPRAS}$. Both counts include all sizes; precise definitions
and proofs appear in \cref{app:curvature-applications}.
\begin{itemize}
 \item \textbf{Rainbow matchings.} Count matchings with no repeated
 edge color, for an arbitrary edge coloring; see
 \cref{cor:rainbow-matchings}.
 This counts all sizes and does not address the rainbow perfect-matching
 enumeration studied by \cite{PerarnauSerra2013}.
 \item \textbf{Degree-constrained forests.} Count forests satisfying
 prescribed upper bounds on the vertex degrees. When every bound is
 two, this counts unions of vertex-disjoint paths; see
 \cref{cor:degree-constrained-forests}.
 This adds vertex-degree bounds to the forest-counting result of
 \cite{AnariEtAl2019}.
\end{itemize}

\subsection{Connections with Previous Notions}\label{subsec:prior-notions}
We compare the spectral curvature condition in
\cref{def:spectral-curvature} with earlier notions that appeared in the literature.

\paragraph{Trickle-Down Condition for Field Dynamics} Previous work of \cite{ChenEtAl2025RandomRegular} derived a sufficient condition for the so-called \emph{field dynamics} \cite{ChenFengYinZhang2024} to satisfy a Poincar\'{e} Inequality based on a certain ``pairwise dependency matrix''. This matrix turns out to be equivalent to the curvature matrix appearing in \cref{def:spectral-curvature} up to a simple transformation, which we now describe.

Let $\mu:2^{\cU}\to\mathbb R_{\ge0}$ be a nonzero weight function with downward closed support. For each $S\in\supp(\mu)$, put
\begin{align*}
    D_S \coloneqq \diag\wrapp{\frac{\mu(S + u)}{\mu(S)} : u \notin S, \mu(S + u) > 0}.
\end{align*}
The matrices below are indexed by the feasible
additions $u\notin S$ with $\mu(S+u)>0$.


\begin{definition}[Pairwise dependency matrix {\cite{ChenEtAl2025RandomRegular}}]
\label{def:pairwise-dependency}
Let $S\in\supp(\mu)$. The \emph{pairwise dependency matrix} of $\mu$
at $S$ is the symmetric matrix $M_S$ with entries
\begin{align*}
 M_S(u,v)&\coloneqq
 \begin{cases}
 \dfrac{\mu(S+u+v) \cdot \mu(S)}{\mu(S+u) \cdot \mu(S+v)}-1,&u\ne v,\\
 0,&u=v.
 \end{cases}
\end{align*}
\end{definition}
A straightforward calculation reveals the identity
\begin{align*}
 A_S^{\mu,\kappa}
 &=D_S \cdot \wrapp{I-M_S-\kappa D_S^{-1}} \cdot D_S.
\end{align*}
Here we restrict $A_S^{\mu,\kappa}$ to the feasible additions;
downward closure makes all other rows and columns zero.
Thus the trickle-down criterion for field dynamics
$M_S\preceq I+(1-\delta)D_S^{-1}$ for every $S$, with
$0<\delta<1$, in \cite[Theorems~1.9 and~1.14]{ChenEtAl2025RandomRegular}
is exactly upward $(\delta-1)$-spectral curvature.
The stronger condition $M_S\preceq I$ is the case $\kappa=0$.

A shorter proof of the spectral gap bound under the aforementioned trickle-down criterion appears in \cite{GoebelEtAl2026}, which is based on an integrated Bochner method. Our proof of \cref{thm:poincare} is a direct generalization of their argument.


\begin{remark}
For $\delta>0$, the $\delta$-ultra-log-concavity of
\cite[Definition~3.5]{ChenWang2026} extends this criterion to measures
on $\mathbb N^{\cU}$ with downward-closed support.
On the Boolean cube, it is exactly upward $(\delta-1)$-spectral
curvature; the case $\delta=1$ corresponds to zero upward curvature.
\end{remark}

\paragraph{Log-concave polynomials.}
Log-concavity of the generating polynomial gives another way to bound
spectral curvature.

\begin{definition}[Fractional log-concavity {\cite{AlimohammadiAnariShiragurVuong2021}}]
\label{def:fractional-log-concavity}
Let $0<\alpha\le1$, and let $p\in\mathbb R[x_1,\ldots,x_n]$
be a nonzero polynomial with nonnegative coefficients.
We say that $p$ is \emph{$\alpha$-fractionally log-concave} if
$\bm{x}\mapsto\log p(x_1^\alpha,\ldots,x_n^\alpha)$ is concave on
$\mathbb R_{>0}^n$.
\end{definition}
For multiaffine $p$, this gives a curvature bound for the coefficients
of $p(\mathbf1+\bm{x})$.

\begin{proposition}
\label{prop:flc-curvature}
Let $p$ be nonzero, multiaffine, and $\alpha$-fractionally log-concave,
with nonnegative coefficients and $0<\alpha\le1$. The coefficient
measure $\widehat\mu$ defined by
$g_{\widehat\mu}(\bm{x})\coloneqq p(\mathbf1+\bm{x})$ has downward-closed support and
upward $(1-\alpha^{-1})$-spectral curvature.
\end{proposition}
The proof is given in \cref{app:flc-curvature}.

\begin{remark}
Lorentzian polynomials extend the theory of multi-affine homogeneous log-concave polynomials to polynomials with $\mathrm M$-convex supports
\cite[Theorems~2.30 and~3.10]{BrandenHuh2020}.
For a nonzero multiaffine Lorentzian polynomial $p$,
\cref{prop:flc-curvature} with $\alpha=1$ gives zero upward spectral
curvature for the coefficient measure of $p(\mathbf1+\bm{x})$.
\end{remark}

$0$-spectral curvature is closely related to log-concavity through the
trickle-down principle~\cite{Oppenheim2018} and analogous local-to-global
results for polynomials~\cite{AnariEtAl2019,BrandenHuh2020}.
These polynomial results are stated for homogeneous polynomials,
meaning that every monomial has the same degree.
\cite{ChenEtAl2025RandomRegular} gives a direct extension to generating
polynomials that need not be homogeneous, for measures with downward-closed
support, yielding the following result.
\begin{proposition}[{\cite[Proposition~3.8]{ChenEtAl2025RandomRegular}}]
\label{prop:zero-curvature-log-concavity}
Let $\mu : 2^{\cU} \to \R_{\geq0}$ have downward closed support. If $\mu$ has (upward) $0$-spectral curvature, then $g_\mu$ is log-concave.
\end{proposition}

\subsection{Prior Work}\label{subsec:prior-work}

\paragraph{Trickle-down and field dynamics.}
Trickle-down arguments grew out of earlier local-to-global ideas
\cite{Garland1973,KaufmanMass2017,DinurKaufman2017}
and establish rapid mixing of ``global'' random walks by
analyzing related ``local'' random walks, which are usually easier to study
\cite{Oppenheim2018,KaufmanOppenheim2020,AlevLau2020,AnariEtAl2019,ALO2020}.
Refinements and extensions appear in
\cite{GuoMousa2020,ALG2021,AbdolazimiOveisGharan2023,GotlibKaufman2023,AnariKoehlerVuong2024,LeakeLindbergOveisGharan2025,LeakeOveisGharan2025,GuoZhang2026Planar,BLO26}.
Related sampling results appear in
\cite{ChenLiuVigodaContraction2021,ChenLiuVigoda2021,FGYZ2022,BlancaEtAl2022,Liu2021Coupling,AnariEtAl2024}.
Field dynamics~\cite{ChenFengYinZhang2024} provides a mechanism
for reducing sampling from a ``hard'' distribution to sampling from
relatively ``easier'' conditional distributions with smaller external fields.
This approach and related
methods~\cite{Eldan2013,EldanKoehlerZeitouni2022,ChenEldan2022,ChenFengYinZhang2022Optimal}
have led to rapid-mixing results for the hardcore model (and more general antiferromagnetic two-spin systems)
in the tree-uniqueness regime and at the uniqueness
threshold~\cite{ChenChenYinZhang2025,ChenJuMiaoYinZhang2026}.
Our curvature condition builds on top of a criterion that lies between these two lines of work
\cite{ChenEtAl2025RandomRegular}.

\paragraph{Bochner-Type Methods.}
Bochner-type methods bound the spectral gap of random walks by relating it to the local ``curvature'' of the underlying space. The terminology originates from Bochner's formula in Riemannian geometry \cite{Bochner1946,Lichnerowicz1958}, and its application to continuous diffusions in high-dimensional probability is classical \cite{BakryEmery1985}; see also \cite{BGL2014}.
Adaptations of these techniques to discrete spaces appear in
\cite{BCDP2006,CaputoPosta2007,Caputo2008,CDP2009,KondratievKunaOhlerich2013,DaiPraPosta2013,Johnson2017,JungelYue2017,WeberZacher2021,Pedrotti2025}.
A recent flurry of work, starting with \cite{GoebelEtAl2026}, applies a variant of this method and closely related arguments for
spin systems on random regular and planar graphs
\cite{GoebelEtAl2026,GuoZhang2026Planar}, the Sherrington--Kirkpatrick model
\cite{Wang2026SK,BLO26}, graph colorings
\cite{ChenLiu2026Girth,Haqi2026LineGraph}, and ``ultra-log-concave'' distributions
\cite{ChenWang2026}.
It also gives simpler proofs of known trickle-down results; see
\cite{ChenLiu2026DownUp} and references therein.
Our proof of \cref{thm:poincare} uses this method.

\paragraph{Geometry of polynomials.}
{
Properties of generating polynomials have been used to prove inequalities
for combinatorial counts and probabilities
\cite{BorceaBrandenLiggett2009,Gurvits2009,Huh2012,HuhKatz2012,HuhWang2017,AdiprasitoHuhKatz2018,AnariOveisGharanVinzant2018,AnariEtAl2018Mason,BrandenHuh2020}.

Gurvits capacity and related entropy maximization programs connect
these inequalities to approximate counting and optimization
\cite{Gurvits2006,Gurvits2007MixedVolume,NikolovSingh2016,StraszakVishnoi2017,GurvitsLeake2021Capacity}.
Applications include permanents
\cite{LinialSamorodnitskyWigderson2000,Gurvits2006},
bases and common bases of matroids
\cite{AnariOveisGharan2017,StraszakVishnoi2017,AnariOveisGharanVinzant2018},
contingency tables and integer flows
\cite{Barvinok2009Contingency,BrandenLeakePak2023},
and lattice points in bounded totally unimodular polytopes
\cite{LeakeMohammadiYekta2026TU}.
These methods often give polynomial-time algorithms with exponential
approximation ratios. Recent work has improved deterministic approximation
guarantees for the permanent
\cite{Anari2026BeyondBethe,KudriaLuoMajid2026,DongJain2026Permanent}.

For sampling and randomized counting on matroids, earlier algorithms
covered special classes
\cite{FederMihail1992,JerrumSon2002,AnariOveisGharanRezaei2016}.
Later advances based on log-concavity handled general matroids and
improved the running times
\cite{AnariEtAl2019,CryanGuoMousa2021,AnariEtAl2021LCIV,AnariDerezinski2020}.
}

\paragraph{Annealing and adjusting weights.}
The Markov-chain approach to approximating the permanent was
initiated by Broder~\cite{Broder1986} and developed by Jerrum and
Sinclair~\cite{JerrumSinclair1989}.
Jerrum, Sinclair, and Vigoda~\cite{JerrumSinclairVigoda2004} gave the first
$\mathsf{FPRAS}$ for the permanent of an arbitrary nonnegative matrix.
This algorithm and subsequent $\mathsf{FPRAS}$ for the permanent
\cite{BezakovaEtAl2008,ChenVigodaYang2026,ChenGuoVigodaYang2026}
and two-terminal reliability \cite{FengFuGuo2026} add extra
configurations to make sampling easier, and gradually change the input
weights (annealing; see \cite{StefankovicVempalaVigoda2009,GuoEtAl2026})
while learning weights that keep desired
configurations, such as perfect matchings, from becoming rare.
Those algorithms learn correction weights along a carefully designed annealing
path; our method chooses fields directly by stochastic gradient descent.
The resulting fields make the desired pairs frequent.
We then estimate the total weight using the standard reduction from sampling
to counting~\cite{JVV1986}.
See \cref{subsec:field-optimization} for a detailed discussion.

\subsection{Acknowledgements}\label{subsec:ack}

This work was initiated and completed while the authors were visiting the Simons Institute for the Theory of Computing for the program on Spectral Theory Beyond Graphs (Fall 2026), where most of the discussions behind it took place. We are grateful to the Simons Institute for its hospitality.

The authors thank Shayan Oveis Gharan and Amit Rajaraman for asking clarifying questions that helped significantly improve the exposition of this paper. K.L. thanks Nima Anari, Shayan Oveis Gharan, Francisco Pernice, Amit Rajaraman, and Piyush Srivastava for insightful discussions on matroid intersection and related topics over the years. K.L. also thanks Ronen Eldan for providing access to GPT Pro.

\paragraph{AI Disclosure} For transparency, we detail the progression of this project and the role AI tools played. While the authors did provide many high level directions, we emphasize that AI contributions were critical to the success of this project.

This project began with a much more modest problem, namely establishing rapid mixing of Glauber dynamics for sampling from the distribution $\mu(S) \propto \lambda^{\abs{S}}$ for common independent sets $S$ of all sizes (see \cref{app:matroid-monomer-dimer}). Although the authors were unaware of any previous work formulating this specific problem, it is a somewhat obvious generalization of the monomer-dimer model on all matchings of a graph, for which Glauber dynamics provably mixes rapidly for all $\lambda$ \cite{JerrumSinclair1989,ChenLiuVigoda2021,ChenEtAl2025JerrumSinclair,MiaoYin2026Matchings,ChenChenZhang2026Holant}. Based on this, as well as rapid mixing for a similar distribution over all independent sets of a single matroid, the authors conjectured that no dynamical phase transition occurs, and that the mixing time remains polynomial for every fixed $\lambda \geq 0 $. This conjecture was previously posed by the authors, for instance, during ``The geometry of polynomials in combinatorics and sampling'' workshop at the American Institute of Mathematics in early March 2025. The main motivation for studying this specific problem was two-fold:
\begin{itemize}
    \item Its resolution would provide a first stepping stone towards an $\mathsf{FPRAS}$ for common bases of two matroids, paralleling the chronological development of the first $\mathsf{FPRAS}$ for the permanent of nonnegative matrices.
    \item Independent of whether or not it would eventually lead to an $\mathsf{FPRAS}$ for common bases of two matroids, the authors believed this problem was interesting because all previously known techniques for proving rapid mixing appear to fail for this model: {There are simple examples of pairs of matroids for which the spectral-independence parameter of the resulting distribution on common independent sets is $\Omega(\mathsf{poly}(n))$.\footnote{For example, at $\lambda=1$, the monomer--dimer model on an even cycle of length $\Theta(n^{1/3})$ with each edge replaced by $\Theta(n^{2/3})$ parallel edges has  $\Theta(n^{1/3})$-spectral-independence, where $n$ is the total number of edges; see \cite{ChenYangYinZhang2025} for the parallel-edge construction. This bipartite graph's matchings are exactly the common independent sets of two partition matroids.} It is also believed that constructing canonical paths or flows with polynomial congestion would be hard in general.}
\end{itemize}
In light of recent exciting work of \cite{GoebelEtAl2026}, the authors suspected that (integrated) Bochner-type methods had a chance of overcoming these obstacles. The authors instructed GPT 5.6 Sol Ultra to take inspiration from \cite{GoebelEtAl2026} and prove rapid mixing for all $\lambda \geq 0$ using an integrated Bochner-type argument. Initially, GPT 5.6 Sol Ultra worked on this task for 17 hours without success. Fortuitously, GPT 6 Astra Ultra was released at this time, and it was able to solve the problem with an additional 3 hours. The result was a fairly short and simple proof, parts of which are reproduced in \cref{sec:poincare-proof}, that was tailored to this specific problem. During the process of digesting this proof, the authors independently realized that the argument works much more generally for intersections of probability measures satisfying suitable ``curvature'' conditions, and that this notion of ``curvature'' is essentially equivalent to a notion contained in prior work of \cite{ChenEtAl2025RandomRegular}. This is how we came to the formulation of \cref{thm:poincare}.

{
Our progress on sampling common independent sets of matroids led us to ask
whether similar ideas could give an $\mathsf{FPRAS}$ for common bases.
To explore this question, we provided ChatGPT Pro with several separate drafts:
\begin{itemize}[leftmargin=2em,labelsep=0.5em,itemsep=2pt,parsep=0pt,topsep=4pt,partopsep=0pt]
    \item a revised note on sampling common independent sets;
    \item a revised note on a common framework for recent progress on Bochner-type methods;
    \item a draft on a faster $\mathsf{FPRAS}$ for the permanent~\cite{ChenGuoVigodaYang2026}.
\end{itemize}

ChatGPT first identified an
obstruction to obtaining an $\mathsf{FPRAS}$ directly from our sampler for
common independent sets. We replied that this was understandable and suggested
an annealing approach, as in algorithms for the permanent. It agreed, and
after a further prompt to ``keep trying,'' it found a solution.

While studying the proof produced by ChatGPT, we recognized the underlying
stochastic gradient descent framework (\cref{subsec:field-optimization}).
With further help from ChatGPT,
we generalized the $\mathsf{FPRAS}$ for common bases to subset convolution,
both for measures with log-concave generating polynomials and for measures
with downward closed supports satisfying our nonnegative curvature condition
(\cref{thm:log-concave-convolution,thm:curvature-convolution}).
We also used ChatGPT to search for reductions of other known problems
to these subset-convolution problems (see \cref{sec:additional-application}).
}

\paragraph{Concurrent Work} On October~6, 2026, one day after we posted the first version of this paper to arXiv,
\href{https://github.com/openai/math}{OpenAI} publicly released concurrent
work giving $\mathsf{FPRAS}$ for counting common bases of two matroids,
common integer bases of two integral polymatroids, and cell-bounded
contingency tables
\cite{OpenAI2026CommonBases,OpenAI2026PolymatroidBases,OpenAI2026ContingencyTables}.
Their analysis uses transport flows~\cite{ChenEtAl2025JerrumSinclair}
and restricted Poincar\'e inequalities~\cite{ChenVigodaYang2026},
with annealing and learned defect weights closely following
the approach of Jerrum, Sinclair, and Vigoda for the permanent
\cite{JerrumSinclairVigoda2004}.
Our analysis uses an integrated Bochner-type method, with external fields
chosen by convex programming.

Their polymatroid and contingency-table results allow binary-encoded,
exponentially large capacities by reducing to polynomial-size weighted
instances. Our corresponding reductions expand the coordinates into
individual copies and therefore give polynomial-time guarantees for
unary capacities or polynomially bounded entry ranges
(\cref{prop:integer-weight-intersection,cor:weighted-contingency-tables}).
Our subset-convolution framework extends in other directions, including
six-vertex partition functions under the triangle inequalities,
PSD mixed discriminants, and normalizing constants of products of two
determinantal point processes
(\cref{cor:six-vertex,cor:determinantal-applications}).

\paragraph{Follow-Up Work} In upcoming work~\cite{ChenLiu2026MixedVolumes}, we use the $\FPRAS$ for computing general mixed discriminants given by \cref{cor:determinantal-applications} to design the first $\FPRAS$ for computing mixed volumes of general convex bodies specified by (weak) membership oracles.

\section{Technical Overview}\label{subsec:convolution-outline}

\subsection{Poincar\'e Inequality via an Integrated Bochner Method}\label{subsubsec:outline-bochner}
We prove \cref{thm:poincare} in \cref{sec:poincare-proof} by an integrated
Bochner-type method, which bounds the spectral gap of a Laplacian using its
square. The following abstract criterion holds for the Laplacian $\cL$
of any irreducible reversible Markov chain; see
\cref{subsec:markov-chains} for the notation and proof.

\begin{proposition}\label{prop:bochner-criterion}
Let $\cL$ be the Laplacian of an irreducible reversible Markov chain on
a finite state space $\Omega$, with stationary distribution $\pi$. Then
its spectral gap satisfies
\begin{align*}
 \gamma\coloneqq \inf_{f\text{ nonconstant}}
   \frac{\langle f,\cL f\rangle_\pi}{\Var[\pi]{f}}
 &=\inf_{f\text{ nonconstant}}
   \frac{\langle\cL f,\cL f\rangle_\pi}{\langle f,\cL f\rangle_\pi}.
\end{align*}
\end{proposition}

{We apply \cref{prop:bochner-criterion} to the Laplacian $\sL$
defined in \cref{eq:full-laplacian}.}
{The common support $\Omega$ is connected under single-element
additions and deletions, and every feasible transition has positive rate,
so the chain with Laplacian $\sL$ is irreducible.}
Thus, to prove $\gamma\ge\eta+\kappa$, it suffices to show
\begin{align*}
 \langle\sL f,\sL f\rangle_\rho\ge(\eta+\kappa) \cdot \langle f,\sL f\rangle_\rho
 \qquad \forall f:\Omega\to\R.
\end{align*}
We prove this by writing
$\langle f,(\sL^2-(\eta+\kappa) \cdot \sL)f\rangle_\rho$ as a sum-of-squares (\cref{prop:bochner-identity}), taking heavy inspiration from \cite{GoebelEtAl2026}. The spectral-curvature conditions make
the terms involving the curvature matrices nonnegative, and the
remaining terms are squared second differences of $f$ with nonnegative
weights.

\subsection{Reduction to the Homogeneous Log-Concave Case}\label{subsubsec:outline-reduction}
The following special case of \Cref{thm:log-concave-convolution} (and \cref{thm:curvature-convolution}) is the most interesting case. It is also, in some sense, the most difficult.

\begin{theorem}[Simplified version of \cref{thm:homogeneous-convolution}]
\label{thm:homogeneous-convolution-informal}
Let $\mu:\binom{\cU}{a}\to\mathbb R_{\ge0}$ and
$\nu:\binom{\cU}{b}\to\mathbb R_{\ge0}$, extended by zero to $2^{\cU}$,
where $a,b\ge0$ are integers with $a+b=n\coloneqq |\cU|\ge1$.
Suppose $g_\mu,g_\nu$ are log-concave. Given value oracles for $\mu,\nu$,
the coefficient ranges $\cR_\mu,\cR_\nu$, and one set in the support of
each measure, there is an $\mathsf{FPRAS}$ for the subset-convolution value
\begin{align*}
 Z_\cap&\coloneqq (\mu\ast\nu)(\cU)=[\bm{x}^{\cU}]\bigl(g_\mu(\bm{x}) \cdot g_\nu(\bm{x})\bigr)
       =\sum_{S\subseteq\cU}\mu(S)\cdot\nu(\cU\setminus S).
\end{align*}
\end{theorem}

In this subsection, we describe how we reduce \cref{thm:curvature-convolution} (which generalizes \cref{thm:log-concave-convolution} as mentioned in \cref{rmk:curvature-LC-generalization}) to \cref{thm:homogeneous-convolution-informal} and \cref{thm:poincare}. A pictorial representation is provided in \cref{fig:convolution-reduction}, with detailed proofs deferred to \cref{app:convolution-proofs}. Throughout this subsection, we focus on the case of convolution where $\bm{\alpha} = \allone$; this is essentially without loss of generality, as one can simply pin the coordinates with exponent $0$ or $2$ for a general $\bm{\alpha} \in \{0,1,2\}^{\cU}$. We briefly explain our reductions as follows:
\begin{itemize}[leftmargin=1.5em,itemsep=2pt,topsep=3pt]
\item For \cref{thm:curvature-convolution} with $\gamma=\max\{\eta,\kappa\}>0$, we order the factors so that $\kappa=\gamma$ and apply a uniform external field $t=1-\epsilon/(16n)$ to $\mu$. This gives a spectral gap of at least $\Omega(\epsilon\gamma/n)$ in \cref{thm:poincare}, while keeping the new partition function $\epsilon$-multiplicatively close to the original. After normalizing the transition rates, the resulting discrete-time chain has gap $\Omega(\epsilon/(\chi n^2))$. We then apply the standard self-reduction~\cite{JVV1986}.
\item For \cref{thm:log-concave-convolution}, which includes
\cref{thm:curvature-convolution} with $\eta=\kappa=0$, we split by degree into at most $n+1$
homogeneous pairs and apply \cref{thm:homogeneous-convolution-informal}
to each.
\end{itemize}

\begin{center}
\begin{tikzpicture}[
 >=stealth',
 box/.style={draw=black!55, rounded corners=2pt, fill=black!2,
             text width=4.05cm, minimum height=9mm, inner sep=5pt,
             align=center, font=\small},
 flow/.style={->, draw=black!65, line width=.6pt},
 note/.style={font=\footnotesize, fill=white, inner sep=2pt,
              align=center}
]
 \node[box, text width=8.3cm] (input) at (0,0)
   {\Cref{thm:curvature-convolution}: subset convolution\\
    $Z_\cap=[\bm{x}^{\cU}](g_\mu g_\nu)$};
 \node[box, minimum height=10mm] (positive) at (-5,-1.45)
   {$\eta+\kappa>0$};
 \node[box, minimum height=10mm] (boundary) at (0,-1.45)
   {$\eta+\kappa=0$\\$\max\{\eta,\kappa\}>0$};
 \node[box, minimum height=10mm] (zero) at (5,-1.45)
   {$\eta=\kappa=0$};
 \node[box, text width=6.2cm, minimum height=11mm]
   (direct) at (-2.5,-3.9)
   {\Cref{thm:poincare}: Poincar\'e inequality\\
    Self-reduction~\cite{JVV1986}};
 \node[box, minimum height=11mm] (lc) at (5,-3.9)
   {\Cref{thm:log-concave-convolution}\\(log-concave case)};
 \node[box, minimum height=11mm, fill=NavyBlue!5,
       draw=NavyBlue!65] (homogeneous) at (5,-5.95)
   {\Cref{thm:homogeneous-convolution-informal}\\
    (homogeneous log-concave case)};
 \draw[draw=black!65, line width=.6pt] (input.south) -- (0,-.75);
 \draw[flow] (0,-.75) -| (positive.north);
 \draw[flow] (0,-.75) -- (boundary.north);
 \draw[flow] (0,-.75) -| (zero.north);
 \draw[draw=black!65, line width=.6pt]
   (positive.south) -- node[note, midway]
   {{Positive gap}} (-5,-3.15) -- (-2.5,-3.15);
 \draw[draw=black!65, line width=.6pt]
   (boundary.south) -- node[note, midway]
   {Small tilt creates\\a positive gap} (0,-3.15) -- (-2.5,-3.15);
 \draw[flow] (-2.5,-3.15) -- (direct.north);
 \draw[flow] (zero.south) -- node[note, midway]
   {Zero curvature implies\\log-concavity} (lc.north);
 \draw[flow] (lc.south) -- node[note, midway]
   {Split by degree} (homogeneous.north);
\end{tikzpicture}
\captionsetup{hypcap=false}

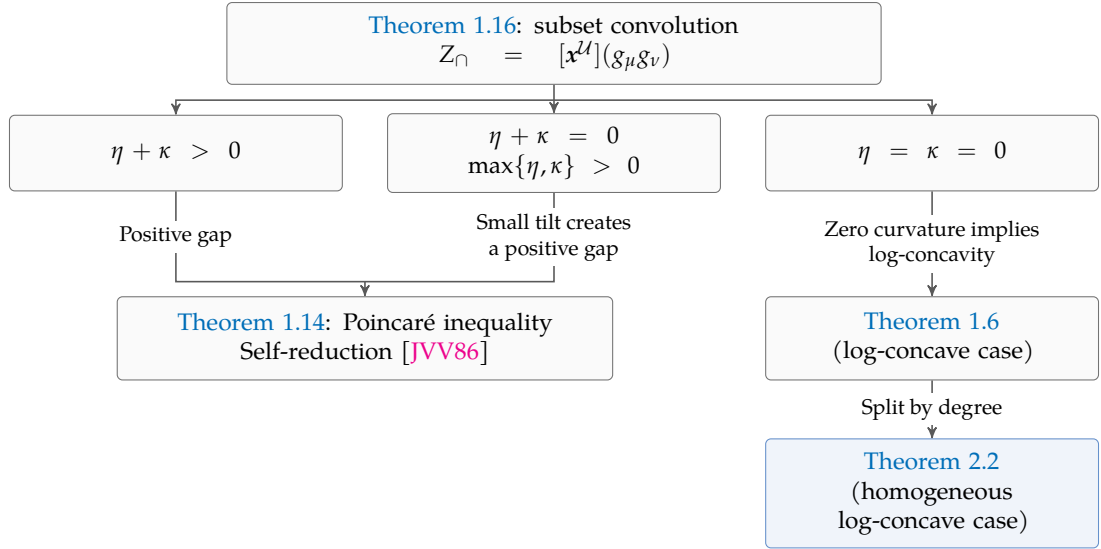
\captionof{figure}{The three curvature cases in the proof of
\cref{thm:curvature-convolution}. Arrows point toward the result used
by each reduction.}
\label{fig:convolution-reduction}
\end{center}

\subsection{From Poincar\'{e} to Counting via Convex Programming and Softening}\label{subsec:field-optimization}
In this section, we outline the proof of \cref{thm:homogeneous-convolution-informal}.
The algorithm and its proof of correctness provide a general framework to turn sampling from a ``soft'' intersection of two measures on sets into approximate counting for the corresponding ``hard'' intersection. This method may be of independent interest, and is based on the following general principle: A Poincar\'e inequality that
holds uniformly over external fields turns subset convolution into a
dual maximum entropy convex program. Variants of this principle based on stronger assumptions (e.g. log-concavity of the generating polynomial) with significantly weaker approximation guarantees have previously appeared in the literature \cite{Gurvits2006, AnariOveisGharan2017}. We discuss this in greater detail in \cref{rmk:gurvits-capacity}.

Let $\mu,\nu:2^{\cU}\to\R_{\ge0}$ be
nonzero, put $n\coloneqq|\cU|$, and consider
\begin{align*}
 Z_\cap&\coloneqq \sum_{S\subseteq\cU}\mu(S) \cdot \nu(\cU\setminus S).
\end{align*}
When $Z_\cap>0$, the direct target measure $\mu\odot\nu^*$, with weights
proportional to $\mu(S)\nu(\cU\setminus S)$, can be poorly connected under
natural local moves: the hard constraint $T=\cU\setminus S$ forces changes
to the two sets to be coordinated.
To bypass this, it is natural to consider a variant of $\mu \odot \nu^*$ in which the hard constraint $T = \cU \setminus S$ is ``softened''.

For penalty vectors $\bm{t}\in(0,1]^{\cU}$, write
$\bm{t}^R\coloneqq\prod_{u\in R}t_u$. For external fields
$\bm{z}\in\R_{>0}^{\cU}$, the \emph{soft intersection} is the distribution
on $\supp\mu\times\supp\nu$ given by
\begin{align*}
 \pi_{\bm{t},\bm{z}}(S,T)&\propto\mu(S) \cdot \nu(T) \cdot \bm{z}^S\bm{z}^T \cdot \bm{t}^{S\cap T},
\end{align*}
that is, the Hadamard product
\begin{align*}
 \pi_{\bm{t},\bm{z}}&=(\mu_{\bm{z}}\otimes\nu_{\bm{z}})\odot\psi_{\bm{t}},&
 \mu_{\bm{z}}(S)&\coloneqq \mu(S) \cdot \bm{z}^S,\quad \nu_{\bm{z}}(T)\coloneqq \nu(T) \cdot \bm{z}^T,\quad
 \psi_{\bm{t}}(S,T)\coloneqq \bm{t}^{S\cap T}.
\end{align*}
\begin{samepage}
Let
\begin{align*}
    Z_{\bm{t},\bm{z}} \coloneqq \sum_{S,T \subseteq \cU} \mu(S) \cdot \nu(T) \cdot \bm{z}^S\bm{z}^T \cdot \bm{t}^{S\cap T}
\end{align*}
be its partition function.
\end{samepage}
When all penalties equal a scalar $t\in(0,1]$, we write a subscript
$t$ for the uniform penalty vector $\bm{t}=t \cdot \mathbf1$.


Note that $\pi_{\bm{t},\bm{z}}$ penalizes pairs $(S,T)$ where $S,T$ have a large intersection. On the other hand, complementary pairs $T=\cU\setminus S$ carry no penalty. Moreover, all such complementary pairs have the common factor $\bm{z}^{\cU}$, so
\begin{align*}
    Z_\cap&=\frac{Z_{\bm{t},\bm{z}}}{\bm{z}^{\cU}}\cdot\pi_{\bm{t},\bm{z}}(T=\cU\setminus S), \qquad \forall \bm{z} \in \R_{>0}^{\cU}.
\end{align*}

{Note that in our application, since $\mu$ and $\nu$ are log-concave measures, so is $\mu\otimes \nu$. Hence, the down-up walk (a.k.a. the basis exchange walk) on $\mu\otimes \nu$ is rapidly mixing~\cite{AnariEtAl2019}. $\mu \otimes \nu$ corresponds to the case $t = 1$, and we can hope that when $t$ is not too close to $0$, local Markov chains for sampling from $\pi_{t,\bm{z}}$ remain rapidly mixing.
In the later part of the proof, we will formally prove this by an application of \cref{thm:poincare} and show that there is an efficiently implementable Markov chain that samples from $\pi_{t,\bm{z}}$ in polynomial time provided $t \gtrsim 1/\mathsf{poly}(n)$.
}
%
This also allows us to approximate $Z_{t,\bm{z}}$ by standard counting-to-sampling reductions \cite{JVV1986}.

{
All that remains is to estimate $\pi_{t,\bm{z}}(T = \cU \setminus S)$ with good relative error, since putting it together with an accurate estimate of $Z_{t,\bm z}$ yields an accurate estimate of $Z_{\cap}$. The na{\"i}ve Monte Carlo method can work well as long as $\pi_{t,\bm{z}}(T = \cU \setminus S) \gtrsim 1/\mathsf{poly}(n)$.
The problem is that choosing $t$ to be as small as $1/\mathsf{poly}(n)$ does not automatically make $\pi_{t,\bm{z}}(T = \cU \setminus S)$ large.
See \Cref{rmk:penalty-alone} for an example.
}

{
\begin{remark}\label{rmk:penalty-alone}
An inverse-polynomial $t$ need not give complementary pairs
inverse-polynomial probability when $\bm{z}=\mathbf1$. Consider the
following ladder of Jerrum, Sinclair, and
Vigoda~\cite[Figure~2]{JerrumSinclairVigoda2004}.
In the upper diagrams, blue $S$ and orange $B=\cU\setminus T$ are bases of the partition matroids on blue and orange vertices, respectively. In the lower diagrams, filled and hollow vertices indicate matched and unmatched vertices.
Shared edges are drawn twice above and once below.
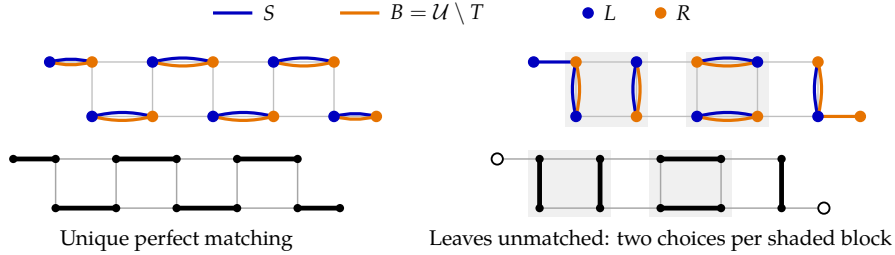
\begin{center}
\begin{tabular}{@{}c@{}}
\begin{tikzpicture}[x=.8cm,y=.72cm,
 edge/.style={draw=black!25,line width=.5pt},
 leftbasis/.style={draw=blue!75!black,line width=1.2pt},
 rightbasis/.style={draw=orange!90!black,line width=1.2pt},
 leftvertex/.style={circle,draw=blue!75!black,fill=blue!75!black,inner sep=1.35pt,line width=.5pt},
 rightvertex/.style={circle,draw=orange!90!black,fill=orange!90!black,inner sep=1.35pt,line width=.5pt},
 every node/.style={font=\footnotesize}]
 \draw[leftbasis] (2,1.9)--(2.65,1.9) node[right,text=black] {$S$};
 \draw[rightbasis] (4.1,1.9)--(4.75,1.9)
   node[right,text=black] {$B=\cU\setminus T$};
 \node[leftvertex,label=right:{$L$}] at (8.2,1.9) {};
 \node[rightvertex,label=right:{$R$}] at (9.4,1.9) {};
 \foreach \panel in {0,1} {
  \begin{scope}[xshift=\panel*6.4cm]
   \ifnum\panel=1
    \fill[black!6] (-.18,-.18) rectangle (1.18,1.18);
    \fill[black!6] (1.82,-.18) rectangle (3.18,1.18);
   \fi
   \draw[edge] (-.7,1)--(0,1)--(4,1) (0,0)--(4,0)--(4.7,0);
   \foreach \i in {0,...,4} {\draw[edge] (\i,0)--(\i,1);}
   \ifnum\panel=0
    \foreach \xa/\ya/\xb/\yb in {
      -.7/1/0/1,4/0/4.7/0,0/0/1/0,2/0/3/0,1/1/2/1,3/1/4/1} {
     \draw[leftbasis] (\xa,\ya) to[bend left=13] (\xb,\yb);
     \draw[rightbasis] (\xa,\ya) to[bend right=13] (\xb,\yb);
    }
   \else
    \foreach \xa/\ya/\xb/\yb in {
      0/0/0/1,1/0/1/1,2/0/3/0,2/1/3/1,4/0/4/1} {
     \draw[leftbasis] (\xa,\ya) to[bend left=13] (\xb,\yb);
     \draw[rightbasis] (\xa,\ya) to[bend right=13] (\xb,\yb);
    }
    \draw[leftbasis] (-.7,1)--(0,1);
    \draw[rightbasis] (4,0)--(4.7,0);
   \fi
   \foreach \i in {0,...,4} {
    \ifodd\i
     \node[leftvertex] at (\i,1) {};
     \node[rightvertex] at (\i,0) {};
    \else
     \node[leftvertex] at (\i,0) {};
     \node[rightvertex] at (\i,1) {};
    \fi
   }
   \node[leftvertex] at (-.7,1) {};
   \node[rightvertex] at (4.7,0) {};
  \end{scope}
 }
\end{tikzpicture}
\\[1.2ex]
\begin{tikzpicture}[x=.8cm,y=.65cm,
 edge/.style={draw=black!35,line width=.55pt},
 matched/.style={draw=black,line width=1.8pt},
 vertex/.style={circle,fill=black,inner sep=1.1pt},
 freevertex/.style={circle,draw=black,fill=white,inner sep=1.5pt,line width=.7pt},
 every node/.style={font=\footnotesize}]
 \foreach \panel in {0,1} {
  \begin{scope}[xshift=\panel*6.4cm]
   \ifnum\panel=1
    \fill[black!6] (-.18,-.18) rectangle (1.18,1.18);
    \fill[black!6] (1.82,-.18) rectangle (3.18,1.18);
   \fi
   \draw[edge] (-.7,1)--(0,1)--(4,1) (0,0)--(4,0)--(4.7,0);
   \foreach \i in {0,...,4} {\draw[edge] (\i,0)--(\i,1);}
   \ifnum\panel=0
    \draw[matched] (-.7,1)--(0,1) (4,0)--(4.7,0)
      (0,0)--(1,0) (2,0)--(3,0) (1,1)--(2,1) (3,1)--(4,1);
   \else
    \draw[matched] (0,0)--(0,1) (1,0)--(1,1)
      (2,0)--(3,0) (2,1)--(3,1) (4,0)--(4,1);
   \fi
   \foreach \i in {0,...,4} {
    \node[vertex] at (\i,0) {}; \node[vertex] at (\i,1) {};
   }
   \ifnum\panel=0
    \node[vertex] at (-.7,1) {}; \node[vertex] at (4.7,0) {};
    \node[align=center] at (2,-.65) {Unique perfect matching};
   \else
    \node[freevertex] at (-.7,1) {}; \node[freevertex] at (4.7,0) {};
    \node[align=center] at (2,-.65) {Leaves unmatched: two choices per shaded block};
   \fi
  \end{scope}
 }
\end{tikzpicture}
\end{tabular}
\end{center}
The pendant edges force a unique perfect matching $M$. Leaving both
leaves unmatched gives exponentially many near-perfect matchings $C$,
with two choices per shaded block and the final rung fixed. Completing $C$ as shown gives a
distinct soft configuration with $|S\cap T|=|S\setminus B|=1$,
hence only one factor of $t$. Their total weight makes the perfect matching exponentially unlikely.

\end{remark}
}

{
Since $Z_\cap$ isn't affected by the choice of $\bm z$ and $t$,
the hope is that we can still make $\pi_{t,\bm{z}}(T = \cU \setminus S) \gtrsim 1/\mathsf{poly}(n)$ by finding a good $\bm z \in \R_{>0}^{\cU}$.
To find such a $\bm z$, it is natural to consider the following dual maximum entropy convex program.
\begin{definition} \label{def:Phi-program}
Fix $t\in(0,1]$. Write the external field as $\bm{z}=\e^{\bm{y}}$,
where $\bm{y}\in\R^{\cU}$ gives its logarithmic coordinates, and define
\begin{align}\label{eq:field-gradient}
 \Phi(\bm{y})&\coloneqq \log Z_{t,\bm{z}}-\sum_{u\in\cU}y_u,&
 \partial_u\Phi(\bm{y})&=\pi_{t,\bm{z}}(u\in S)-\pi_{t,\bm{z}}(u\notin T).
\end{align}
We consider the following convex program:
\begin{align} \label{eq:Phi-program}
    \inf_{\bm{y}\in\R^{\cU}}\Phi(\bm{y}).
\end{align}
\end{definition}
}
The quantity $\exp(\inf_{\bm{y}\in\R^{\cU}}\Phi(\bm{y}))=\operatorname{Cap}_{\mathbf1}(Z_{t,\bm{x}})$ is the Gurvits capacity of $Z_{t,\bm{x}}$, viewed as a polynomial in $\bm{x}$; see \cref{rmk:gurvits-capacity} for the definition and its connection to earlier work.
{
The potential $\Phi$ is convex because its Hessian is the covariance
matrix of $\mathbf1_S+\mathbf1_T-\mathbf1$ for $(S,T)\sim\pi_{t,\e^{\bm{y}}}$.
When $Z_\cap>0$, it holds that
\begin{align*}
 \pi_{t,\e^{\bm{y}}}(T=\cU\setminus S)
 &=Z_\cap \cdot \exp\tp{-\Phi(\bm{y})}.
\end{align*}
Hence, minimizing $\Phi$ maximizes the probability $\pi_{t,\e^{\bm{y}}}(T=\cU\setminus S)$.

The next lemma shows that, under suitable Poincaré inequality assumptions, any near-stationary point $\bm{y}$ with small $\norm{\nabla\Phi(\bm y)}_\infty$ makes $\pi_{t,\e^{\bm{y}}}(T=\cU\setminus S)$ large when $t$ is sufficiently small. By \cref{eq:field-gradient}, the gradient of $\Phi$ can be estimated by sampling. Thus, we can find such a near stationary point by using standard stochastic gradient descent algorithms.

}


\begin{restatable}{lemma}{fieldconcavitylemma}\label{lem:field-concavity}
Let $\gamma>0$. Suppose that for every $\bm{t}\in(0,1]^{\cU}$ and
$\bm{z}\in\R_{>0}^{\cU}$, there is an irreducible continuous-time Markov chain
reversible with respect to $\pi_{\bm{t},\bm{z}}$, with spectral gap at least $\gamma$.
Each move updates only one of the two sets $S,T$.
Whenever $u\in S\oplus T$, the total rate of moves removing
its unique copy is at most $1$. Then for every scalar $t\in(0,1]$ and every
$\bm{z}\in\R_{>0}^{\cU}$, with $\bm{y}=\log\bm{z}$,
\begin{align}\label{eq:field-certificate}
 \pi_{t,\bm{z}}(T\ne\cU\setminus S)
 &\le n\left(\|\nabla\Phi(\bm{y})\|_\infty+\frac{4\sqrt t}{\gamma}\right).
\end{align}
\end{restatable}
{The proof is deferred to \cref{subsec:fpras-fields}.}

To summarize, we choose a small scalar penalty $t \asymp 1/\mathsf{poly}(n)$ and implement stochastic gradient descent on $\Phi$, seeking logarithmic coordinates
$\bm{y}$ with small $\nabla\Phi(\bm{y})$ and setting $\bm{z}=\e^{\bm{y}}$.
Note that by \cref{eq:field-gradient}, each gradient coordinate can
be efficiently estimated as the difference of two empirical frequencies
from random samples of $\pi_{t,\bm{z}}$.
With an inverse-polynomial lower bound on $\gamma$, taking $\|\nabla\Phi(\bm{y})\|_\infty\le1/\operatorname{poly}(n)$ sufficiently small
makes complementary pairs occur with constant probability, by
\cref{lem:field-concavity}. We can therefore estimate
this probability $\pi_{t,\bm{z}}(T=\cU\setminus S)$ from polynomially many samples as desired.

\paragraph{Application to the Homogeneous Log-Concave Case}
We apply the previous framework under the
hypotheses of \cref{thm:homogeneous-convolution-informal}, with degrees
$a,b$ and $a+b=n$. Here $\pi_{t,\bm{z}}$ is supported on
$\binom{\cU}{a}\times\binom{\cU}{b}$, and we call it the
\emph{bases pair law}; see \cref{def:pair-law}. The key input we need to verify is the Poincar\'{e} Inequality for $\pi_{t,\bm{z}}$ required by \cref{lem:field-concavity}. We describe here the argument at a high level; the formal argument is provided in \cref{subsec:fpras-activity-limit}.

We first extend the measures to downward-closed set families.
{Define the weights $\widehat\mu_{\bm{z}},\widehat\nu_{\bm{z}}$ through
their generating polynomials by}
\begin{align*}
 g_{\widehat\mu_{\bm{z}}}(\bm{x})&\coloneqq g_\mu(\mathbf1+\bm{z}\bm{x}),&
 g_{\widehat\nu_{\bm{z}}}(\bm{x})&\coloneqq g_\nu(\mathbf1+\bm{z}\bm{x}),
 \qquad \bm{z}\bm{x}\coloneqq (z_ux_u)_{u\in\cU}.
\end{align*}
The \emph{independent-set pair law}
on $2^{\cU}\times2^{\cU}$ (\cref{def:independent-pair}) is the
normalized Hadamard product
\begin{align*}
 \widehat\pi_{t,\bm{z}}
 &\coloneqq (\widehat\mu_{\bm{z}}\otimes\widehat\nu_{\bm{z}})\odot\psi_t.
\end{align*}
Conditioning $\widehat\pi_{t,\bm{z}}$ on
$\binom{\cU}{a}\times\binom{\cU}{b}$ gives exactly $\pi_{t,\bm{z}}$, and
$\widehat\pi_{t,\lambda \bm{z}}$ converges to $\pi_{t,\bm{z}}$ as
$\lambda\to\infty$.
The first factor has $0$-spectral curvature, and the dual of the second
factor has $1$-spectral curvature
(\cref{obs:w-0-curvature,obs:psi-star-1-curvature}).
Thus \cref{thm:poincare} gives
a Poincar\'e bound with constant $1$ for $\widehat\pi_{t,\lambda \bm{z}}$,
uniformly in $\lambda$. Taking $\lambda\to\infty$ transfers this bound
to an exchange chain for $\pi_{t,\bm{z}}$ (\cref{lem:bases-pair-poincare}),
yielding an efficient sampler using only the original value oracles.

\paragraph{Connections and Additional Intuition}

{
\begin{remark}[Partition-matroid interpretation]\label{rmk:partition-penalty}
The penalty $\psi_{\bm{t}}$ can be understood as an analogue of the
random-cluster model, where there is a soft penalty for non-independent
sets. Let $\mathcal P$ be the partition
matroid on $\cU\times\{1,2\}$ allowing at most one copy of each $u$,
and set $X\coloneqq(S\times\{1\})\cup(T\times\{2\})$.
For uniform penalties $t_u=t\in(0,1]$,
\[
 \psi_{t\allone}(S,T)=t^{|S\cap T|}
 =t^{|X|-{\rk_{\mathcal P}}(X)}.
\]
\end{remark}
}

{
\begin{remark}[Gurvits capacity]\label{rmk:gurvits-capacity}
For a nonzero finitely supported measure $w$ on $\mathbb Z_{\ge0}^n$,
write $g_w(\bm x)=\sum_{\bm k}w(\bm k)\bm x^{\bm k}$.
Its \emph{Gurvits capacity} at $\bm\alpha\in\R_{\ge0}^n$ is
\[
 \operatorname{Cap}_{\bm\alpha}(g_w)\coloneqq
 \inf_{\bm x\in\R_{>0}^n}\frac{g_w(\bm x)}{\bm x^{\bm\alpha}}.
\]
The Gurvits capacity at $\bm \alpha \in \supp w$ is always an upper bound on $w(\bm \alpha)$.
To see this, define the distribution
$p_{\bm z}(\bm k)=w(\bm k)\bm z^{\bm k}/g_w(\bm z)$.
If $\bm\alpha\in\supp w$, then
\[
 \operatorname{Cap}_{\bm\alpha}(g_w)
 =\frac{w(\bm\alpha)}{\sup_{\bm z>0}p_{\bm z}(\bm\alpha)} \geq w(\bm \alpha).
\]
Intuitively, it gives the best possible approximation of $w(\bm \alpha)$ that we can get by only tuning external fields applied to the base measure.

After taking logarithms, the optimization over external fields above has
a dual formulation as a minimum $\KL$-divergence convex program with prescribed mean,
or equivalently, a weighted maximum-entropy convex program~\cite{SinghVishnoi2014Entropy}.
For $\bm\alpha\in\supp w$, let
\[
 q^*=\operatorname*{arg\,min}_{q:\,\mathbb E_q[\bm X]=\bm\alpha}
 D_{\mathrm{KL}}\left(q\,\middle\|\,p_{\allone}\right),
\]
where $q$ ranges over probability distributions on $\supp w$,
$p_{\allone}=w/g_w(\mathbf1)$ is the normalized measure, and
$D_{\mathrm{KL}}$ denotes Kullback--Leibler divergence.
Entropy duality and the KL Pythagorean identity~\cite[Theorem~3.2]{CsiszarShields2004}
give
\[
 \log\frac{\operatorname{Cap}_{\bm\alpha}(g_w)}{w(\bm\alpha)}
 =D_{\mathrm{KL}}\left(\delta_{\bm\alpha}\,\middle\|\,q^*\right)
 =-\log q^*(\bm\alpha),
\]
where $\delta_{\bm\alpha}$ is the point mass at $\bm\alpha$.
In other words, we adjust the original distribution as little as possible
in KL divergence to obtain mean $\bm\alpha$. The approximation ratio is
the reciprocal of the probability that this distribution assigns to
$\bm\alpha$ itself.

For every homogeneous \emph{real stable} polynomial $g$ of degree $n$ in $n$
variables with nonnegative coefficients,
\[
 \frac{n!}{n^n}\operatorname{Cap}_{\mathbf1}(g)
 \le [x_1\cdots x_n]g
 \le \operatorname{Cap}_{\mathbf1}(g).
\]
The upper bound follows from the definition of capacity.
Gurvits~\cite{Gurvits2006} proves the lower bound, showing that capacity
overestimates the coefficient by at most a factor $n^n/n!$.
For the permanent of a nonnegative matrix $A$, he uses the polynomial
$g(\bm{x})=\prod_{i=1}^n(\sum_{j=1}^n A_{ij}x_j)$;
for mixed discriminants, he uses
$g(\bm{x})=\det(\sum_{i=1}^n x_iA_i)$, where the $A_i$ are positive semidefinite.
This approach has since been extended to other settings; see
\cref{subsec:prior-work}.

In our setting, the target coefficient
$[\bm{x}^{\cU}]Z_{t,\bm{x}}=Z_\cap$ is independent of $t$.
For fixed $t$, the Gurvits capacity is
\[
 \operatorname{Cap}_{\mathbf1}(Z_{t,\bm{x}})=
 \inf_{\bm{x}\in\R_{>0}^{\cU}}\frac{Z_{t,\bm{x}}}{\bm{x}^{\cU}}
 =\exp\left(\inf_{\bm{y}\in\R^{\cU}}\Phi(\bm{y})\right).
\]
When $Z_\cap>0$, the probability interpretation above gives the
approximation ratio
\[
 \frac{\operatorname{Cap}_{\mathbf1}(Z_{t,\bm x})}{Z_\cap}
 =\frac{1}{\sup_{\bm z>0}\pi_{t,\bm z}(T=\cU\setminus S)}.
\]
Under the hypotheses of \cref{lem:field-concavity}, this ratio satisfies
\[
 1\le\frac{\operatorname{Cap}_{\mathbf1}(Z_{t,\bm{x}})}{Z_\cap}
 \le\left(1-\frac{4n\sqrt t}{\gamma}\right)^{-1},
 \qquad 4n\sqrt t<\gamma.
\]
Indeed, $\Phi$ is bounded below and has a globally bounded Hessian,
so it admits points $\bm{y}$ with arbitrarily small gradient; applying
\cref{lem:field-concavity} to the corresponding fields $\bm{z}=\e^{\bm{y}}$
gives the upper bound on the ratio.
Thus decreasing $t$ makes the capacity bound nearly tight.
\end{remark}
}

{
\begin{remark}[Comparison with permanent algorithms]\label{rmk:permanent-annealing}
\looseness=1
The permanent algorithm of Jerrum--Sinclair--Vigoda~\cite{JerrumSinclairVigoda2004}
learns correction weights along an annealing schedule to balance the total
weight of perfect matchings with that of each near-perfect class, indexed
by its pair of unmatched vertices.
Our framework chooses external fields directly by stochastic gradient
descent. It can serve as a black box for a broad class of counting
problems admitting the above formulation, once the sampling and uniform
Poincar\'e hypotheses are verified.
\end{remark}
}



\section{Preliminaries}\label{sec:preliminaries}


{
\subsection{Notation}\label{subsec:notation}
For a finite ground set $\cU$, write $n\coloneqq|\cU|$ and
$\binom{\cU}{r}\coloneqq\{S\subseteq\cU:|S|=r\}$.
We use $S+u\coloneqq S\cup\{u\}$, $S-u\coloneqq S\setminus\{u\}$,
and $S\oplus T\coloneqq(S\setminus T)\cup(T\setminus S)$.
The symbol $\allone[\cdot]$ denotes an indicator, $\mathbf1_S$ the
indicator vector of $S$, and $\mathbf1$ the all-ones vector.

Vectors are written in bold. Products $\bm{z}\bm{x}$, exponentials
$\e^{\bm{y}}$, and logarithms $\log\bm{z}$ are taken coordinatewise. We write
$\bm{z}^S\coloneqq\prod_{u\in S}z_u$ for $S \subseteq \cU$, and
$\bm{z}^{\bm{\alpha}}\coloneqq\prod_{u\in\cU}z_u^{\alpha_u}$ for $\bm{\alpha} \in \Z_{\geq0}^{\cU}$.
For a polynomial $p$, $[\bm{z}^{\bm{\alpha}}]p$ denotes the coefficient
of $\bm{z}^{\bm{\alpha}}$ in $p$.
For symmetric matrices, $A\succeq B$ means that $A-B$ is positive semidefinite.

A nonnegative set measure $\mu:2^{\cU}\to\R_{\ge0}$ need not be
normalized. Its support is
$\supp\mu\coloneqq\{S\subseteq\cU:\mu(S)>0\}$.
An \emph{external field} is a vector $\bm{z}\in\R_{>0}^{\cU}$.
Applying it to $\mu$ gives the unnormalized measure
\[
 \mu_{\bm{z}}(S)\coloneqq\mu(S)\bm{z}^S.
\]
When $\mu$ is nonzero, the corresponding probability distribution is
$\mu_{\bm{z}}/g_\mu(\bm{z})$.
We write $\bm{y}=\log\bm{z}$ for the logarithmic coordinates of the
external field, so $\bm{z}=\e^{\bm{y}}$.
The product weights $(\mu\otimes\nu)(S,T)=\mu(S)\nu(T)$ are also unnormalized;
$\mu\odot\nu$ denotes the normalized Hadamard product of
\cref{def:hadamard}.

Logarithms are natural unless a base is specified.
$\widetilde{O}(\cdot)$ hides polylogarithmic factors.

\subsection{Matroids}\label{subsec:matroids}
A \emph{matroid} $\cM$ on $\cU$ has a family
$\cI(\cM)\subseteq2^{\cU}$ of \emph{independent sets} satisfying:
\begin{itemize}[leftmargin=1.5em,itemsep=2pt,topsep=3pt]
\item $\varnothing\in\cI(\cM)$.
\item If $I\in\cI(\cM)$ and $J\subseteq I$, then $J\in\cI(\cM)$.
\item If $I,J\in\cI(\cM)$ and $|I|<|J|$, then
$I\cup\{e\}\in\cI(\cM)$ for some $e\in J\setminus I$.
\end{itemize}
A \emph{basis} is a maximal independent set; write $\cB(\cM)$ for
the family of bases. The \emph{rank} of $S\subseteq\cU$ is
\[
 \rk_{\cM}(S)\coloneqq\max\{|I|:I\subseteq S,\ I\in\cI(\cM)\}.
\]
All bases have size $\rk_{\cM}(\cU)$.
An \emph{independence oracle} tests membership in $\cI(\cM)$.
Greedily adding elements of $S$ while preserving independence finds
a largest independent subset of $S$, and hence its rank, with at most
$|S|$ oracle queries.
The \emph{dual matroid} $\cM^*$ is the matroid on $\cU$ whose
independent sets are the sets $I\subseteq\cU$ such that
$\cU\setminus I$ contains a basis of $\cM$.
Equivalently, $\cB(\cM^*)=\{\cU\setminus B:B\in\cB(\cM)\}$;
see \cite{Oxley2011} for further background.

The basis generating polynomial
$g_{\cM}(\bm{x})\coloneqq\sum_{B\in\cB(\cM)}\bm{x}^B$
is log-concave~\cite{AnariOveisGharanVinzant2018,BrandenHuh2020}.
Additional matroid terminology and operations used in the applications
are recalled in \cref{app:application-background}.
}

\subsection{Markov Chains and Mixing Time}\label{subsec:markov-chains}
%
Let $\Omega$ be a finite state space and $\pi$ a probability law on
$\Omega$ with $\pi(x)>0$ for every $x\in\Omega$. For
$f,g:\Omega\to\mathbb R$, define
\begin{align*}
 \langle f,g\rangle_\pi&\coloneqq \sum_{x\in\Omega}\pi(x)f(x)g(x),&
 \|f\|_\pi^2&\coloneqq \langle f,f\rangle_\pi,&
 \Var[\pi]{f}&\coloneqq \|f-\E[\pi]{f}\|_\pi^2.
\end{align*}

Let $r(x,y)\ge0$, $x\ne y$, be transition rates on $\Omega$.
These define a continuous-time Markov chain $(X_t)_{t\ge0}$:
assign each directed edge $x\to y$ an independent Poisson clock of
rate $r(x,y)$. When its clock rings, the chain moves to $y$ if it is
currently at $x$; otherwise, nothing happens.
We assume the rates satisfy the \emph{detailed balance equations}
\begin{align}\label{eq:detailed-balance}
 \pi(x)\,r(x,y)&=\pi(y)\,r(y,x)\qquad\text{for all }x\ne y\in\Omega.
\end{align}
We then say that the chain is \emph{reversible} with respect to $\pi$.
All Markov chains used in this paper are finite and reversible.
The positive Laplacian and Dirichlet form associated with these rates
are
\begin{align}
 (\cL_r f)(x)&\coloneqq \sum_{y\ne x}r(x,y) \cdot \bigl(f(x)-f(y)\bigr),
 \label{eq:rate-laplacian}\\
 \cE_r(f,g)&\coloneqq \langle f,\cL_r g\rangle_\pi
 =\frac12\sum_{\substack{x,y\in\Omega\\x\ne y}}
 \pi(x)r(x,y) \cdot \bigl(f(y)-f(x)\bigr)\cdot\bigl(g(y)-g(x)\bigr).
 \label{eq:dirichlet-form}
\end{align}
The second equality in \cref{eq:dirichlet-form} follows from detailed balance.
The chain has generator $-\cL_r$, so its transition matrix
$P_t(x,y)\coloneqq\Pr{X_t=y\mid X_0=x}$ satisfies
\begin{align*}
 P_t&=\exp(-t\cL_r),\qquad t\ge0.
\end{align*}
We omit the subscript of $\cL_r$ when the rates are clear, and note that $\cE_{c \cdot r}(\cdot,\cdot)=c \cdot \cE_r(\cdot,\cdot)$
for $c>0$. We use $\cL$ for a general Laplacian and reserve $\sL$ for
{the Laplacian defined in \cref{eq:full-laplacian}}.

\paragraph{Spectral gap and Poincar\'e inequalities.}
The chain is \emph{irreducible} if, for every $x,y\in\Omega$, there is
a sequence $x=x_0,\ldots,x_k=y$ such that
$r(x_{j-1},x_j)>0$ for $1\le j\le k$.
In this case, its \emph{spectral gap} is
\begin{align}\label{eq:spectral-gap}
 \gamma(\cL_r)&\coloneqq \inf_{f\text{ nonconstant}}
 \frac{\cE_r(f,f)}{\Var[\pi]{f}},
\end{align}
the smallest nonzero eigenvalue of $\cL_r$; see
\cite[Lemma~13.7]{LevinPeres2017}. Thus $\gamma(\cL_r)\ge\gamma$ if and
only if the following \emph{Poincar\'e inequality} holds:
\begin{align}\label{eq:poincare-rate-form}
 \gamma \cdot \Var[\pi]{f}&\le\cE_r(f,f)
 \qquad\text{for every }f:\Omega\to\mathbb R.
\end{align}

\paragraph{Mixing time.}
The continuous-time chain $P_t$ has a lazy discrete-time implementation:
choose $R>0$ with $R\ge\max_x\sum_{y\ne x}r(x,y)$ and set
$P\coloneqq I-\cL_r/(2R)$.
Applying $P$ at the rings of an independent Poisson clock of rate $2R$
gives the chain $P_t$. For sampling from $\pi$, we run $P$ for a fixed
number of steps, bounded as follows.
For probability laws $\rho,\nu$ on $\Omega$, their \emph{total variation
distance} is
\[\|\rho-\nu\|_{\mathrm{TV}}\coloneqq
\frac12\sum_{x\in\Omega}|\rho(x)-\nu(x)|.\]
For $0<\eps<1$, the \emph{mixing time} of $P$ is
\begin{align*}
 t_{\mathrm{mix}}(\eps)&\coloneqq
 \min\left\{k\in\mathbb Z_{\ge0}:
 \max_{x\in\Omega}\|P^k(x,\cdot)-\pi\|_{\mathrm{TV}}\le\eps\right\}.
\end{align*}
Write $\pi_{\min}\coloneqq\min_{x\in\Omega}\pi(x)$.
Under \cref{eq:poincare-rate-form} with $\gamma>0$, the standard
spectral mixing bound \cite[Theorem~12.4]{LevinPeres2017} gives
\begin{align}\label{eq:poincare-mixing}
 t_{\mathrm{mix}}(\eps)
 &\le\left\lceil\frac{2R}{\gamma}
       \log\frac{1}{\eps\pi_{\min}}\right\rceil.
\end{align}
Running $P$ for this many steps from any starting state gives a sample
within total variation distance $\eps$ of $\pi$. Thus an efficient
implementation of each step gives an efficient sampler whenever this
mixing-time bound is polynomial.

\begin{proof}[Proof of \cref{prop:bochner-criterion}]
Since the chain is reversible, $\cL$ is self-adjoint and positive
semidefinite with respect to $\langle\cdot,\cdot\rangle_\pi$; since it
is irreducible, its kernel consists of the constant functions. Hence
$\cL$ is positive definite on $\mathbf1^\perp$
{(the functions with mean zero under $\pi$)}, and $\cL^{1/2}$ is a
bijection of $\mathbf1^\perp$. Both ratios are unchanged when a constant
is added to $f$, so we may take $f\in\mathbf1^\perp\setminus\{0\}$, for
which $\Var[\pi]{f}=\langle f,f\rangle_\pi$. Substituting
$g\coloneqq\cL^{1/2}f$ gives
\begin{align*}
 \frac{\langle\cL f,\cL f\rangle_\pi}{\langle f,\cL f\rangle_\pi}
 &=\frac{\langle g,\cL g\rangle_\pi}{\langle g,g\rangle_\pi},
\end{align*}
and $g$ ranges over $\mathbf1^\perp\setminus\{0\}$ as $f$ does. Hence
the two infima coincide.
\end{proof}

\subsection{Sampling to Counting}\label{subsec:sampling-to-counting}
We record below the standard reduction from counting to sampling due to
Jerrum, Valiant, and Vazirani~\cite{JVV1986}.

\begin{lemma}[{\cite{JVV1986}}]\label{lem:sampling-to-counting}
Let $m\ge1$, let $w:\{0,1\}^m\to\R_{\ge0}$ be nonzero, and write
$Z\coloneqq\sum_{\sigma}w(\sigma)$ and $\pi(\sigma)\coloneqq w(\sigma)/Z$.
For every nonempty conditional law obtained by fixing coordinates,
suppose a sampler, given any starting state in its support and an
accuracy $0<\tau<1$, returns
a sample in that support whose distribution is within total variation distance $\tau$ of
the conditional law, using fresh independent randomness.
Given an evaluation oracle for $w$, a state $\sigma_0\in\supp w$, and
$0<\eps,\delta<1$, one can compute $\widehat Z$ satisfying
\[
 \Pr{(1-\eps)Z\le\widehat Z\le(1+\eps)Z}\ge1-\delta.
\]
On every execution, the algorithm uses
$Q=O(m^3\eps^{-2}\log(2m/\delta))$ sampler calls, each with total variation distance at most
$\tau=\eps/(16m)$, one evaluation of $w$, and $O(Q)$ additional
arithmetic operations. It supplies a starting state for every call.
\end{lemma}

\subsection{Stochastic Coordinate-Descent}
Stochastic coordinate descent finds a point with small gradient using
noisy gradient estimates: it averages fresh estimates and steps along
any coordinate whose averaged estimate is large. The analysis combines
the coordinate-descent estimate of Nutini et
al.~\cite[Section~3, equation~(2)]{NutiniEtAl2015} with concentration
bounds for the noise. We state it with two guarantees that are
convenient when the gradient oracle is implemented by sampling: every
iterate stays in a bounded region, and the cost is bounded on every
execution.

\begin{theorem}\label{thm:stochastic-coordinate-descent}
Fix $d\ge1$, $\bm{x}_0\in\mathbb R^d$, $L,\sigma>0$, and $H_0\ge0$.
Assume the following.
\begingroup
\setlist[enumerate,1]{label=(\roman*),wide=0pt}
\begin{enumerate}
\item \textbf{Function.} The function $f:\mathbb R^d\to\mathbb R$
is convex and differentiable, and for all $\bm{x}\in\mathbb R^d$,
$u\in\mathbb R$, and $1\le i\le d$,
\begin{align}
 f(\bm{x}+u\bm{e}_i)&\le f(\bm{x})+u\partial_i f(\bm{x})+Lu^2/2,\label{eq:coordinate-smoothness}\\
 f(\bm{x})&\ge f(\bm{x}_0)-H_0.\label{eq:coordinate-lower-bound}
\end{align}
Here $(\bm{e}_i)_j\coloneqq \allone_{\{i=j\}}$.
\item \textbf{Gradient oracle.} Each query at $\bm{x}$ uses fresh independent
randomness and costs at most $T\ge d$. Its output $\bm{G}(\bm{x})$ is unbiased
with coordinatewise $\sigma$-sub-Gaussian noise: for all $\bm{x},u,i$ as above,
\begin{align*}
 \E{\bm{G}(\bm{x})}&=\nabla f(\bm{x}),\\
 \E{\exp\bigl(u(G_i(\bm{x})-\partial_i f(\bm{x}))\bigr)}
 &\le\exp(\sigma^2u^2/2).
\end{align*}
\end{enumerate}
For any $0<\eta\le1$ and $0<\delta<1$, a stochastic coordinate-descent
algorithm has the following guarantees.
\setlist[enumerate,1]{label=(\alph*),wide=0pt}
\begin{enumerate}
\item \textbf{Output and queries.} With probability at least $1-\delta$,
it returns $\bm{x}$ with $\|\nabla f(\bm{x})\|_\infty\le\eta$.
On every execution, every queried or returned point $\bm{y}$ satisfies
\begin{align}
 \|\bm{y}-\bm{x}_0\|_1&=O\left(\frac{H_0}{\eta}+\frac{\eta}{L}\right).\label{eq:scd-iterates}
\end{align}
\item \textbf{Cost.} On every execution, it uses $O(Q)$ oracle calls
and total time $O(TQ)$, where
\begin{align}
 Q&\coloneqq \left(1+\frac{\sigma^2}{\eta^2}\right)
       \left(1+\frac{LH_0}{\eta^2}\right)
       \log\left[\frac{2d}{\delta}
                   \left(1+\frac{LH_0}{\eta^2}\right)\right].\label{eq:scd-budget}
\end{align}
\end{enumerate}
\endgroup
\end{theorem}
For instance, if each coordinate of $\bm{G}(\bm{x})$ takes values in an interval
of length at most $2\sigma$, the noise condition holds by Hoeffding's
lemma~\cite{Hoeffding1963}; centered Gaussian errors of variance at most
$\sigma^2$ are also allowed. The algorithm and its analysis are given in
\cref{app:coordinate-descent}.

\begin{remark}\label{rem:approximate-sampling}
If the $j$th oracle output, conditional on the history, is within total
variation distance $\xi_j$ of its prescribed law, then the conclusion of
\cref{thm:stochastic-coordinate-descent}(a) holds with probability at
least $1-\delta-\sum_j\xi_j$, and the remaining guarantees hold on every
execution. Indeed, coupling each output with an exact one, the two runs
agree except with probability at most $\sum_j\xi_j$.
\end{remark}

\section{Global Mixing via Local Curvature}\label{sec:poincare-proof}
In this section, we prove \cref{thm:poincare}.
We begin by setting up some notation. We extend functions on $\Omega$ by zero on $2^{\cU} \setminus \Omega$.

\begin{definition}[Discrete Gradient]\label{def:discrete-gradient}
For each $u \in \cU$ and $f : 2^{\cU} \to \R$, let
\begin{align*}
    (\nabla_{u}f)(S) \coloneqq f(S + u) - f(S - u).
\end{align*}
Further define the upper and lower discrete derivatives via
\begin{align*}
    (\nabla_{u}^{+}f)(S) \coloneqq (\nabla_{u}f)(S) \cdot \allone[u \notin S] \qquad\text{and}\qquad (\nabla_{u}^{-}f)(S) \coloneqq (\nabla_{u}f)(S) \cdot \allone[u \in S].
\end{align*}
Note that $(\nabla f)(S), (\nabla^{+}f)(S), (\nabla^{-}f)(S)$ are vectors in $\R^{\cU}$.
\end{definition}

It will also be notationally convenient to define a ``downward'' analog of spectral curvature. For each $S \in \supp(\mu)$, let
\begin{align*}
    D_{S}^{\mu,\kappa}(u,v) &\coloneqq A_{\cU \setminus S}^{\mu^{*},\kappa}(u,v) = \begin{cases}
        \frac{\mu(S - v)^{2}}{\mu(S)^{2}} - \kappa \cdot \frac{\mu(S - v)}{\mu(S)}, &\quad\text{if } u,v \in S, u = v \\
        \frac{\mu(S - u)}{\mu(S)} \cdot \frac{\mu(S - v)}{\mu(S)} - \frac{\mu(S - u - v)}{\mu(S)}, &\quad\text{if } u,v \in S, u \neq v \\
        0, &\quad\text{otherwise}
    \end{cases}.
\end{align*}
This is a $\cU \times \cU$ matrix supported only on $S$. The (downward) $\kappa$-spectral curvature is then the condition that $D_{S}^{\mu,\kappa} \succeq 0$.

We will use the following Bochner-type identity to prove \cref{thm:poincare}.
\begin{proposition}\label{prop:bochner-identity}
Under the hypotheses of \cref{thm:poincare}, for every function $f : 2^{\cU} \to \R$, we have the identity,
\begin{align*}
    \wrapa{f, \wrapp{\mathscr{L}^{2} - (\eta + \kappa) \cdot \mathscr{L}} f}_{\rho} &= \sum_{S\in\Omega} \rho(S) \cdot (\nabla^{+}f)(S)^{\top} \cdot A_{S}^{\mu,\eta} \cdot (\nabla^{+}f)(S) \\
    &+ \sum_{S\in\Omega} \rho(S) \cdot (\nabla^{-}f)(S)^{\top} \cdot D_{S}^{\nu,\kappa} \cdot (\nabla^{-}f)(S) \\
    &+ \frac{1}{Z} \sum_{S\in\Omega} \sum_{\substack{u,v\notin S : u\neq v\\S+u+v\in\Omega}} \mu(S + u + v) \cdot \nu(S) \cdot (\Delta_{uv}f)(S)^{2},
\end{align*}
where
\begin{align*}
    (\Delta_{uv}f)(S) \coloneqq (\nabla_{u}\nabla_{v}f)(S) = f(S + u + v) - f(S + u - v) - f(S - u + v) + f(S - u - v),
\end{align*}
and $Z = \sum_{S \subseteq \cU} \mu(S) \cdot \nu(S)$ is the normalizing constant for $\rho$.
\end{proposition}
\cref{prop:bochner-identity} is proved in \cref{subsec:proof-bochner} below.

\begin{proof}[Proof of \cref{thm:poincare}]
In \Cref{prop:bochner-identity}, the term $(\Delta_{uv}f)(S)^2$ is non-negative automatically; and, by the spectral-curvature conditions, the first two terms are also non-negative.
Hence $\langle\sL f,\sL f\rangle_\rho\ge(\eta+\kappa) \cdot \langle f,\sL f\rangle_\rho$
for all $f$. The chain generated by $-\sL$ is reversible with respect
to $\rho$ and irreducible on $\Omega$, since $\Omega$ is
connected under single-element additions and deletions and every feasible
transition has positive rate. Thus \cref{prop:bochner-criterion} gives
spectral gap $\gamma\ge\eta+\kappa$, which is the desired Poincar\'e
inequality.
\end{proof}

\subsection{Proof of \cref{prop:bochner-identity}}\label{subsec:proof-bochner}
\begin{proof}
We show that subtracting the two curvature terms leaves a weighted
sum of squares of the mixed second differences $\Delta_{uv}f$,
the off-diagonal entries of the discrete Hessian of $f$.
Let $S\sim\rho$ and write
$\bm{r}(S)\coloneqq (r_u(S))_{u\in\cU}$,
where $r_u(S)$ is defined just before \cref{eq:full-laplacian}.
Since $r_u(S)$ is signed, the vector
of nonnegative transition rates is $|\bm{r}(S)|\coloneqq (|r_u(S)|)_{u\in\cU}$.
Self-adjointness and \cref{eq:dirichlet-form} give
\begin{align*}
 \inner{f}{\sL^2f}_\rho
 &=\inner{\sL f}{\sL f}_\rho
 =\E[S \sim \rho]{(\nabla f)(S)^\top \bm{r}(S)\bm{r}(S)^\top(\nabla f)(S)},\\
 \inner{f}{\sL f}_\rho
 &=\frac12 \cdot \E[S \sim \rho]{(\nabla f)(S)^\top\diag(|\bm{r}(S)|)(\nabla f)(S)}.
\end{align*}
Detailed balance pairs each addition with its reverse deletion, so
\begin{align*}
 \inner{f}{\sL f}_\rho
 &=\E{(\nabla^+f)(S)^\top\diag(|\bm{r}(S)|)(\nabla^+f)(S)}\\
 &=\E{(\nabla^-f)(S)^\top\diag(|\bm{r}(S)|)(\nabla^-f)(S)}.
\end{align*}
Define the matrix
\begin{align*}
 B_S&\coloneqq \offdiag\bigl(\bm{r}(S)\bm{r}(S)^\top-A_S^{\mu,\eta}-D_S^{\nu,\kappa}\bigr),
\end{align*}
where $\offdiag(M)$ sets the diagonal entries of $M$ to zero.
Subtracting the two curvature terms from the left-hand side gives
\begin{align*}
 \inner{f}{(\sL^2-(\eta+\kappa)\sL)f}_\rho
 &-\E{(\nabla^+f)(S)^\top A_S^{\mu,\eta}(\nabla^+f)(S)}
 -\E{(\nabla^-f)(S)^\top D_S^{\nu,\kappa}(\nabla^-f)(S)}\\
 &=\E{(\nabla f)(S)^\top B_S(\nabla f)(S)}.
\end{align*}
The diagonal terms cancel using the
preceding Dirichlet-form identities and the definitions of
$A_S^{\mu,\eta}$ and $D_S^{\nu,\kappa}$.
Expanding the expectation gives
\begin{align*}
 \E{(\nabla f)(S)^\top B_S(\nabla f)(S)}
 &=\sum_{X\in\Omega}\sum_{\substack{u,v\in\cU:\,u\ne v}}
   \rho(X)B_X(u,v)(\nabla_u f)(X)(\nabla_v f)(X).
\end{align*}

{
Whenever $S,T=S+u+v\in\Omega$ for distinct $u,v\notin S$, both
$S+u$ and $S+v$ also belong to $\Omega$.
Indeed, if $\mu(S+u)=0$, then $A_S^{\mu,\eta}(u,u)=0$ but
$A_S^{\mu,\eta}(u,v)=-\mu(T)/\mu(S)\ne0$, contradicting positive
semidefiniteness. The other intermediate set is handled symmetrically.}

\begin{observation}\label{obs:bochner-coefficients}
For $S\in\Omega$ and distinct $u,v\notin S$ with $T\coloneqq S+u+v\in\Omega$,
set $w\coloneqq \mu(T)\nu(S)/Z$. Then
\begin{align*}
 \rho(S)B_S(u,v)&=\rho(T)B_T(u,v)
 =-\rho(S+u)B_{S+u}(u,v)=-\rho(S+v)B_{S+v}(u,v)=w.
\end{align*}
\end{observation}

We regroup by $S=X-u-v$, writing $T=S+u+v$.
{For each $X\in\{S,S+u,S+v,T\}\cap\Omega$, direct substitution gives
$\rho(X)B_X(u,v)=\pm\mu(T)\nu(S)/Z$.
This vanishes unless both $S,T\in\Omega$, since the two measures have
common support. In that case, both intermediate sets are supported
by the argument above.}
Thus
\begin{align*}
 \E{(\nabla f)(S)^\top B_S(\nabla f)(S)}
 &=\sum_{S\in\Omega}
   \sum_{\substack{u,v\notin S:\,u\ne v\\T=S+u+v\in\Omega}}
   \sum_{X\in\{S,S+u,S+v,T\}}
   \rho(X)B_X(u,v)(\nabla_u f)(X)(\nabla_v f)(X).
\end{align*}
The gradient $\nabla_u f$ is unchanged when $u$ is added or removed,
and similarly for $\nabla_v f$. Thus \cref{obs:bochner-coefficients} gives, for each group,
\begin{align*}
 &\sum_{X\in\{S,S+u,S+v,T\}}
   \rho(X)B_X(u,v)(\nabla_u f)(X)(\nabla_v f)(X)\\
 &=w \cdot \bigl((\nabla_u f)(T)-(\nabla_u f)(S)\bigr)
      \bigl((\nabla_v f)(T)-(\nabla_v f)(S)\bigr)
 =w \cdot (\Delta_{uv}f)(S)^2.
\end{align*}
The last equality holds because each of the two differences equals
$(\nabla_u\nabla_v f)(S)=(\Delta_{uv}f)(S)$.
Substituting into the regrouped sum finishes the proof.
 \end{proof}


\section{Subset Convolution for Homogeneous Log-Concave Measures}\label{sec:fpras}
In this section, we prove the following formal version of
\cref{thm:homogeneous-convolution-informal}.

\begin{theorem}
\label{thm:homogeneous-convolution}
Let $\mu:\binom{\cU}{a}\to\mathbb R_{\ge0}$ and
$\nu:\binom{\cU}{b}\to\mathbb R_{\ge0}$, extended by zero to $2^{\cU}$,
where $a,b\ge0$ are integers with $a+b=n\coloneqq |\cU|\ge1$.
Suppose $g_\mu,g_\nu$ are log-concave. Given value oracles for $\mu,\nu$,
the coefficient ranges $\cR_\mu,\cR_\nu$ in
\cref{eq:coefficient-range}, and sets $S_0,T_0$ with
$\mu(S_0)>0$ and $\nu(T_0)>0$, there is an $\mathsf{FPRAS}$ for the
subset-convolution value $(\mu*\nu)(\cU)$, namely
\begin{align}\label{def:Z-cap}
 Z_\cap&\coloneqq [\bm{x}^{\cU}]\bigl(g_\mu(\bm{x})g_\nu(\bm{x})\bigr)
       =\sum_{S\subseteq\cU}\mu(S)\nu(\cU\setminus S),
 \qquad \bm{x}^{\cU}\coloneqq \prod_{u\in\cU}x_u.
\end{align}
For every $0<\eps\le1$, it uses at most
\begin{align}\label{eq:main-cost}
 O\left(n^{10}(n+\log(\cR_\mu\cR_\nu))^2\eps^{-2}
       \log^2\frac{2(n+\log(\cR_\mu\cR_\nu))}{\eps}\right)
\end{align}
arithmetic operations and value-oracle calls on every execution.
\end{theorem}

Throughout this section, $\mu,\nu$ satisfy the hypotheses of
\cref{thm:homogeneous-convolution}, with degrees $a,b$ and $a+b=n$.
We extend the measures to downward-closed set families by translating
their generating polynomials.

\begin{definition}[Independent-Set Pair Law]\label{def:independent-pair}
Let $\mu,\nu$ be as in \cref{thm:homogeneous-convolution},
$\bm{t}=(t_u)_{u\in\cU}\in(0,1]^{\cU}$, and $\bm{z}\in\mathbb R_{>0}^{\cU}$.
Write $\bm{z}\bm{x}\coloneqq (z_ux_u)_{u\in\cU}$ and, for $R\subseteq\cU$,
$\bm{z}^R\coloneqq \prod_{u\in R}z_u$ and $\bm{t}^R\coloneqq \prod_{u\in R}t_u$.
{Define the weights $\widehat\mu_{\bm{z}},\widehat\nu_{\bm{z}}$ through
their generating polynomials by}
\begin{align*}
 g_{\widehat\mu_{\bm{z}}}(\bm{x})&\coloneqq g_\mu(\mathbf1+\bm{z}\bm{x}),&
 g_{\widehat\nu_{\bm{z}}}(\bm{x})&\coloneqq g_\nu(\mathbf1+\bm{z}\bm{x}).
\end{align*}
Write $(\widehat\mu_{\bm{z}}\otimes\widehat\nu_{\bm{z}})(S,T)
\coloneqq \widehat\mu_{\bm{z}}(S)\widehat\nu_{\bm{z}}(T)$ and $\psi_{\bm{t}}(S,T)\coloneqq \bm{t}^{S\cap T}$.
The \emph{independent-set pair law}
on $2^{\cU}\times2^{\cU}$ is the normalized Hadamard product
\begin{align*}
 \widehat\pi_{\bm{t},\bm{z}}
 &\coloneqq (\widehat\mu_{\bm{z}}\otimes\widehat\nu_{\bm{z}})\odot\psi_{\bm{t}}.
\end{align*}
\end{definition}
{
\begin{remark}\label{rmk:completion-weights}
Explicitly, for $S,T\subseteq\cU$,
\[
 \widehat\mu_{\bm{z}}(S)
 =\bm{z}^S\sum_{S\subseteq A\subseteq\cU}\mu(A),
 \qquad
 \widehat\nu_{\bm{z}}(T)
 =\bm{z}^T\sum_{T\subseteq B\subseteq\cU}\nu(B).
\]
\end{remark}
}

For fixed $\bm{t},\bm{z}$, the law $\widehat\pi_{\bm{t},\lambda \bm{z}}$
converges as $\lambda\to\infty$ to the following bases pair law;
see \cref{subsec:fpras-activity-limit}.

\begin{definition}[Bases Pair Law]\label{def:pair-law}
With $\mu,\nu,\bm{t},\bm{z}$ as in \cref{def:independent-pair}, set
$\Omega\coloneqq \binom{\cU}{a}\times\binom{\cU}{b}$.
The \emph{bases pair law} is
$\widehat\pi_{\bm{t},\bm{z}}$ conditioned on $\Omega$.
Its partition function and probability mass function are
\begin{align*}
 Z_{\bm{t},\bm{z}}
 &\coloneqq \sum_{(S,T)\in\Omega}\mu(S)\nu(T)\bm{t}^{S\cap T}\bm{z}^S\bm{z}^T,\\
 \pi_{\bm{t},\bm{z}}(S,T)
 &\coloneqq \frac{\mu(S)\nu(T)\bm{t}^{S\cap T}\bm{z}^S\bm{z}^T}{Z_{\bm{t},\bm{z}}}.
\end{align*}
\end{definition}
A scalar penalty $t\in(0,1]$ denotes the constant vector $t\mathbf1$
in the notation for these laws and their partition functions.
{The three subsections below implement the framework described in
\cref{subsec:field-optimization}.}

\subsection{Poincar\'e Inequality for the Bases Pair Law \texorpdfstring{$\pi_{\bm{t},\bm{z}}$}{pi(t,z)}}\label{subsec:fpras-activity-limit}
\begin{lemma}\label{lem:bases-pair-poincare}
Let $\bm{t}\in(0,1]^{\cU}$, $\bm{z}\in\mathbb R_{>0}^{\cU}$, and let $\pi_{\bm{t},\bm{z}}$ be the
bases pair law of \cref{def:pair-law}. On $\supp\pi_{\bm{t},\bm{z}}$, consider
the exchange chain with rates
\begin{align}\label{eq:pair-exchange-rates}
 q\bigl((S,T),(S-u+v,T)\bigr)
 &\coloneqq t_u^{-\allone[u\in T]}
       \frac{z_v\mu(S-u+v)}
       {\sum_{w\in\cU\setminus(S-u)}z_w\mu(S-u+w)},
\end{align}
where $u\in S$ and $v\notin S$.
The rates updating $T$ are obtained by interchanging $(S,\mu)$ and
$(T,\nu)$; all other off-diagonal rates are zero. This chain is
reversible with respect to $\pi_{\bm{t},\bm{z}}$ and satisfies
\begin{align}\label{eq:bases-pair-poincare}
 \Var[\pi_{\bm{t},\bm{z}}]{f}&\le\cE_q(f,f), \quad \forall f:\supp\pi_{\bm{t},\bm{z}}\to\mathbb R.
\end{align}
\end{lemma}

Applying \cref{eq:poincare-mixing} to the exchange chain gives the
following sampling bound.
\begin{corollary}\label{cor:bases-pair-sampling}
Under the hypotheses of \cref{lem:bases-pair-poincare} with a common
scalar penalty $t$, let $z_{\min}\coloneqq \min_u z_u$ and $z_{\max}\coloneqq \max_u z_u$.
Given a pair $(S_0,T_0)\in\supp\pi_{t,\bm{z}}$, the coefficient ranges
$\cR_\mu,\cR_\nu$ in \cref{eq:coefficient-range}, and $0<\eps<1$, one can sample
within total variation distance $\eps$ of $\pi_{t,\bm{z}}$ using
\begin{align*}
 O\left(\frac{n^2}{t}\left(
 n+\log(\cR_\mu\cR_\nu)+n\log\frac{z_{\max}}{z_{\min}}
 +\min\{a,b\}\log\frac1t+\log\frac1\eps\right)\right)
\end{align*}
arithmetic operations and calls to the value oracles for $\mu,\nu$.
\end{corollary}
\begin{proof}
Apply \cref{eq:poincare-mixing} with $\gamma=1$ and $R=n/t$:
there are $a+b=n$ possible deletions, each with total rate at most $1/t$.
{Each step of the lazy discrete-time chain costs $O(n)$ operations and oracle calls:}
choose a present copy uniformly and, if it is $u\in S$, with probability
$t^{1-\allone[u\in T]}/2$ resample it using the reinsertion weights in
\cref{eq:pair-exchange-rates}, allowing reinsertion of $u$; the update
of $T$ is symmetric. Since there are at most $4^n$ pairs, each of total
size $n$,
\begin{align*}
 \pi_{\min}
 &\ge\frac{t^{\min\{a,b\}}}{4^n\cR_\mu\cR_\nu}
       \left(\frac{z_{\min}}{z_{\max}}\right)^n.
\end{align*}
Substituting this bound and multiplying by the cost per step proves
the claim.
\end{proof}

To prove \Cref{lem:bases-pair-poincare}, we first apply
\cref{thm:poincare} to the independent-set pair law
$\widehat\pi_{\bm{t},\lambda \bm{z}}$ of \cref{def:independent-pair}.
The induced chain on pairs of bases inherits the Poincar\'e bound, and
letting $\lambda\to\infty$ gives the desired exchange chain.

\paragraph{Poincar{\'e} inequality for $\widehat{\pi}_{\bm{t},\lambda \bm{z}}$.}
We will first use \Cref{thm:poincare} to show the Poincar{\'e} inequality on $\widehat{\pi}_{\bm{t},\lambda \bm{z}}$. For convenience, for $\lambda>0$, define the measure
\begin{align*}
 w_\lambda(S,T)&\coloneqq \widehat\mu_{\lambda \bm{z}}(S)
                    \widehat\nu_{\lambda \bm{z}}(T),
\end{align*}
whose support is identical to
\begin{align*}
    \supp\widehat{\pi}_{\bm{t},\lambda \bm{z}} &=\supp\widehat{\pi}_{\bm{t},\bm{z}} = \supp\widehat{\mu}_{\lambda \bm{z}}\times\supp\widehat{\nu}_{\lambda \bm{z}}.
\end{align*}
Restrict $\psi_{\bm{t}}$ from \cref{def:independent-pair} to this downward closed support, and extend both measures by zero outside this support.
Recall from \Cref{def:independent-pair} that
\begin{align} \label{eq:hat-pi-recall}
    \widehat{\pi}_{\bm{t},\lambda \bm{z}} = w_\lambda \odot \psi_{\bm{t}}.
\end{align}

\begin{corollary}
\label{cor:independent-pair-poincare}
Let $\bm{t}\in(0,1]^{\cU}$, $\bm{z}\in\mathbb R_{>0}^{\cU}$, and $\lambda>0$.
Let $\sL$ be the Laplacian in \cref{eq:full-laplacian} for the
factors $w_\lambda,\psi_{\bm{t}}$, and let $\widehat q_\lambda$ be its transition rates.
For every $F:\supp\widehat{\pi}_{\bm{t},\lambda \bm{z}}\to\mathbb R$, we have
\begin{align}\label{eq:independent-pair-poincare}
 \Var[\widehat\pi_{\bm{t},\lambda \bm{z}}]{F}
 &\le\cE_{\widehat q_\lambda}(F,F)
 =\langle F,\sL F\rangle_{\widehat\pi_{\bm{t},\lambda \bm{z}}}.
\end{align}
\end{corollary}

In order to apply \Cref{thm:poincare} to $\widehat{\pi}_{\bm{t},\lambda \bm{z}}$, we have the following observations.
\begin{observation} \label{obs:w-0-curvature}
    $w_\lambda$ has $0$-spectral curvature.
\end{observation}
\begin{proof}
    The generating polynomial of $w_\lambda$ is
    \begin{align*}
     g_{w_\lambda}(\bm{x},\bm{y})
     &=g_\mu(\mathbf1+\lambda \bm{z}\bm{x})\,
              g_\nu(\mathbf1+\lambda \bm{z}\bm{y}).
    \end{align*}
    Since $p(\bm{x},\bm{y})\coloneqq g_\mu(\bm{x})g_\nu(\bm{y})$ is log-concave,
    \cref{prop:flc-curvature} with $\alpha=1$ gives zero curvature for
    the coefficient measure of $p(\mathbf1+\bm{x},\mathbf1+\bm{y})$.
    {Positive rescaling of the variables preserves zero curvature:
    each curvature matrix $A$ becomes $DAD$, where $D$ is a diagonal matrix
    with positive entries, so positive semidefiniteness is preserved.}
    Applying the rescaling $(\bm{x},\bm{y})\mapsto(\lambda \bm{z}\bm{x},\lambda \bm{z}\bm{y})$
    therefore gives the assertion for
    $g_{w_\lambda}(\bm{x},\bm{y})=p(\mathbf1+\lambda \bm{z}\bm{x},\mathbf1+\lambda \bm{z}\bm{y})$.
\end{proof}
\begin{observation} \label{obs:psi-star-1-curvature}
    $\psi_{\bm{t}}^*$ has $1$-spectral curvature.
\end{observation}
\begin{proof}
By definition, the dual measure $\psi_{\bm{t}}^*$ is
\begin{align*}
\psi_{\bm{t}}^*(S,T)
&\coloneqq \psi_{\bm{t}}(\cU\setminus S,\cU\setminus T)
=\bm{t}^{\cU\setminus(S\cup T)}
\allone[(\cU\setminus S,\cU\setminus T)\in\supp\widehat{\pi}_{\bm{t},\bm{z}}].
\end{align*}
For $Y=(S,T)\in\supp\psi_{\bm{t}}^*$, downward closure of $\supp\widehat{\pi}_{\bm{t},\bm{z}}$
ensures that all additions to $Y$ remain supported.
Substitution into \cref{def:spectral-curvature} gives, for
$u,v\in\cU$ and $i,j\in\{1,2\}$,
\begin{align*}
A_Y^{\psi_{\bm{t}}^*,1}\bigl((u,i),(v,j)\bigr)
&=\begin{cases}
(1-t_u)/t_u^2,& u=v,\ u\notin S\cup T,\\
0,&\text{otherwise}.
\end{cases}
\end{align*}
Here, we use $(u,1)$ and $(u,2)$ to denote that the element $u$ is added to $S$ and $T$, respectively.
Thus $A_Y^{\psi_{\bm{t}}^*,1}$ is a direct sum of zero blocks and blocks
$\frac{1-t_u}{t_u^2}\begin{psmallmatrix}1&1\\1&1\end{psmallmatrix}$,
all positive semidefinite. Hence $\psi_{\bm{t}}^*$ has upward curvature one,
and \cref{thm:poincare} applies with $\eta=0$ and $\kappa=1$.
\end{proof}

\begin{proof}[Proof of \Cref{cor:independent-pair-poincare}]
Apply \cref{thm:poincare} to $w_\lambda\odot\psi_{\bm{t}}$ with
$(\eta,\kappa)=(0,1)$, using
\cref{obs:w-0-curvature,obs:psi-star-1-curvature}.
Normalizing the factors does not change their weight ratios.
\end{proof}

\paragraph{Poincar{\'e} inequality of $\pi_{\bm{t},\bm{z}}$ by letting $\lambda \to \infty$.}
We recall the induced-chain construction and its Poincar\'e comparison.
\begin{definition}[Induced chain]\label{def:induced-chain}
Let $(X_s)_{s\ge0}$ be an irreducible continuous-time Markov chain on a
finite state space $\Omega$, with off-diagonal rates $r(x,y)$, and let
$\varnothing\ne B\subseteq\Omega$. Its \emph{induced chain on $B$} is the chain
obtained by deleting all time spent outside $B$. Writing
$\tau_B\coloneqq \inf\{s\ge0:X_s\in B\}$, its off-diagonal rates are
\begin{align*}
 r^B(x,y)&\coloneqq r(x,y)+\sum_{u\in\Omega\setminus B}
 r(x,u){\Pr{X_{\tau_B}=y\mid X_0=u}}
 \qquad(x,y\in B,\ x\ne y).
\end{align*}
\end{definition}

\begin{proposition}[{\cite[Theorem~13.16]{LevinPeres2017}}]
\label{prop:induced-poincare}
In the setting of \cref{def:induced-chain}, suppose the chain is
reversible with stationary law $\pi$ and has spectral gap at least
$\gamma>0$. Then the induced chain is reversible with respect to
$\pi_B(x)\coloneqq \pi(x)/\pi(B)$ for $x\in B$, and for every $f:B\to\mathbb R$,
\begin{align*}
 \gamma\Var[\pi_B]{f}&\le\cE_{r^B}(f,f).
\end{align*}
\end{proposition}

The proof of \cite[Theorem~13.16]{LevinPeres2017} is for discrete-time
Markov chains but extends directly to the continuous-time setting.
For the reader's convenience, we include a proof in
\cref{app:induced-poincare}.

\begin{proof}[Proof of \Cref{lem:bases-pair-poincare}]
We pass to the induced chain on pairs of bases and show that its rates converge
to \cref{eq:pair-exchange-rates}. Singleton support is immediate.
Recall from \cref{eq:full-laplacian} that the chain with generator
$-\sL$ has rates
\begin{align*}
 \widehat q_\lambda\bigl((S,T),(S+u,T)\bigr)
 &=\frac{w_\lambda(S+u,T)}{w_\lambda(S,T)}
 =\frac{\widehat\mu_{\lambda \bm{z}}(S+u)}{\widehat\mu_{\lambda \bm{z}}(S)}
 &&(u\notin S),\\
 \widehat q_\lambda\bigl((S,T),(S-u,T)\bigr)
 &=\frac{\psi_{\bm{t}}(S-u,T)}{\psi_{\bm{t}}(S,T)}
 =t_u^{-\allone[u\in T]}
 &&(u\in S).
\end{align*}
The rates updating $T$ are symmetric; unsupported additions are
omitted and all other off-diagonal rates are zero.

Let $q_\lambda^{\mathrm{ind}}$ be the rates of the chain induced on
$\supp\pi_{\bm{t},\bm{z}}$ from \cref{def:induced-chain}.
Conditioning $\widehat\pi_{\bm{t},\lambda \bm{z}}$ on pairs of bases cancels
the common factor $\lambda^n$, giving exactly $\pi_{\bm{t},\bm{z}}$.
By \cref{prop:induced-poincare,cor:independent-pair-poincare}, the induced chain
is therefore reversible with stationary law $\pi_{\bm{t},\bm{z}}$ and has
spectral gap at least one.

It remains to identify the limiting rates. After deleting $u\in S$
from a pair of bases $(S,T)$, put $S'\coloneqq S-u$. At $(S',T)$, additions
have positive total rate proportional to $\lambda$, whereas the
 total deletion rate is at most $n/\min_i t_i$. Thus the next jump is an
addition with probability $1-O(\lambda^{-1})$. Every such addition
returns to a pair of bases, and its conditional reinsertion law is
\begin{align*}
 \Pr{\text{reinsert }v\mid\text{next jump is an addition}}
 &=\frac{z_v\mu(S'+v)}{\sum_{w\notin S'}z_w\mu(S'+w)}
 \qquad(v\notin S').
\end{align*}
Multiplying by the initial deletion rate $t_u^{-\allone[u\in T]}$
gives \cref{eq:pair-exchange-rates} in the limit. The same argument
applies to deletions from $T$; excursions with a further deletion
have probability $O(\lambda^{-1})$. Hence
$q_\lambda^{\mathrm{ind}}(X,Y)\to q(X,Y)$ for every distinct pair
$X,Y\in\supp\pi_{\bm{t},\bm{z}}$.

Reversibility passes to the limit. Since the state space is finite,
passing to the limit in the induced chains' Poincar\'e inequalities
gives, for every $f$,
\begin{align*}
 \Var[\pi_{\bm{t},\bm{z}}]{f}
 &\le\lim_{\lambda\to\infty}\cE_{q_\lambda^{\mathrm{ind}}}(f,f)
 =\cE_q(f,f). \qedhere
\end{align*}
\end{proof}

\subsection{\texorpdfstring{{Choosing External Fields by Stochastic Gradient Descent}}{Choosing External Fields by Stochastic Gradient Descent}}\label{subsec:fpras-fields}

{
We now use the Poincar\'e bound to choose fields that make complementary
pairs occur with constant probability when $Z_\cap>0$.
We first prove \cref{lem:field-concavity}, restated below, and then apply
stochastic gradient descent to find logarithmic coordinates $\bm{y}$
with small $\nabla\Phi(\bm{y})$, then set $\bm{z}=\e^{\bm{y}}$.

\fieldconcavitylemma*

\begin{proof}
Fix $u$, and for $j \in \{0,1,2\}$, let $p_j$ be the probability under $\pi_{t,\bm{z}}$ that $u$ appears in exactly $j$ of the sets $S,T$. Since $p_2-p_0=\partial_u\Phi(\bm{y})$,
\begin{align*}
 p_0+p_2&\le|p_2-p_0|+2\sqrt{p_0p_2}
 =|\partial_u\Phi(\bm{y})|+2\sqrt{p_0p_2},
\end{align*}
so a union bound over $u$ reduces the lemma to
\begin{align*}
 p_0p_2&\le\frac{4t}{\gamma^2}\,p_1^2.
\end{align*}
Assume $p_0p_2>0$, and set $s\coloneqq \sqrt{tp_0/p_2}$. Remove the penalty at
$u$ and multiply $z_u$ by $s$: let $\bm{t}',\bm{z}'$ agree with $t\mathbf1,\bm{z}$
except at $u$, where $t'_u\coloneqq 1$ and $z'_u\coloneqq sz_u$. Then
\begin{align*}
 \pi_{\bm{t}',\bm{z}'}(S,T)&\propto\pi_{t,\bm{z}}(S,T)\cdot
 s^{\allone[u\in S]+\allone[u\in T]}\,t^{-\allone[u\in S\cap T]},
\end{align*}
so the probabilities $q_0,q_1,q_2$ under $\pi_{\bm{t}',\bm{z}'}$ that $u$ lies in
exactly $0,1,2$ sets satisfy
\begin{align*}
 (q_0,q_1,q_2)&\propto(p_0,\,sp_1,\,s^2p_2/t)=(p_0,\,sp_1,\,p_0),
 \qquad\text{so }q_0=q_2\le\tfrac12.
\end{align*}
Let $\cL'$ be the Laplacian supplied by the hypothesis for
$\pi_{\bm{t}',\bm{z}'}$. By assumption, the total rate of moves from one copy of
$u$ to none is at most $1$, and no move goes from two copies to none.
By reversibility, the Dirichlet form of
$g(S,T)\coloneqq \allone[u\notin S\cup T]$ is therefore at most $q_1$.
The Poincar\'e inequality gives
\begin{align*}
 \gamma\,q_0(1-q_0)=\gamma\Var[\pi_{\bm{t}',\bm{z}'}]{g}
 &\le\langle g,\cL' g\rangle_{\pi_{\bm{t}',\bm{z}'}}\le q_1,
 &&\text{hence}&
 \frac\gamma2&\le\frac{q_1}{q_0}=\frac{\sqrt t\,p_1}{\sqrt{p_0p_2}}.
 \qedhere
\end{align*}
\end{proof}
}

For the bases pair law, we apply \cref{lem:field-concavity} to the exchange chain of
\cref{lem:bases-pair-poincare}, which has gap $1$ for every $\bm{t},\bm{z}$
and updates only one set at a time. If $u\in S\setminus T$, its total
removal rate is
\begin{align*}
 \sum_{v\notin S}q\bigl((S,T),(S-u+v,T)\bigr)
 &=\frac{\sum_{v\notin S}z_v\mu(S-u+v)}
         {\sum_{w\notin(S-u)}z_w\mu(S-u+w)}\le1.
\end{align*}
The same bound holds for $u\in T\setminus S$, so the hypotheses of
\cref{lem:field-concavity} hold with $\gamma=1$.
Thus making complementary pairs frequent reduces to finding an external
field $\bm{z}$ with small $\nabla\Phi(\log\bm{z})$.
We use stochastic coordinate descent as follows.

\begin{corollary}
\label{cor:balanced-fields}
Under the hypotheses of \cref{thm:homogeneous-convolution}, put
\begin{align*}
 t&\coloneqq \frac1{256n^2},&
 B&\coloneqq 2n+\lceil\log_2(\cR_\mu\cR_\nu)\rceil.
\end{align*}
Given value oracles for $\mu,\nu$, the coefficient ranges $\cR_\mu,\cR_\nu$ in
\cref{eq:coefficient-range}, and a supported pair $(S_0,T_0)$,
there is a randomized algorithm that returns $\bm{z}\in\mathbb R_{>0}^{\cU}$
or reports failure. If $Z_\cap>0$, with probability at least $31/32$
it returns a field satisfying
\begin{align*}
 \pi_{t,\bm{z}}(T=\cU\setminus S)&\ge\frac12.
\end{align*}
On every execution, it uses
$O(n^{10}B^2\log(2B))$ arithmetic operations and value-oracle calls,
and every queried or returned field satisfies
\begin{align}\label{eq:balancing-bounds}
 \|\log \bm{z}\|_1&=O(nB).
\end{align}
\end{corollary}
\begin{proof}
By \cref{lem:field-concavity}, it suffices to find $\bm{y}$ with
$\|\nabla\Phi(\bm{y})\|_\infty\le\eta\coloneqq 1/(4n)$ and return $\bm{z}\coloneqq \exp(\bm{y})$:
then $n(\eta+ 4\sqrt t)=1/2$. We apply
\cref{thm:stochastic-coordinate-descent} with
$d=n$, $L=\sigma=1$, $H_0=B$, and $\bm{x}_0=\bm{0}$.
For $(S,T)\sim\pi_{t,\exp(\bm{y})}$, the vector
$\bm{G}\coloneqq \mathbf1_S+\mathbf1_T-\mathbf1$ satisfies
\begin{align*}
 \nabla\Phi(\bm{y})&=\E{\bm{G}},&
 \nabla^2\Phi(\bm{y})&=\operatorname{Cov}(\bm{G})\succeq0,\\
 \partial_u^2\Phi(\bm{y})&=\Var{G_u}\le1,&
 \bm{G}&\in[-1,1]^n.
\end{align*}
Thus $\Phi$ is convex and coordinatewise $1$-smooth, and Hoeffding's
lemma gives coordinatewise $1$-sub-Gaussian oracle noise.
When $Z_\cap>0$, every complementary pair contributes its unchanged
weight to $\exp(\Phi(\bm{y}))$. Writing
$M\coloneqq \max_S\mu(S)\max_T\nu(T)$, we therefore have
\begin{align*}
 \Phi(\bm{y})&\ge\log Z_\cap\ge\log\frac{M}{\cR_\mu\cR_\nu},&
 \Phi(\bm{0})&\le\log(4^nM).
\end{align*}
Thus the required gap bound is at most $B$, so take $H_0\coloneqq B$;
no minimizing point need exist.

Take $\delta=1/64$ in \cref{thm:stochastic-coordinate-descent}.
The algorithm uses at most $Q_{\mathrm{cd}}=O(n^4B\log(2B))$ gradient
samples and $O(T'Q_{\mathrm{cd}})$ operations, where $T'$ bounds the
cost of one gradient sample. It also gives
$\|\bm{y}\|_1=O(B/\eta+\eta)=O(nB)$ for every queried or returned point,
which proves \cref{eq:balancing-bounds}.
We stop and report failure before exceeding any of these bounds.
Thus they also hold when $Z_\cap=0$; when $Z_\cap>0$, this stopping
rule does not change the algorithm with exact samples.

Supply each gradient sample using \cref{cor:bases-pair-sampling},
with fresh randomness and total variation error at most
$1/(64Q_{\mathrm{cd}})$.
By \cref{rem:approximate-sampling}, these sampling errors add at most
$1/64$ to the failure probability. Hence the algorithm finds the
required field with probability at least $1-1/64-1/64=31/32$
when $Z_\cap>0$.

For the running time, the field bound gives
$n\log(z_{\max}/z_{\min})=O(n^2B)$ at every query.
Since $t^{-1}=256 n^2$ and $B=\Theta(n+\log(\cR_\mu\cR_\nu))$,
\cref{cor:bases-pair-sampling} gives $T'=O(n^6B)$ operations and
value-oracle calls per gradient sample.
Multiplying by $Q_{\mathrm{cd}}$ proves the claimed bound;
the $O(nQ_{\mathrm{cd}})$ additional operations cost less.
The field updates can be implemented using arithmetic operations only;
see \cref{app:counting-arithmetic}.
\end{proof}

\subsection{Wrapping Up}\label{subsec:fpras-counting}

\begin{proof}[Proof of \Cref{thm:homogeneous-convolution}]
The cases $a=0$ or $b=0$ are evaluated directly, so assume $n\ge2$.
We choose fields making complementary pairs frequent, then estimate
the two unknown factors in
\begin{align}\label{eq:counting-factorization}
 Z_\cap&=\frac{Z_{t,\bm{z}}}{\bm{z}^{\cU}} \times
          \pi_{t,\bm{z}}(T=\cU\setminus S).
\end{align}
Let $B,t$ be as in \cref{cor:balanced-fields}, and put
$\beta\coloneqq \eps/4$ and $\delta_0\coloneqq 1/64$.
Apply \cref{cor:balanced-fields}; if it reports failure, return zero.
Otherwise fix the returned field $\bm{z}$. When $Z_\cap>0$, with probability
at least $1-2\delta_0$, this field satisfies
$\pi_{t,\bm{z}}(T=\cU\setminus S)\ge1/2$.

\paragraph{Estimating the normalizer.}
For any returned field $\bm{z}$, we apply
\cref{lem:sampling-to-counting} to the $2n$ membership bits of $(S,T)$,
with weight
\begin{align*}
 W(S,T)&\coloneqq \mu(S)\nu(T)t^{|S\cap T|}\bm{z}^S\bm{z}^T.
\end{align*}
The supplied pair $(S_0,T_0)$ is a valid starting state.

Each conditional law is obtained by restricting $\mu,\nu$ to the
prescribed inclusions and exclusions, with the same $t,\bm{z}$.
We keep the original ground set and the forced elements in each set,
so the degrees remain $a,b$. Inclusion replaces a generating polynomial
by its coordinate derivative times that variable; exclusion sets the
variable to zero. These operations preserve log-concavity, and the
coefficient ranges do not increase. For every nonempty
conditional law, both measures are nonzero. Thus \cref{cor:bases-pair-sampling} applies
to every conditional law.

\Cref{lem:sampling-to-counting}, with relative error
$\beta$ and failure probability $\delta_0$, therefore gives an estimate
$\widehat Z$ of $Z_{t,\bm{z}}$ using
$O(n^3\beta^{-2}\log(2n/\delta_0))$ conditional samples, each within
total variation distance $\beta/(32n)$ of its law.
The sample counts and running times are bounded on every execution.

\paragraph{Combining the estimates.}
Independently draw
$m=\lceil8\beta^{-2}\log(2/\delta_0)\rceil$ pairs from $\pi_{t,\bm{z}}$,
each to total variation error $\beta/4$, and let $\widehat p$ be the
fraction that are complementary. Conditional on the returned fields,
Hoeffding's inequality gives
$|\widehat p-\pi_{t,\bm{z}}(T=\cU\setminus S)|\le\beta/2$ except with probability $\delta_0$:
the sampling bias is at most $\beta/4$, as is the fluctuation of
the empirical average. Return
\begin{align*}
 \widehat Z_\cap&\coloneqq \widehat p\,\widehat Z/\bm{z}^{\cU}.
\end{align*}
{When the returned field satisfies
$\pi_{t,\bm{z}}(T=\cU\setminus S)\ge1/2$ and both preceding
estimation bounds hold, this is a product of two estimates, each with
relative error at most $\beta$. These conditions hold jointly with
probability at least $1-4\delta_0>3/4$, so}
\begin{align*}
 (1-\beta)^2&\le\widehat Z_\cap/Z_\cap\le(1+\beta)^2,
\end{align*}
which implies relative error $\eps$. If $Z_\cap=0$, every supported
pair fails the complementary-pair test. Thus every execution returns
zero, either at the field-finding step or because $\widehat p=0$.

\paragraph{Running time.}
Finding the field costs $O(n^{10}B^2\log(2B))$ by
\cref{cor:balanced-fields}. Throughout the self-reduction, the
coefficient ranges do not increase, and
$n\log(z_{\max}/z_{\min})=O(n^2B)$ by
\cref{eq:balancing-bounds}. Each restricted oracle uses one original
query and $O(n)$ operations to check the prescribed memberships.
Thus \cref{cor:bases-pair-sampling}, with $t^{-1}= 256 n^2$, bounds the
cost of a conditional sample by
\begin{align*}
 O\left(n^5\left(n^2B+\log\frac{n}{\eps}\right)\right).
\end{align*}
Self-reduction uses $O(n^3\eps^{-2}\log(2n))$ such samples.
Estimating the complementary-pair probability takes
$O(\eps^{-2})$ additional samples and is covered by the same bound.
Since $B=\Theta(n+\log(\cR_\mu\cR_\nu))$, the total cost is
\begin{align*}
 O\left(n^{10}B^2\eps^{-2}\log^2\frac{2B}{\eps}\right),
\end{align*}
as required by \cref{eq:main-cost}.
\end{proof}

{\sloppy
\printbibliography
}

\clearpage
\appendix
\crefalias{section}{appendix}
\crefalias{subsection}{subappendix}
\crefalias{subsubsection}{subsubappendix}

\section{Deferred Proofs}\label{sec:deferred-proof}

\subsection{Proofs of the Convolution Theorems}\label{app:convolution-proofs}
We prove \cref{thm:log-concave-convolution,thm:curvature-convolution}
by a common reduction to the coefficient $[\bm{x}^\cU](g_\mu g_\nu)$,
followed by separate arguments for the two hypotheses.
We first record three lemmas about supports and homogeneous components;
proofs of the first two are deferred to \cref{app:support-lemmas}.

\begin{lemma}\label{lem:positive-curvature-support}
Let $\omega$ be a nonzero nonnegative measure on $2^{\cU}$ whose support
$\Omega$ is connected under single-element additions and deletions, and put
$V\coloneqq\bigcup_{S\in\Omega}S$.
If $\omega$ has upward spectral curvature \mbox{$\kappa>0$}, then $\Omega$ is upward closed in $2^V$.
In particular, if $\Omega$ is also downward closed, then $\Omega=2^V$,
where \mbox{$V=\{u:\omega(\{u\})>0\}$}.
\end{lemma}

\begin{samepage}
\begin{lemma}
\label{lem:lc-support-access}
Let $g_\omega(\bm{x})=\sum_{A\subseteq\cU}\omega(A)\bm{x}^A$ be nonzero and
log-concave, with nonnegative coefficients, and put $n\coloneqq |\cU|$.
Given a value oracle for $\omega$, a set $A_0\in\supp\omega$, and disjoint
$C,O\subseteq\cU$, an algorithm using $O(n^3)$ queries can:
\begin{itemize}
\item certify that no $A\in\supp\omega$ contains $C$ and avoids $O$; or
\item find, for every size at which such a set exists, a set
$A\in\supp\omega$ containing $C$ and avoiding $O$.
\end{itemize}
\end{lemma}
\end{samepage}

\begin{lemma}[{\cite[Corollary~28]{RobinsonSraJegelka2019}}]
\label{lem:lc-homogeneous-layers}
Let $g$ be a multiaffine log-concave polynomial with nonnegative
coefficients. For each degree $k$, the sum of the degree-$k$ terms of $g$
is either zero or log-concave.
\end{lemma}

\paragraph{Reduce to the coefficient $[\bm{x}^\cU](g_\mu g_\nu)$.}
Put $C\coloneqq \{u:\alpha_u=2\}$, $W\coloneqq \{u:\alpha_u=1\}$, and
$O\coloneqq \{u:\alpha_u=0\}$. In a pair contributing to $Z_{\bm{\alpha}}$, both sets
contain $C$ and avoid $O$, while each element of $W$ belongs to exactly
one of them. Define measures $\mu',\nu'$ on $2^W$ by
\begin{align*}
 \mu'(A)&\coloneqq \mu(C\cup A),&
 \nu'(A)&\coloneqq \nu(C\cup A),\qquad A\subseteq W.
\end{align*}
Then
\begin{align*}
 Z_{\bm{\alpha}}&=\sum_{A\subseteq W}\mu'(A)\nu'(W\setminus A).
\end{align*}
If $W=\varnothing$, return $Z_{\bm{\alpha}}=\mu(C)\nu(C)$ directly.
Henceforth assume $W\ne\varnothing$.
We check that these measures satisfy the assumptions needed for the reduction.

Under \cref{thm:curvature-convolution}, if $\mu(C)\nu(C)=0$,
then $Z_{\bm{\alpha}}=0$ by downward closure. Otherwise, both new measures
are nonzero, have downward-closed supports, and satisfy
\[
 A_S^{\mu',\eta}
 =\bigl(A_{C\cup S}^{\mu,\eta}\bigr)_{W\times W}\succeq0,
 \qquad S\in\supp\mu'.
\]
The same identity holds for $\nu'$. Thus $\mu'$ and $\nu'$ have upward
spectral curvatures $\eta$ and $\kappa$, respectively.

When $\gamma>0$, for $S\in\supp\mu'$ and $u\in W\setminus S$,
\begin{align*}
 \mu'(S+u)&=\mu(C\cup S+u)
 \le\chi\gamma\,\mu(C\cup S)=\chi\gamma\,\mu'(S).
\end{align*}
The same calculation applies to $\nu'$, so the bound in \cref{thm:curvature-convolution}
holds for both new measures.

Under \cref{thm:log-concave-convolution}, obtain $g_{\mu'}$ from $g_\mu$
by differentiating once with respect to each variable indexed by $C$,
then setting the variables indexed by $O$ to zero. The first operation
keeps the terms containing $C$ and removes their $C$ factors; the second
removes the terms containing any element of $O$.
For multiaffine polynomials, both operations preserve log-concavity
whenever the resulting polynomial is not zero. Hence $g_{\mu'}$ is
either zero or log-concave, and the same holds for $g_{\nu'}$.

To obtain starting sets, apply \cref{lem:lc-support-access} using the
supplied sets $S_0,T_0$. This finds sets in $\supp\mu$ and $\supp\nu$
that contain $C$ and avoid $O$, or reports that one does not exist.
In the latter case, $Z_{\bm{\alpha}}=0$. Otherwise, removing $C$ from the
returned sets gives starting sets in $\supp\mu'$ and $\supp\nu'$.

In both cases, evaluating $\mu'(A)$ requires only the oracle query
$\mu(C\cup A)$, and similarly for $\nu'$.
Also, $\cR_{\mu'}\le\cR_\mu$ and $\cR_{\nu'}\le\cR_\nu$,
since each new measure uses only values of its original measure.
We now write $\mu,\nu$ for $\mu',\nu'$ and $\cU,n$ for $W,|W|$.
It therefore suffices to estimate
\begin{align*}
 Z_\cap&=\sum_{A\subseteq\cU}\mu(A)\nu(\cU\setminus A).
\end{align*}

\begin{proof}[Proof of \cref{thm:log-concave-convolution}]
By the reduction above, it suffices to estimate $Z_\cap$,
using \cref{thm:homogeneous-convolution}, proved in \cref{sec:fpras}.
Write $g_{\mu,k}$ for the sum of the degree-$k$ terms of $g_\mu$,
and similarly for $\nu$. By \cref{lem:lc-homogeneous-layers}, each of
these polynomials is either zero or log-concave. Then
\begin{align*}
 Z_\cap&=\sum_{k=0}^n
     [\bm{x}^\cU]\bigl(g_{\mu,k}(\bm{x})g_{\nu,n-k}(\bm{x})\bigr).
\end{align*}
Each layer inherits its value oracle and coefficient-range bound;
\cref{lem:lc-support-access} supplies {sets of positive weight} in all nonempty
layers using $O(n^3)$ queries. Apply \cref{thm:homogeneous-convolution}
to each pair of nonempty layers, with relative error $\eps$ and failure
probability $1/(4(n+1))$, using $O(\log(2n))$ repetitions.
Summing the estimates and taking a union bound proves the guarantee.
{The at most $n+1$ applications of \cref{eq:main-cost}, including
the repetitions needed to boost the success probability, cost}
\begin{align*}
 O\left(n^{11}(n+\log(\cR_\mu\cR_\nu))^2\eps^{-2}
       \log(2n)\log^2\frac{2(n+\log(\cR_\mu\cR_\nu))}{\eps}\right).
\end{align*}
This is the bound in \cref{thm:log-concave-convolution}.
\end{proof}

\begin{proof}[Proof of \cref{thm:curvature-convolution}]
By the reduction above, it suffices to estimate $Z_\cap$.
We use the supplied curvature parameters $\eta,\kappa$ to select the case below and to compute $\gamma$.
If $\max\{\kappa,\eta\}=0$, then $\eta=\kappa=0$ and
\cref{prop:zero-curvature-log-concavity} makes both generating polynomials
log-concave; the empty sets supply {starting sets of positive weight}. Thus
\cref{thm:log-concave-convolution} applies, and its cost is the first term in the stated bound.
It remains to treat the case $\max\{\eta,\kappa\}>0$, whose cost is the
second term.

\paragraph{The case $\max\{\eta,\kappa\}>0$.}
Interchange $\mu$ and $\nu$ if necessary so that $\nu$ has upward
spectral curvature $\kappa=\gamma>0$.
Since $\supp\nu$ is downward closed and hence connected, the final conclusion of \cref{lem:positive-curvature-support} gives
$\supp\nu=2^V$, where $V\coloneqq \{u:\nu(\{u\})>0\}$.
As in the reduction above, we may assume without loss of generality
that $V=\cU$, by pinning $\cU\setminus V$ into $\mu$
and restricting to $V$ (returning zero if the pinning is infeasible,
and evaluating the case $V=\varnothing$ directly).
The measure $\mu$ still has downward-closed support,
and $\supp\nu=2^{\cU}$.

To obtain a positive gap even when $\eta+\kappa=0$, decrease the field
of $\mu$ by setting
\begin{align*}
 t&\coloneqq 1-\frac{\eps}{16n},&
 \widetilde\mu(A)&\coloneqq t^{|A|}\mu(A).
\end{align*}
This multiplies each ratio $\mu(S+u)/\mu(S)$ by $t$ and each ratio
$\mu(S+u+v)/\mu(S)$, for distinct $u,v\notin S$, by $t^2$.
Substituting into \cref{def:spectral-curvature} with parameter
$t\eta$ gives
\[
 A_S^{\widetilde\mu,t\eta}=t^2A_S^{\mu,\eta}\succeq0.
\]
Thus $\widetilde\mu$ has upward spectral curvature $t\eta$, and
\[
 \delta\coloneqq t\eta+\kappa
 =t(\eta+\kappa)+(1-t)\kappa
 \ge\frac{\eps\gamma}{16n}.
\]
Let $\widetilde Z_\cap$ be the convolution value with $\mu$ replaced by
$\widetilde\mu$. Since $t^n\ge1-\eps/16$,
\[
 (1-\eps/16)Z_\cap\le t^n Z_\cap\le\widetilde Z_\cap\le Z_\cap.
\]
{Thus an estimate of $\widetilde Z_\cap$ with relative error
at most $\eps/4$ also estimates $Z_\cap$ with \mbox{relative error at most $\eps$}.}

Apply \cref{thm:poincare} with first factor $\widetilde\mu$ and second
factor $\nu^*(A)=\nu(\cU\setminus A)$, restricted to
$\Omega\coloneqq \supp\mu$. The second factor needs downward curvature
$\kappa$, which follows from
\[
 D_A^{\nu^*,\kappa}=A_{\cU\setminus A}^{\nu,\kappa}\succeq0.
\]
Restricting $\nu^*$ to $\Omega$ leaves these matrices unchanged,
since every deletion of a set in $\Omega$ remains in $\Omega$.
The resulting chain samples
$\rho(A)\propto\widetilde\mu(A)\nu^*(A)$, has gap at least $\delta$,
and has rates
\begin{align*}
 q(A,A+u)&=t\frac{\mu(A+u)}{\mu(A)},&
 q(A,A-u)&=\frac{\nu((\cU\setminus A)+u)}
                         {\nu(\cU\setminus A)}.
\end{align*}
Only the addition rates are multiplied by $t<1$.
Both rates are at most $B\coloneqq \chi\gamma$. Normalizing the rates by $2nB$
gives a lazy chain with gap at least
\[
 \frac{\delta}{2nB}
 \ge\frac{\eps\gamma/(16n)}{2n\chi\gamma}
 =\frac{\eps}{32\chi n^2}.
\]
The factor $\gamma$ cancels because the rate bound also contains $\gamma$.
The minimum stationary mass is at least $t^n/(2^n\cR_\mu\cR_\nu)$,
and each step costs $O(1)$.
It starts from the supported empty set. The mixing bound
\cref{eq:poincare-mixing} gives a sample with total-variation
error $\tau$ in
\begin{align*}
 O\left(\frac{n^2\chi}{\eps}
       \left[n+\log(\cR_\mu\cR_\nu)+\log(1/\tau)\right]\right)
\end{align*}
steps. {The same bounds hold for the conditional distributions
obtained by fixing coordinates.}

\Cref{lem:sampling-to-counting}, applied to the $n$
membership bits with weight $\widetilde\mu(A)\nu(\cU\setminus A)$
and starting state $\varnothing$, estimates $\widetilde Z_\cap$ to
relative error $\eps/4$, with success probability at least $3/4$, using
$O(n^3\eps^{-2}\log(2n))$ conditional samples of accuracy
$\tau=\eps/(64n)$. With fixed walk lengths, this gives total cost
\begin{align*}
 O\left(n^5\chi\eps^{-3}\log(2n)
       \left[n+\log(\cR_\mu\cR_\nu)+\log\frac{2n}{\eps}\right]\right)
\end{align*}
arithmetic operations and value-oracle calls on every execution.
The tilt bound gives the required relative estimate of $Z_\cap$.
\end{proof}

\subsection{Proofs of the Support Lemmas}\label{app:support-lemmas}
We now prove \cref{lem:positive-curvature-support,lem:lc-support-access}.

\begin{proof}[Proof of \cref{lem:positive-curvature-support}]
Fix $u\in V$. We show that $u$ can be added to every supported set
that does not already contain it. The key observation is that, along
an edge between two such sets, $u$ can be added at both endpoints or
at neither.

To see this, let $S,S+v\in\Omega$ with distinct $u,v\notin S$, and put
\[
 a=\frac{\omega(S+u)}{\omega(S)},\qquad
 b=\frac{\omega(S+v)}{\omega(S)}>0,\qquad
 c=\frac{\omega(S+u+v)}{\omega(S)}.
\]
The principal submatrix of $A_S^{\omega,\kappa}$ on $u,v$ is
\[
 M=\begin{pmatrix}a^2-\kappa a&ab-c\\ab-c&b^2-\kappa b\end{pmatrix}\succeq0.
\]
If $a=0$, the first diagonal entry is zero, so positive semidefiniteness
forces the off-diagonal entry to vanish, giving $c=0$.
If $a>0$, positive semidefiniteness also gives
\[
 0\le (b,-a)M(b,-a)^{\mathsf T}
 =ab\bigl(2c-\kappa(a+b)\bigr),
\]
so $c\ge\kappa(a+b)/2>0$.
Thus $a>0$ if and only if $c>0$, proving the observation.

Now take any $S\in\Omega$ with $u\notin S$.
Since $u\in V$ and $\Omega$ is connected, there is a path in $\Omega$
from $S$ to a set containing $u$. Stop at the first such set.
The last step adds $u$, so $u$ can be added to the preceding set.
Every earlier edge has both endpoints omitting $u$; by the observation,
working backwards along the path shows that $u$ can also be added to $S$.
Since $u\in V$ was arbitrary, $\Omega$ is upward closed in $2^V$.

If $\Omega$ is also downward closed, then $\varnothing\in\Omega$.
Upward closure now gives $\Omega=2^V$, and {the description of $V$
as the set of elements $u$ with $\omega(\{u\})>0$ follows}.
\end{proof}

\begin{proof}[Proof of \cref{lem:lc-support-access}]
Let $P\coloneqq \operatorname{conv}\{\mathbf1_A:A\in\supp\omega\}$, where
$\mathbf1_A$ is the incidence vector of $A$. This convex hull is the
\emph{Newton polytope} of $g_\omega$. Its vertices are precisely the
incidence vectors of sets in $\supp\omega$, and an \emph{edge} is a
one-dimensional face. By \cite[Proposition~10.4]{OveisGharan2020Lecture10},
every edge is parallel to a coordinate vector $\bm{e}_i$ or a difference
$\bm{e}_i-\bm{e}_j$. Consequently, a linear objective with no improving addition,
deletion, or exchange within $\supp\omega$ is globally optimal: these
moves include every incident edge. Starting at $A_0$, minimize
$|C\setminus A|+|A\cap O|$ by such moves. At most $n$ improvements,
each found with $O(n^2)$ queries, either produce a feasible set or
certify infeasibility. The convex hull of the feasible incidence vectors
is a face of $P$ and inherits its edge directions. Greedy additions and deletions
within this face therefore reach the maximum and minimum
cardinalities, visiting every intervening size in $O(n^2)$ further queries.
\end{proof}

\subsection{Optimality of the Curvature Threshold}\label{app:curvature-threshold}
\begin{proof}[Proof of \cref{rmk:curvature-threshold}]
We reduce counting proper two-colourings of a $3$-uniform hypergraph
to subset convolution. Deciding whether such a colouring exists is
$\mathsf{NP}$-complete~\cite{DinurRegevSmyth2005}.
Fix $\delta>0$ and a rational constant
$0<c<\min\{1,\delta/2\}$. For a $3$-uniform hypergraph
$\mathcal H=(\cU,\mathcal F)$ with $n\coloneqq|\cU|\ge1$, put
$a\coloneqq c/n$ and define
\begin{align*}
 \mathcal I&\coloneqq
 \{S\subseteq\cU:F\nsubseteq S\text{ for every }F\in\mathcal F\},\\
 \mu(S)=\nu(S)&\coloneqq a^{|S|}\allone[S\in\mathcal I].
\end{align*}
The two colour classes form a proper colouring exactly when both
belong to $\mathcal I$. Consequently,
\begin{align*}
 [\bm{x}^{\cU}]\bigl(g_\mu(\bm{x})g_\nu(\bm{x})\bigr)
 &=a^n\,\#\{\text{proper two-colourings of }\mathcal H\}.
\end{align*}
Since $a^n$ is known, it remains to verify the promised curvature and
oracle bounds.

Fix $S\in\mathcal I$, and let
$J_S\coloneqq\{u\in\cU\setminus S:S+u\in\mathcal I\}$ be its
feasible additions. On $J_S$, define the symmetric matrix
\begin{align*}
 B_S(u,v)&\coloneqq
 \allone[u\ne v,\ S+u+v\notin\mathcal I].
\end{align*}
Each row of $B_S$ has at most $n-1$ nonzero entries, all equal to one,
so $B_S\succeq-(n-1)I$. Substituting the weight ratios into
\cref{def:spectral-curvature} gives
\begin{align*}
 A_S^{\mu,-c}\big|_{J_S\times J_S}
 &=a^2(I+B_S)+caI
 =a^2\bigl((n+1)I+B_S\bigr)
 \succeq2a^2I.
\end{align*}
All other rows and columns vanish by downward closure. Thus both
measures have upward curvature $-c$, and we may take
$\eta=\kappa=-c$, giving $\eta+\kappa=-2c>-\delta$.
Their common support is downward closed, their rational weights
are computable in polynomial time, and
\begin{align*}
 \log(\cR_\mu\cR_\nu)&\le2n\log(1/a)
 =2n\log(n/c)=O(n\log n).
\end{align*}

An $\mathsf{FPRAS}$ on these instances would therefore give a
bounded-error polynomial-time test for hypergraph two-colourability.
By reducing the decision queries in the usual $\mathsf{SAT}$
self-reduction to this problem, amplifying their success probabilities,
and verifying the resulting assignment, we obtain an $\mathsf{RP}$
algorithm for $\mathsf{SAT}$. Hence $\mathsf{NP}=\mathsf{RP}$.
\end{proof}

\subsection{The Coordinate-Descent Bound}\label{app:coordinate-descent}
\begin{proof}[Proof of \cref{thm:stochastic-coordinate-descent}]
We average the noisy gradients and take a fixed step in any coordinate
whose empirical derivative is large. Choose
\begin{align}
 \frac{\eta}{16L}&\le s\le\frac{\eta}{8L},\label{eq:scd-step}\\
 K&\coloneqq \lceil128LH_0\eta^{-2}\rceil+1,\qquad
 m_b\coloneqq \left\lceil32\sigma^2\eta^{-2}
 \log\bigl(2d(K+1)/\delta\bigr)\right\rceil.\label{eq:balance-params}
\end{align}
Start at $\bm{x}\coloneqq \bm{x}_0$. At each iteration, average $m_b$ fresh oracle
vectors to obtain $\widehat{\bm{g}}\coloneqq m_b^{-1}\sum_{j=1}^{m_b}\bm{G}^{(j)}(\bm{x})$.
Return $\bm{x}$ if $\|\widehat{\bm{g}}\|_\infty\le\eta/2$; otherwise declare
failure if $K$ updates have already occurred, and otherwise choose
$i$ with $|\widehat g_i|>\eta/2$ and update
\begin{align*}
 \bm{x}&\leftarrow \bm{x}-s\sgn(\widehat g_i)\bm{e}_i,
 \qquad \sgn(u)\coloneqq u/|u|\quad(u\ne0).
\end{align*}
There are at most $K+1$ batches and $K$ updates, so every queried or
returned point is $\bm{x}=\bm{x}_0+s\bm{k}$ for an integer vector $\bm{k}$ satisfying
\begin{align}\label{eq:scd-field-bound}
 \|\bm{k}\|_1&\le K,\qquad
 \|\bm{x}-\bm{x}_0\|_1\le Ks\le16H_0/\eta+\eta/(4L).
\end{align}
These deterministic bounds give \cref{eq:scd-iterates}. Moreover,
$(K+1)m_b=O(Q)$ for $Q$ in \cref{eq:scd-budget}; averaging and
selecting a coordinate cost $O(dm_b)$ per batch, so $T\ge d$ gives
the total cost $O(TQ)$ on every history.

For correctness, conditional sub-Gaussian concentration gives
\begin{align*}
 \Pr{|\widehat g_i-\partial_i f(\bm{x})|>\eta/4\mid\text{past}}
 &\le2\mathrm{e}^{-m_b\eta^2/(32\sigma^2)}
 \le\frac{\delta}{d(K+1)}.
\end{align*}
Thus, with probability at least $1-\delta$, every coordinate estimate
in every batch is accurate within $\eta/4$. On this event, any returned
point satisfies $\|\nabla f(\bm{x})\|_\infty\le3\eta/4<\eta$.
Each updated coordinate has $|\partial_i f(\bm{x})|>\eta/4$ and the same
sign $\tau\coloneqq \sgn(\widehat g_i)$ as its estimate. By
\cref{eq:coordinate-smoothness} and \cref{eq:scd-step},
\begin{align}
 f(\bm{x}-s\tau\bm{e}_i)
 &\le f(\bm{x})-s|\partial_i f(\bm{x})|+Ls^2/2\notag\\
 &\le f(\bm{x})-\frac{\eta^2}{64L}+\frac{\eta^2}{128L}
 =f(\bm{x})-\frac{\eta^2}{128L}.\label{eq:descent}
\end{align}
After $K$ updates this would give
$f(\bm{x})\le f(\bm{x}_0)-K\eta^2/(128L)<f(\bm{x}_0)-H_0$, contradicting
the lower bound. Hence the algorithm returns before making $K$ updates.
\end{proof}

\subsection{Arithmetic for the Counting Algorithm}\label{app:counting-arithmetic}
The field-finding and counting algorithms use only arithmetic
operations. Ceilings of binary logarithms are computed by repeated
doubling; other logarithms in sample counts can be replaced by these
upper bounds, changing only absolute constants.
In the coordinate-descent application, $L=\sigma=1$.
Take $s=\log(1+\eta/8)$, which lies in
$[\eta/16,\eta/8]$, and maintain $z_u$ directly: an update multiplies
or divides one field coordinate $z_u$ by $1+\eta/8$. Each update costs $O(1)$
operations, and neither $s$ nor $\bm{y}$ needs to be evaluated.
The self-reduction uses membership tests, empirical frequencies, and
products and quotients of oracle values, fields, and integer powers
of $t$, so it requires no further logarithms or exponentials.

\begin{remark}[Bit complexity]\label{rem:bit-complexity}
For rational input parameters and polynomial-time value oracles returning
rationals of polynomial encoding length, the algorithm in
\cref{thm:log-concave-convolution} has an implementation whose running time
is polynomial in the input length and $\eps^{-1}$ on every execution.
Indeed, the update rule above gives $z_u=(1+1/(32n))^{k_u}$ for integers
$k_u$, and the field bound \cref{eq:balancing-bounds} implies
$\sum_u|k_u|=O(n^2B)$. Each field coordinate can therefore be stored using
polynomially many bits. The probabilities for choosing which element
to add back are ratios of sums of at most $n$ weighted oracle values,
and empirical frequencies are
ratios of sample counts. The final estimate uses products and ratios
of these numbers, fields, and powers of $t$. All these numbers have
polynomially many bits.

{We simulate the required rational probabilities using fair coin flips;
see \cite[Proposition~13]{NacuPeres2005}.}
Each binary choice can be simulated with probability at most $2^{-r}$
of requiring more than $r$ bits. If $N$ is a polynomial bound on the
total number of binary choices in the homogeneous algorithm, aborting
and returning zero when any choice needs more than
$\lceil\log_2(64N)\rceil$ bits adds at most $1/64$ to its failure
probability. The success probability proved in \cref{sec:fpras} is at
least $15/16$, so it remains above $3/4$. Conditioning, amplification,
and summation over homogeneous layers preserve polynomial bit complexity.
In particular, the rational
six-vertex instances in \cref{app:quantum} satisfy these oracle assumptions.
\end{remark}

\subsection{Fractional Log-Concavity and Spectral Curvature}\label{app:flc-curvature}
\begin{proof}[Proof of \cref{prop:flc-curvature}]
Expanding $p(\mathbf1+\bm{x})$ using its nonnegative coefficients shows
that $\supp\widehat\mu$ is downward closed.
For $S\in\supp\widehat\mu$, set $q\coloneqq \partial_Sp$, so
$q(\mathbf1)=\widehat\mu(S)>0$.
As verified below, $q$ is $\alpha$-fractionally log-concave.
The coefficient ratios in \cref{def:spectral-curvature} and the
chain rule therefore give
\begin{align*}
 A_S^{\widehat\mu,\,1-\alpha^{-1}}
 &=-\nabla^2\log q(\mathbf1)
   -(1-\alpha^{-1})\operatorname{diag}(\nabla\log q(\mathbf1))\\
 &=-\alpha^{-2}
   \left.\nabla_{\bm{x}}^2\log q(x_1^\alpha,\ldots,x_n^\alpha)
   \right|_{\bm{x}=\mathbf1}\succeq0,
\end{align*}
which proves the asserted curvature bound.

It remains to verify that nonzero coordinate derivatives preserve
fractional log-concavity. We use the limiting argument of
\cite[Section~1.5]{AlimohammadiAnariShiragurVuong2021}.
Writing $\bm{x}^\alpha\coloneqq (x_1^\alpha,\ldots,x_n^\alpha)$, multiaffinity gives
\begin{align*}
 (\partial_i p)(\bm{x}^\alpha)
 &=\lim_{s\to\infty}s^{-\alpha}p\bigl((\bm{x}+s\bm{e}_i)^\alpha\bigr).
\end{align*}
Each function on the right is a positive multiple of a translate
of $\bm{x}\mapsto p(\bm{x}^\alpha)$, hence is log-concave. The log-concavity
inequality is preserved under the limit. Iterating proves the assertion
for every nonzero $\partial_Sp$.
\end{proof}

\begin{samepage}
\subsection{Examples of Spectral Curvature}\label{app:curvature-examples}
We show that every $\kappa\in\R$ is the largest admissible curvature of
some distribution.
For a positive value, take the product measure
$\mu(S)\propto\lambda^{|S|}$ with $\lambda>0$: the curvature matrix
on coordinates outside $S$ is $\lambda(\lambda-\kappa)I$, so the
largest admissible $\kappa$ is $\lambda$. For a nonpositive value,
take $\cU=\{1,2\}$ and
\[
 \bigl(\mu(\varnothing),\mu(\{1\}),\mu(\{2\}),\mu(\{1,2\})\bigr)
 =\frac{(1,1,1,M)}{3+M},\qquad M\ge2.
\]
The smallest eigenvalue of $A_{\varnothing}^{\mu,\kappa}$ is
$2-M-\kappa$, and the singleton matrices are positive semidefinite
exactly when $\kappa\le M$. Thus the largest admissible $\kappa$ is
$2-M$, giving every nonpositive value.
\end{samepage}

\subsection{Poincar\'e Inequality for Induced Chains}\label{app:induced-poincare}
\begin{proof}[Proof of \cref{prop:induced-poincare}]
Reversing paths whose intermediate states lie outside $B$ and applying
detailed balance gives $\pi(x)r^B(x,y)=\pi(y)r^B(y,x)$, proving
reversibility with respect to $\pi_B$.
To transfer the Poincar\'e inequality, extend $f:B\to\mathbb R$
to $\Omega$ by
\[
 h(x)\coloneqq\E{f(X_{\tau_B})\mid X_0=x}.
\]
Then $h=f$ on $B$. We claim that
\[
 (\cL_rh)(x)=
 \begin{cases}
  0,&x\notin B,\\
  (\cL_{r^B}f)(x),&x\in B.
 \end{cases}
\]
Assuming these identities, we have
$\cE_r(h,h)=\pi(B)\cE_{r^B}(f,f)$, while the law of total variance
gives $\Var[\pi]{h}\ge\pi(B)\Var[\pi_B]{f}$. Applying the original
Poincar\'e inequality yields
\begin{align*}
 \pi(B)\cE_{r^B}(f,f)
 &=\cE_r(h,h)\\
 &\ge\gamma\Var[\pi]{h}
 \ge\gamma\pi(B)\Var[\pi_B]{f}.
\end{align*}
Canceling $\pi(B)>0$ proves the desired inequality.

It remains to verify the two identities. For $x\notin B$, expanding
the definition gives
\begin{align*}
 (\cL_rh)(x)
 &=\sum_{y\ne x}r(x,y)(h(x)-h(y))\\
 &=\left(\sum_{y\ne x}r(x,y)\right)h(x)
   -\sum_{y\ne x}r(x,y)h(y)=0,
\end{align*}
where the last equality is justified as follows. Put
$q_x\coloneqq\sum_{z\ne x}r(x,z)>0$, and let
$\tau\coloneqq\inf\{s>0:X_s\ne X_0\}$ be the first jump time.
Starting from $x$, the first jump goes to $y\ne x$ with probability
$r(x,y)/q_x$. Since $x\notin B$, the chain has not hit $B$ before
$\tau$. If $X_\tau=y\in B$, then $\tau_B=\tau$ and
$f(X_{\tau_B})=f(y)=h(y)$. If $y\notin B$, the strong Markov property
at $\tau$ says that the remaining trajectory has the law of the
chain started at $y$, so the conditional expectation of
$f(X_{\tau_B})$ is again $h(y)$. Consequently,
\[
 h(x)=\E{h(X_\tau)\mid X_0=x}
 =\sum_{y\ne x}\frac{r(x,y)}{q_x}h(y).
\]
Multiplying by $q_x$ gives the cancellation above.

For $x\in B$, substituting
$h(u)=\sum_{y\in B}\Pr{X_{\tau_B}=y\mid X_0=u}f(y)$ for $u\notin B$
gives
\begin{align*}
 (\cL_rh)(x)
 &=\sum_{y\in B\setminus\{x\}}r(x,y)(f(x)-f(y))
   +\sum_{u\notin B}r(x,u)(f(x)-h(u))\\
 &=\sum_{y\in B\setminus\{x\}}r^B(x,y)(f(x)-f(y))
  =(\cL_{r^B}f)(x). \qedhere
\end{align*}
\end{proof}

\section{Background for Applications}\label{app:application-background}\label{app:matroid-101}

{
We recall the background used in \cref{sec:additional-application};
see \cite{Oxley2011} for standard matroid facts.

\subsection{\texorpdfstring{{Further Matroid Terminology and Examples}}{Further Matroid Terminology and Examples}}\label{app:matroid-basics}

{We use the definitions and notation from \cref{subsec:matroids}.}
A set is \emph{spanning} if it
contains a basis. An element $e$ is a \emph{loop} if
${\rk_{\cM}}(\{e\})=0$; distinct nonloops $e,f$ are \emph{parallel} if
${\rk_{\cM}}(\{e,f\})=1$. Together with equality, this partitions the
nonloops into parallel classes.

In a \emph{partition matroid}, $\cU=\bigsqcup_i E_i$ and independence
means $|I\cap E_i|\le b_i$ for prescribed nonnegative integer bounds.
In a \emph{graphic matroid}, independent edge sets are forests and
bases are spanning forests, one tree per connected component.
In a \emph{column matroid}, independent sets index linearly
independent columns of a matrix over a field.

\subsection{Operations and Oracles}\label{app:matroid-operations}

The \emph{restriction} $\cM|_A$ keeps the independent sets contained
in $A$; deletion of $A$ is restriction to $\cU\setminus A$.
The \emph{contraction} $\cM/A$ has ground set
$\cU\setminus A$. {The ranks of the dual and contraction satisfy}
\begin{align}
 {\rk_{\cM^*}}(T)&=|T|+{\rk_{\cM}}(\cU\setminus T)-{\rk_{\cM}}(\cU),
                 \label{eq:matroid-dual-rank}\\
 {\rk_{\cM/A}}(S)&={\rk_{\cM}}(S\cup A)-{\rk_{\cM}}(A)
                 \qquad(S\subseteq\cU\setminus A).\notag
\end{align}
Thus $T$ is independent in $\cM^*$ exactly when its complement is spanning.
If $B_A$ is a basis of $\cM|_A$, contraction independence is tested by
$B_A\cup S\in\cI(\cM)$, even when $A$ is dependent.

The \emph{direct sum} of matroids on disjoint ground sets has
independent sets whose restriction to each component is independent.
The \emph{union} of matroids on a common ground set has independent
sets admitting a partition
$I=I_1\sqcup\cdots\sqcup I_k$, with $I_i\in\cI(\cM_i)$.

{The rank formulas above, together with the greedy rank
computation in \cref{subsec:matroids}, give independence oracles for
duals and contractions.}
With separate independence oracles, matroid intersection finds a
maximum common independent set in polynomial time
\cite{Edmonds1979MatroidIntersection}, and matroid union has a
polynomial-time independence oracle
\cite[Section~4.1]{BercziFrank2010Packing}.

\subsection{Submodularity and Integral Polymatroids}\label{app:polymatroid-background}

A set function $f:2^{[m]}\to\mathbb R$ is \emph{submodular} if
\[
 f(A)+f(B)\ge f(A\cup B)+f(A\cap B)
 \qquad(A,B\subseteq[m]).
\]
It is \emph{monotone} if $A\subseteq B$ implies $f(A)\le f(B)$,
and \emph{normalized} if $f(\varnothing)=0$.
{Matroid rank functions are a special case of normalized
monotone submodular functions.} For such an integer-valued $f$, the
\emph{polymatroid base polytope} is
\[
 B(f)\coloneqq\left\{\bm{x}\in\mathbb R_{\ge0}^m:
 {\sum_{i=1}^m x_i}=f([m]),\quad
 {\sum_{i\in A}x_i}\le f(A)\ \forall A\subseteq[m]\right\}.
\]
Its vertices are integral~\cite[Section~2]{FrankMurota2019BasePolyhedra}.
{For $f=\rk_{\cM}$, its integer points are exactly the
basis indicators.}

\subsection{Log-Concave and Lorentzian Polynomials}\label{app:weight-background}

{Positive element weights preserve the log-concavity of the
basis polynomial recalled in \cref{subsec:matroids}, by variable rescaling:}
\[
 \sum_{B\in\cB(\cM)}\left(\prod_{e\in B}w_e\right)\bm{x}^B
 =g_{\cM}((w_ex_e)_{e\in\cU}).
\]
Zero weights are handled by deletion; if the original rank drops,
the weighted polynomial is zero. Products on disjoint variable sets
also preserve log-concavity and multiaffinity.

A nonzero homogeneous polynomial with nonnegative coefficients is
\emph{Lorentzian} if every iterated coordinate partial derivative,
including order zero, is either zero or log-concave on the positive
orthant~\cite[Theorem~2.30]{BrandenHuh2020}.
{
For a polynomial $p(z_1,\ldots,z_m)$, choose positive integers $c_i$
with $\deg_{z_i}p\le c_i$ and new variables $x_{i,1},\ldots,x_{i,c_i}$.
Its \emph{polarization} applies the rule
\[
 z_i^k\ \longmapsto\
 \frac{1}{\binom{c_i}{k}}
 \sum_{\substack{J\subseteq[c_i]\\|J|=k}}\prod_{j\in J}x_{i,j},
 \qquad 0\le k\le c_i,
\]
to each variable in each monomial and extends linearly.
This operation preserves the Lorentzian property
\cite[Proposition~3.1]{BrandenHuh2020}.

A nonzero, finitely supported weight
$\omega:\mathbb Z_{\ge0}^m\to\mathbb R_{\ge0}$ is
\emph{log-$\mathrm M$-concave} if, for some integer $r$,
$\sum_i k_i=r$ for every $\bm{k}\in\supp\omega$, and, for every
$\bm{x},\bm{y}\in\supp\omega$ and every $i$ with
$x_i>y_i$, there is a $j$ with $x_j<y_j$ such that
\begin{align}\label{eq:integer-weight-exchange}
 \omega(\bm{x})\omega(\bm{y})
 &\le
 \omega(\bm{x}-\bm{e}_i+\bm{e}_j)
 \omega(\bm{y}+\bm{e}_i-\bm{e}_j).
\end{align}
Here $\bm{e}_i$ is the $i$th coordinate unit vector.
For such a weight $\omega$,
the polynomial $\sum_{\bm{k}}\omega(\bm{k})\bm{z}^{\bm{k}}/\bm{k}!$
is Lorentzian, where $\bm{k}!=\prod_i k_i!$
\cite[Corollary~3.16]{BrandenHuh2020}.
}
}

\clearpage
\section{Applications}\label{sec:additional-application}

We apply the subset-convolution framework to counting and approximation
problems involving graphs, determinants, and partition functions. The common step is to encode the relevant weights and
constraints in two measures that satisfy the log-concavity or
spectral-curvature hypotheses of our convolution theorems.
{The shared matroid and log-concavity preliminaries are
collected in \cref{app:application-background}.}

\Cref{fig:application-reductions} summarizes the reductions to
\cref{thm:log-concave-convolution,thm:curvature-convolution}.
The corollaries below give the precise assumptions and guarantees.
We group the applications by the reductions they share.

\begin{center}
\begin{minipage}{\textwidth}
\centering
\begin{tikzpicture}[
 >=stealth',
 yscale=1.02,
 problem/.style={draw=black!45, rounded corners=2pt, fill=black!2,
                 text width=3.3cm, inner sep=1mm,
                 align=center, font=\fontsize{9}{10.5}\selectfont},
 wide/.style={problem, text width=3.5cm},
 result/.style={problem, text width=2.35cm, fill=NavyBlue!6,
                draw=NavyBlue!65, line width=.7pt},
 group/.style={draw=black!45, dashed, rounded corners=3pt,
               line width=.5pt},
 branch/.style={draw=black!60, line width=.55pt},
 flow/.style={branch, ->}
]
 \node[problem] (orientations) at (1.75,0)
   {Strong and $k$-arc-\\connected orientations\\
    {\footnotesize\cref{cor:strong-orientations}}};
 \node[problem] (euler) at (1.75,-1.5)
   {Undirected Euler tours\\via degree-four\\$A$-trails\\
    {\footnotesize\cref{cor:euler-tours}}};
 \node[problem] (linkage) at (1.75,-3)
   {$S$--$T$ linkage reliability\\
    {\footnotesize\cref{cor:network-reliability}}};
 \node[problem] (rooted) at (1.75,-4.5)
   {Directed root-to-all\\reliability\\
    {\footnotesize\cref{cor:rooted-reliability}}};
 \node[wide] (bases) at (5.85,-.75)
   {Weighted common\\bases\\
    {\footnotesize\cref{cor:weighted-common-bases}}};
 \node[wide] (availability) at (5.85,-3.75)
   {Random sets containing\\a common basis\\
    {\footnotesize\cref{lem:common-basis-availability}}};

 \draw[group] (-.08,-9.35) rectangle (7.8,-12.8);
 \node[fill=white, inner sep=2pt, font=\fontsize{9}{10.5}\selectfont]
   at (3.85,-9.35)
   {Determinantal applications (\cref{cor:determinantal-applications})};
 \node[problem] (mixed) at (1.75,-10.35)
   {PSD mixed discriminants};
 \node[wide] (constrained) at (5.85,-10.35)
   {Determinant sums over\\matroid bases, with fields};

 \node[problem] (quantum) at (1.75,-6.35)
   {Quantum Max-Cut\\{\footnotesize(bipartite graphs)}\\
    EPR {\footnotesize(arbitrary graphs)}\\optimal values\\
    {\footnotesize\cref{cor:quantum-values}}};
 \node[wide] (sixvertex) at (5.85,-6.35)
   {Six-vertex\\partition function\\
    {for weights satisfying\\the triangle inequalities}\\
    {\footnotesize\cref{cor:six-vertex}}};

 \node[problem] (gibbs) at (1.75,-8.35)
   {Matroid congestion\\Gibbs distributions\\
    {\footnotesize\cref{cor:congestion-gibbs}}};

%
 \node[problem] (vertices) at (1.75,-13.8)
   {Vertices of network-matrix\\$0/1$ polytopes\\
    {\footnotesize\cref{cor:network-polytope-vertices}}};
 \node[wide] (flows) at (5.85,-13.8)
   {Bounded integer flows\\with unary capacities\\
    {\footnotesize\cref{cor:bounded-integer-flows}}};
 \node[problem, text width=2.65cm] (contingency) at (9.65,-13.8)
   {Contingency tables\\with log-concave\\entry weights\\
    {\footnotesize\cref{cor:weighted-contingency-tables}}};
 \node[problem, text width=2.65cm] (integerintersection) at (13.3,-13.8)
   {Inner products of\\log-$\mathrm M$-concave\\weights\\
    {\footnotesize\cref{prop:integer-weight-intersection}}};

 \node[problem] (dpps) at (1.75,-12.15)
   {Products of two DPPs};
 \node[problem, text width=2.35cm] (constraint) at (9.65,-6.35)
   {Matroid-basis\\constraint\\
    {\footnotesize\cref{lem:matroid-constraint}}};
 \node[result] (convolution) at (13.3,-6.35)
   {\textbf{Subset\\convolution}\\
    \cref{thm:log-concave-convolution}};

 \node[problem] (rainbow) at (1.75,-15.3)
   {Rainbow matchings\\
    {\footnotesize\cref{cor:rainbow-matchings}}};
 \node[problem] (forests) at (1.75,-16.8)
   {Degree-constrained\\forests\\
    {\footnotesize\cref{cor:degree-constrained-forests}}};
 \node[wide] (matroidmatchings) at (5.85,-16.05)
   {$b$-matchings independent\\in a matroid\\
    {\footnotesize all sizes; \cref{app:curvature-applications}}};
 \node[result] (curvature) at (13.3,-16.05)
   {\textbf{Subset\\convolution}\\
    \cref{thm:curvature-convolution}};

 \foreach \source in {orientations,euler}
   \draw[branch] (\source.east) -- (3.75,0 |- \source.east);
 \draw[branch] (3.75,0) -- (3.75,-1.5);
 \draw[flow] (3.75,-.75) -- (bases.west);
 \foreach \source in {linkage,rooted}
   \draw[branch] (\source.east) -- (3.75,0 |- \source.east);
 \draw[branch] (3.75,-3) -- (3.75,-4.5);
 \draw[flow] (3.75,-3.75) -- (availability.west);
 \draw[flow] (availability.north) -- (bases.south);
 \foreach \source/\destination in
   {quantum/sixvertex,mixed/constrained,vertices/flows}
   \draw[flow] (\source.east) -- (\destination.west);
 \foreach \source in {bases,constrained,sixvertex,gibbs}
   \draw[branch] (\source.east) -- (8.05,0 |- \source.east);
 \draw[branch] (8.05,-.75) -- (8.05,-10.35);
 \draw[flow] (8.05,-6.35) -- (constraint.west);
 \draw[branch] (constraint.east) -- (11.5,-6.35);
 \draw[branch] (dpps.east) -- (11.5,0 |- dpps.east);
 \draw[branch] (11.5,-6.35) -- (11.5,-12.15);
 \draw[flow] (11.5,-6.35) -- (convolution.west);
 \draw[flow] (flows.east) -- (contingency.west);
 \draw[flow] (contingency.east) -- (integerintersection.west);
 \draw[flow] (integerintersection.north) -- (convolution.south);
 \foreach \source in {rainbow,forests}
   \draw[branch] (\source.east) -- (3.75,0 |- \source.east);
 \draw[branch] (3.75,-15.3) -- (3.75,-16.8);
 \draw[flow] (3.75,-16.05) -- (matroidmatchings.west);
 \draw[flow] (matroidmatchings.east) -- (curvature.west);
\end{tikzpicture}
\captionsetup{hypcap=false}

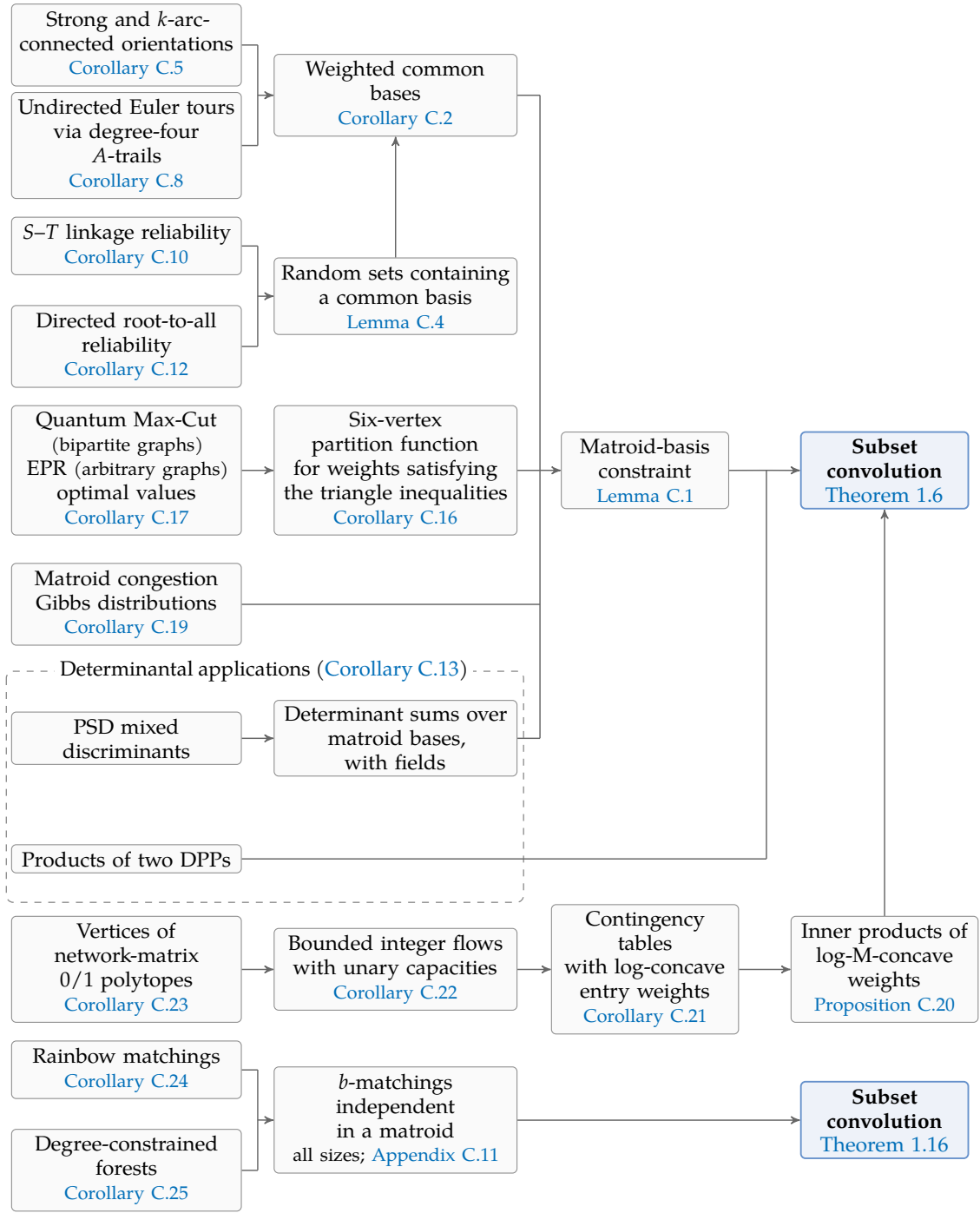
\captionof{figure}{Reduction diagram for the applications of
\cref{thm:log-concave-convolution,thm:curvature-convolution}.
Arrows point toward the problem or result used by each reduction.
Blue boxes highlight the two convolution theorems.
Some intermediate steps are omitted.}
\label{fig:application-reductions}
\end{minipage}
\end{center}
\clearpage

Unless stated otherwise, running-time statements below use the arithmetic model of
\cref{thm:log-concave-convolution,thm:curvature-convolution};
we omit bit-complexity analysis for simplicity.
The EPR-in-$\mathsf{BPP}$ conclusion still uses the
bit-complexity bound in \cref{rem:bit-complexity}.
We first record two reductions used throughout the appendix.
Products of polynomially many rational input weights, determinants
of rational input matrices, and factorials of polynomially bounded
integers have polynomial bit length. We use this observation below
to bound the logarithmic ranges of positive coefficients.

\begin{lemma}\label{lem:matroid-constraint}
Let $\omega:2^{\cU}\to\mathbb R_{\ge0}$ have a log-concave generating
polynomial, an evaluation oracle, a given set in $\supp\omega$, and
a supplied bound on $\cR_\omega$, as in
\cref{thm:log-concave-convolution}.
For any matroid $\cM$ on $\cU$ given by an independence oracle, there
is an $\mathsf{FPRAS}$ for
\[
 \sum_{S\in\cB(\cM)}\omega(S),
\]
with running time polynomial in $|\cU|$, $\log\cR_\omega$, and the
inverse accuracy, counting oracle calls as single operations.
\end{lemma}

\begin{proof}
Set $\nu(T)\coloneqq\allone[\cU\setminus T\in\cB(\cM)]$.
Its generating polynomial is the basis polynomial of $\cM^*$,
hence is log-concave, and $\cR_\nu=1$.
The complement of a greedy basis initializes $\nu$, and
$[\bm{x}^{\cU}](g_\omega g_\nu)$ is the desired sum.
Apply \cref{thm:log-concave-convolution}.
\end{proof}

\begin{corollary}\label{cor:weighted-common-bases}
Let $\cM,\cN$ be matroids of the same rank $r$ on $\cU$, given by
independence oracles, and let $w_e$ be nonnegative rational element
weights. There is an $\mathsf{FPRAS}$ for
\[
 \sum_{S\in\cB(\cM)\cap\cB(\cN)}\prod_{e\in S}w_e.
\]
\end{corollary}

\begin{proof}
Delete zero-weight elements; if either rank drops below $r$, return
zero. The weight
$\omega(S)\coloneqq\allone[S\in\cB(\cM)]\prod_{e\in S}w_e$
has a log-concave generating polynomial by positive variable rescaling.
A greedy basis initializes it, and its logarithmic coefficient range
is polynomial in the weight encoding length.
Apply \cref{lem:matroid-constraint} with the matroid $\cN$.
\end{proof}

To count random subsets that contain a common basis, we use
the following consequence of Rado's theorem~\cite{Rado1942}.

\begin{theorem}[Rado matroids {\cite[Theorem~11.2.12]{Oxley2011}}]\label{thm:rado-matroid}
Let $G$ be a finite bipartite graph with left vertex set $L$ and right
vertex set $R$, and let $\mathcal R$ be a matroid on $R$.
There is a matroid $\mathcal L$ on $L$ such that $A\subseteq L$ is an independent set of $\mathcal L$
exactly when $G$ has a matching whose left endpoints are exactly $A$
and whose right endpoints form an independent set of $\mathcal R$.
\end{theorem}

\begin{lemma}\label{lem:common-basis-availability}
Let $\cM,\cN$ be matroids of the same rank on $\cU$, given by
independence oracles. Retain each element $u$ independently with
rational probability $p_u\in[0,1]$. There is an $\mathsf{FPRAS}$ for
the probability that the retained set contains a common basis.
\end{lemma}

\begin{proof}
We encode the event by weighted common bases of two new matroids.
Put $m\coloneqq|\cU|$, and let
$\cU^+=\{u^+:u\in\cU\}$ and $\cU^-=\{u^-:u\in\cU\}$ be two
disjoint copies of $\cU$. Let $\mathcal H$ be the direct sum of copies
of $\cM$ and $\cN^*$ on $\cU^+$ and $\cU^-$, respectively; it has rank $m$.
Form a bipartite graph with left vertex set
$\mathcal X\coloneqq\{x_u^0,x_u^1:u\in\cU\}$ and right vertex set
$\cU^+\cup\cU^-$. The new elements $x_u^0,x_u^1$ represent $u$ being
absent or retained, respectively. Join $x_u^0$ only to $u^-$, and join
$x_u^1$ to both $u^-$ and $u^+$.

Apply \cref{thm:rado-matroid} with $L=\mathcal X$,
$R=\cU^+\cup\cU^-$, and $\mathcal R=\mathcal H$, and denote the
resulting matroid $\mathcal L$ by $\mathcal Q$.
A subset is independent in $\mathcal Q$
exactly when it can be matched to an independent set of $\mathcal H$.
Let $\mathcal P$ be the partition
matroid on $\mathcal X$ allowing at most one of $x_u^0,x_u^1$ for each $u$.
Both matroids have rank $m$: any basis of $\mathcal H$ can be matched
to distinct left vertices by matching $u^+$ to $x_u^1$ and $u^-$ to $x_u^0$.

The retained set $F\subseteq\cU$ has the unique representation
$X_F\coloneqq\{x_u^1:u\in F\}\cup\{x_u^0:u\notin F\}$.
Since $|X_F|=m={\rk}(\mathcal H)$, $X_F$ is independent
in $\mathcal Q$ exactly when it can be matched to a basis of $\mathcal H$.
A matching covering $X_F$ uses {exactly one of $u^+$ and $u^-$ for each
element $u$}, so its matched right vertices form a set
$B^+\cup(\cU\setminus B)^-$ with
$B\subseteq F$. This set is a basis of $\mathcal H$ exactly when
$B$ is a common basis of $\cM,\cN$. Consequently,
\[
 X_F\in\cB(\mathcal P)\cap\cB(\mathcal Q)
 \quad\Longleftrightarrow\quad
 F\text{ contains a common basis of }\cM,\cN.
\]
Giving $x_u^1,x_u^0$ weights $p_u,1-p_u$ makes the weighted
common-bases sum the required probability. Each $F$ is counted once,
regardless of the number of common bases it contains.

It remains to implement an independence oracle for $\mathcal Q$.
We do so by matroid intersection on restrictions and contractions
of $\cM$ and $\cN$. For a queried set $Y\subseteq\mathcal X$, let $S_2$
contain the elements $u$ for which both $x_u^0,x_u^1$ belong to $Y$,
let $S_0$ contain those for which neither belongs to $Y$, and let
$S_1\coloneqq\{u:x_u^1\in Y,\ x_u^0\notin Y\}$.
The possible sets of right endpoints of a matching with left endpoint set $Y$ are exactly
\[
 (S_2\cup B)^+\cup\bigl(\cU\setminus(S_0\cup B)\bigr)^-,
 \qquad B\subseteq S_1.
\]
Here $B$ consists of the elements $u\in S_1$ for which $x_u^1$ is matched to $u^+\in\cU^+$.
By duality, the displayed set is independent in $\mathcal H$ exactly when
$S_2\cup B$ is independent in $\cM$ and $S_0\cup B$ contains a basis of $\cN$.

Reject if $S_2$ is dependent in $\cM$. Otherwise, apply matroid intersection
to $(\cM/S_2)|_{S_1}$ and $(\cN/S_0)|_{S_1}$, and accept exactly when they have a
common independent set of size $k\coloneqq {\rk_{\cN}}(\cU)-{\rk_{\cN}}(S_0)$,
where ${\rk_{\cN}}$ is the rank function of $\cN$.
The required oracles use only the original independence oracles:
test $I\subseteq S_1$ in $\cM/S_2$ by testing $S_2\cup I$ in $\cM$, and
test it in $\cN/S_0$ by testing $B_0\cup I$ in $\cN$, where $B_0$ is a
basis of $\cN|_{S_0}$ found greedily. This gives a polynomial-time oracle for
$\mathcal Q$, so
\cref{cor:weighted-common-bases} applies.
\end{proof}

\subsection{Strong Orientations}\label{app:strong-orientations}

Let $G=(V,E)$ be a loopless undirected multigraph, with
$|V|\ge2$, and let $k\ge1$ be an integer. An orientation {of the edges of $G$} is
\emph{$k$-arc-connected} if {for every nonempty proper subset $X\subsetneq V$,
there are at least $k$ arcs directed from $V\setminus X$ into $X$}.
Give each direction of each labeled edge a
nonnegative rational weight $w_{e,v}$, where $v$ is its head. Define
\[
 Z_k(G;\bm{w})\coloneqq
 \sum_{D\text{ a }k\text{-arc-connected orientation}}
       \prod_{e\in E}w_{e,\operatorname{head}_D(e)}.
\]
The case $k=1$ counts strong orientations, and zero weights allow fixed
directions.

\begin{corollary}\label{cor:strong-orientations}
There is an $\mathsf{FPRAS}$ for $Z_k(G;\bm{w})$, even when $k$
is part of the input in binary.
\end{corollary}

{
The following form of the classical common-bases representation
\cite[Lecture~16]{Goemans2012Orientations} follows from the indegree
characterization in \cite{FrankMurota2019BasePolyhedra} and the matroid
construction of McDiarmid~\cite{McDiarmid1973Independence}; see also
\cite[Theorem~2.1]{BoninChunFife2023NaturalMatroid}.

\begin{lemma}[Common bases for strong orientations]\label{lem:orientation-common-bases}
Let $G=(V,E)$ be a loopless undirected multigraph with $|V|\ge2$, and
let $k\ge1$ be an integer. Suppose every nontrivial undirected cut
has size at least $2k$. Put $m\coloneqq|E|$, and define
\[
 \cU\coloneqq\{(e,v):e\in E,\ v\in e\},
 \qquad \cU_X\coloneqq\{(e,v)\in\cU:v\in X\}.
\]
For an orientation $D$, let $d^-_D(v)$ be the number of edges directed
into $v$. The function
\[
 f(X)\coloneqq\max_{D\text{ a }k\text{-arc-connected orientation}}
                    \sum_{v\in X}d^-_D(v)
 \qquad(X\subseteq V)
\]
is integer-valued, monotone, and submodular, with $f(\emptyset)=0$
and $f(V)=m$. Let $\cM$ be the partition matroid on $\cU$ allowing
at most one incidence per edge. There is a matroid $\cN$ on $\cU$
with rank function
\[
 {\rk_{\cN}}(Y)=\min_{X\subseteq V}\bigl(f(X)+|Y\setminus\cU_X|\bigr)
 \qquad(Y\subseteq\cU).
\]
Both matroids have rank $m$, and the map
\[
 D\longmapsto A_D\coloneqq\{(e,\operatorname{head}_D(e)):e\in E\}
\]
is a bijection from the $k$-arc-connected orientations of $G$ to
$\cB(\cM)\cap\cB(\cN)$.
\end{lemma}

\begin{proof}[Proof of \cref{cor:strong-orientations}]
We apply \cref{cor:weighted-common-bases} to the matroids in
\cref{lem:orientation-common-bases}. First, by the orientation theorem
of Nash-Williams~\cite[Theorem~1]{NashWilliams1960Orientations},
a $k$-arc-connected orientation exists exactly when every nontrivial
undirected cut has size at least $2k$. Check this by a minimum cut,
returning zero if it fails, and put $m\coloneqq|E|$.

Give incidence $(e,v)$ weight $w_{e,v}$. The bijection in
\cref{lem:orientation-common-bases} preserves the product of weights,
so the weighted common-bases sum is exactly $Z_k(G;\bm{w})$.
It remains to implement the two independence oracles.
For $\cM$, check that the queried set contains at most one incidence
per edge.

For $\cN$, given any $Y\subseteq\cU$, compute ${\rk_{\cN}}(Y)$ using the
lemma's rank formula and accept exactly when ${\rk_{\cN}}(Y)=|Y|$.
For fixed $Y$, the term
\[
 |Y\setminus\cU_X|
   =|Y|-\sum_{v\in X}|Y\cap\cU_{\{v\}}|
\]
is a constant minus a sum of fixed vertex weights. Adding it to $f(X)$
therefore preserves submodularity, so the rank formula can be evaluated
by submodular minimization~\cite{Schrijver2000SubmodularMinimization}.
To supply the required values $f(X)$, assign cost zero to directions
whose heads lie in $X$ and cost one to the others. The minimum cost
of a $k$-arc-connected orientation is $m-f(X)$, and it can be found
in polynomial time~\cite{IwataKobayashi2005Orientations}.
Thus every value query, and hence each independence query, takes
polynomial time.

With these oracles, \cref{cor:weighted-common-bases} gives the claimed
$\mathsf{FPRAS}$. If $k$ is part of the input in binary, the successful
minimum-cut test bounds $k$ by $m$, so the oracle implementations
still take polynomial time in the input length.
\end{proof}
}

\subsection{Undirected Euler Tours}\label{app:euler-tours}

Let $G$ be a connected $4$-regular multigraph.
{
Each edge has two ends, called \emph{half-edges}; a loop contributes
two distinct half-edges at the same vertex. A \emph{transition} at a
vertex pairs its four incident half-edges into two pairs; there are
three choices. Entering along either half-edge of a pair is followed
by leaving along the other. A \emph{closed trail} is a closed walk that does not
repeat edges. Choosing a transition at every vertex partitions the
edges into closed trails. We count the choices that give a single
closed trail, or equivalently Euler tours up to cyclic shift and
reversal.}
{
For each vertex $v$ and each of its three transitions $\tau$, give
the whole pairing $\tau$ a nonnegative rational weight $w_v(\tau)$.
These three weights may be chosen arbitrarily at each vertex.
The total weight of Euler tours is
\[
 \sum_P\prod_{v\in V(G)}w_v(P_v),
\]
where the sum ranges over all Euler tours $P$, and $P_v$ is the
transition used by $P$ at $v$.
Thus the weights are attached to local pairings, not to individual
edge directions.
\par
}

{
\begin{corollary}[Weighted Euler tours]\label{cor:weighted-euler-tours}
Let $G$ be a connected $4$-regular multigraph with nonnegative rational
weights on its local transitions. There is an $\mathsf{FPRAS}$ for the
total weight of Euler tours in $G$.
\end{corollary}
}

\begin{proof}
{
We encode the weighted Euler tours in $G$ as weighted common bases
of two matroids and apply \cref{cor:weighted-common-bases}.
}
{
Fix an Euler tour $C$, viewed as a cyclic sequence of vertex visits.
Each vertex occurs twice because $G$ is $4$-regular. Two distinct
vertices $u,v$ \emph{alternate along $C$} if, after ignoring visits to
all other vertices, their cyclic order is $u,v,u,v$.
The \emph{interlace matrix} $A\in\mathbb F_2^{V(G)\times V(G)}$ is
defined by $A_{uv}=1$ when $u$ and $v$ alternate, and $A_{uv}=0$
otherwise; in particular, $A_{uu}=0$~\cite[p.~2]{Traldi2009BinaryNullity}.
}
{
Choose a direction for traversing $C$ and orient every edge accordingly.
At a vertex $v$, label the two incoming half-edges $a_1,a_2$ and the
two outgoing half-edges $b_1,b_2$ so that $C$ enters through $a_i$ and
leaves through $b_i$ on its $i$th visit, for $i=1,2$. Define
\[
\begin{aligned}
 \phi_v&\coloneqq\bigl\{\{a_1,b_1\},\{a_2,b_2\}\bigr\},\\
 \chi_v&\coloneqq\bigl\{\{a_1,b_2\},\{a_2,b_1\}\bigr\},\\
 \psi_v&\coloneqq\bigl\{\{a_1,a_2\},\{b_1,b_2\}\bigr\}.
\end{aligned}
\]
Only $\phi_v$ is used by $C$. The alternative $\chi_v$ also pairs each
incoming half-edge with an outgoing half-edge, but switches which
outgoing half-edge follows each incoming one. The transition $\psi_v$
pairs the two incoming half-edges together and the two outgoing
half-edges together.
}
Associate these choices with the labeled columns
\[
 \phi_v\longleftrightarrow \bm{e}_v,\qquad
 \chi_v\longleftrightarrow \bm{A}_v,\qquad
 \psi_v\longleftrightarrow \bm{e}_v+\bm{A}_v
 \quad\text{of }[\,I\mid A\mid I+A\,]\text{ over }\mathbb F_2.
\]
For a transition choice $P$, let $K_P$ contain its selected columns.
The extended Cohn--Lempel equality
\cite[Theorem~4]{Traldi2009BinaryNullity} gives
\[
 {\#\{\text{closed trails of }P\}}=1+\operatorname{nullity}K_P.
\]
Indeed, deleting the identity choices leaves the principal interlace
matrix with a diagonal one at each $\psi$ choice, which is the matrix
in that equality. Hence $P$ is an Euler tour precisely when it is a
common basis of the column matroid and the partition matroid choosing
one column from each vertex triple. Giving each column its transition
weight preserves the weight and counts each tour once. {Gaussian
elimination and checking that at most one column is selected from each
vertex triple give the independence oracles}, so
\cref{cor:weighted-common-bases} applies.
\end{proof}

{
Ge and \v{S}tefankovi\v{c}
\cite[Theorem~4]{GeStefankovic2010EulerTours} give an
approximation-preserving reduction from counting Euler tours in
general undirected graphs to counting $A$-trails in $4$-regular graphs
with prescribed cyclic orders.

In a $4$-regular multigraph, suppose that a cyclic order of the four
half-edges is specified separately at each vertex. An \emph{$A$-trail} is
an Euler tour such that, at every visit to a vertex, the entering and
leaving half-edges are adjacent in that vertex's cyclic order.
Either direction through an adjacent pair is allowed. For cyclic
order $h_1,h_2,h_3,h_4$, the allowed arrival--departure pairs are the
unordered pairs
\[
 \{h_1,h_2\},\qquad \{h_2,h_3\},\qquad
 \{h_3,h_4\},\qquad \{h_4,h_1\}.
\]
For example, entering through $h_1$ and leaving through $h_2$ is
allowed, as is entering through $h_2$ and leaving through $h_1$.
The transition at that vertex must therefore be one of
\[
 \bigl\{\{h_1,h_2\},\{h_3,h_4\}\bigr\}
 \qquad\text{or}\qquad
 \bigl\{\{h_1,h_4\},\{h_2,h_3\}\bigr\}.
\]
The remaining pairing $\{\{h_1,h_3\},\{h_2,h_4\}\}$ is forbidden.
}

{
Counting $A$-trails is a special case of weighted Euler-tour counting:
assign weight one to each allowed transition and zero to each forbidden
transition. Every $A$-trail then has weight one, and every other Euler
tour has weight zero. By \cref{cor:weighted-euler-tours}, this count admits an
$\mathsf{FPRAS}$. Composing that algorithm with the reduction gives
the following corollary.
\par
}

\begin{corollary}[Euler tours in general graphs]\label{cor:euler-tours}
There is an $\mathsf{FPRAS}$ for the unweighted number of Euler tours
in general undirected graphs.
\end{corollary}

\subsection{Network Reliability and Disjoint Linkages}\label{app:reliability-proof}

\begin{definition}[$S$--$T$ linkage]\label{def:st-linkage}
Let $G = (V,E)$ be a directed or undirected graph, and
let $S,T\subseteq V$ be disjoint terminal sets with $|S|=|T|=k\ge1$.
An \emph{$S$--$T$ linkage} consists of $k$ paths that share no vertices.
Each path joins a vertex of $S$ to a vertex of $T$, and every vertex of
$S$ and $T$ is used as an endpoint exactly once. We do not specify which
vertex of $S$ is joined to which vertex of $T$. If $G$ is directed,
all paths must follow edge directions from $S$ to $T$.
\end{definition}

\begin{corollary}[$S$--$T$ Network reliability]\label{cor:network-reliability}
Let $G=(V,E)$ be a directed or undirected graph, and let $S,T\subseteq V$
be disjoint terminal sets with $|S|=|T|=k\ge1$. Retain each edge $e$
independently with given rational probability $p_e\in[0,1]$.
There is an $\mathsf{FPRAS}$ for
\begin{align}\label{eq:linkage-reliability}
 \operatorname{Rel}_{S,T}(G;\bm{p})
 &\coloneqq \sum_{\substack{F\subseteq E:\\ (V,F)\text{ admits an }S\text{--}T\text{ linkage}}}
       \prod_{e\in F}p_e\prod_{e\in E\setminus F}(1-p_e).
\end{align}
\end{corollary}
\begin{remark}
Taking $S=\{s\}$ and $T=\{t\}$ in an undirected graph recovers
the $\mathsf{FPRAS}$ for $s$--$t$ reliability in \cite{FengFuGuo2026} with worse running time.
\end{remark}

\begin{proof}[Proof of \cref{cor:network-reliability}]
We first encode directed linkage as the existence of a perfect matching
in a random bipartite graph, then apply
\cref{lem:common-basis-availability}. An undirected edge gadget gives
the undirected case.

For a directed graph, discard loops, arcs entering $S$, and arcs leaving
$T$. Form a bipartite graph with vertex classes
\begin{align*}
 V_L&\coloneqq \{\ell_v:v\in V\setminus T\},&
 V_R&\coloneqq \{r_v:v\in V\setminus S\}.
\end{align*}
{Replace each remaining arc $u\to v$ by the edge $\ell_ur_v$.
Keep this edge with the same probability as the original arc.}
Add an always-present edge $\ell_vr_v$ for
each $v\notin S\cup T$. A linkage gives a perfect matching by selecting
its path edges and the diagonal edges at unused vertices. Conversely,
following a perfect matching from each source reaches a distinct sink:
entering a cycle or another source's path would give a vertex two selected
incoming edges. Thus, for every {retained edge set},
\begin{align*}
 \text{the directed graph admits an $S$--$T$ linkage}
 &\quad\Longleftrightarrow\quad
 \text{the bipartite graph has a perfect matching}.
\end{align*}

Let $\cM,\cN$ be the partition matroids on the bipartite edge set,
allowing at most one edge at each left and right endpoint, respectively.
If the full bipartite graph has no perfect matching, return zero.
Otherwise both matroids have rank $|V|-k$, and a retained edge set
contains a perfect matching exactly when it contains a common basis.
Thus \cref{lem:common-basis-availability} approximates the desired
probability, including the always-present diagonal edges.

Finally, for an undirected edge $e=\{u,v\}$, introduce new vertices
$a_e,b_e$ and arcs
\begin{align*}
 u\to a_e,\quad v\to a_e,\quad a_e\to b_e,\quad
 b_e\to u,\quad b_e\to v.
\end{align*}
Only $a_e\to b_e$ is random, with retention probability $p_e$; the
other four arcs are always present. Vertex-disjoint undirected paths
lift to vertex-disjoint directed paths. Conversely, simple directed
paths project to vertex-disjoint undirected paths, since {returning
through the gadget to the vertex where the path entered would repeat
that vertex}. This preserves
the linkage event for every {choice of retained original edges} and uses one independent
random variable per original edge. The directed result therefore also
proves the undirected case.
\end{proof}

\subsection{Directed Root-to-All Reliability}\label{app:rooted-reliability}

Let $D=(V,E)$ be a directed multigraph with root $r$, and retain each
arc $e$ independently with
rational probability $p_e\in[0,1]$.
{
For each $F\subseteq E$, let $\Lambda(V,F)\in\mathbb R^{V\times V}$
be the matrix whose entry $\Lambda(V,F)_{u,v}$, for $u\ne v$, is the
maximum number of arc-disjoint directed paths from $u$ to $v$ in
$(V,F)$. Set its diagonal entries to zero.
}
For an integer $k\ge1$, define
\[
 R_k(D,r;\bm{p})\coloneqq
 \sum_{\substack{F\subseteq E:\\
        {\Lambda(V,F)_{r,v}}\ge k\\ \forall v\ne r}}
       \prod_{e\in F}p_e\prod_{e\notin F}(1-p_e).
\]
The paths to different target vertices need
not be mutually disjoint. In particular, $k=1$ asks that the root
reach every vertex. Reversing the arcs gives the all-to-root version.

\begin{corollary}\label{cor:rooted-reliability}
There is an $\mathsf{FPRAS}$ for $R_k(D,r;\bm{p})$, even when $k$
is part of the input in binary.
\end{corollary}

\begin{proof}
We express the event as containment of a common basis and apply
\cref{lem:common-basis-availability}.
Delete loops and arcs entering $r$, which do not affect this event.
The one-vertex case has probability one. Otherwise put
$q\coloneqq k(|V|-1)$, and return zero if $q>|E|$.
Let $\cM$ be the union of $k$ copies of the graphic matroid of the
underlying undirected multigraph.
{A set is independent in $\cM$ exactly when it can be
partitioned into $k$ forests after ignoring the arc directions.}
Let $\cN$ be the partition
matroid allowing at most $k$ arcs entering each nonroot vertex.
{The polynomial-time matroid-union algorithm
\cite[Section~4.1]{BercziFrank2010Packing} gives an independence oracle
for $\cM$, using cycle detection to test independence in each graphic
matroid. For $\cN$, we count the queried arcs entering each nonroot
vertex and check that every count is at most $k$.}
If either rank is less than $q$, the event is impossible; otherwise
both ranks equal $q$.

{
For each $F\subseteq E$, we have
\[
 \Lambda(V,F)_{r,v}\ge k\ \text{for all }v\ne r
 \quad\Longleftrightarrow\quad
 F\text{ contains a common basis of }\cM,\cN.
\]
Edmonds's arborescence-packing theorem
\cite{Edmonds1973EdgeDisjointBranchings}
(see also \cite[Theorem~1.1]{BercziFrank2010Packing}) equates the
left-hand condition with the existence in $F$ of $k$ arc-disjoint
spanning out-arborescences rooted at $r$, that is, directed spanning
trees in which $r$ reaches every vertex. The standard matroid
representation \cite[Section~4.1]{BercziFrank2010Packing} identifies
the unions of $k$ such arborescences with the common bases of
$\cM,\cN$.
\par
}

Applying \cref{lem:common-basis-availability} now gives the claimed
approximation. The preliminary bound $k(|V|-1)\le|E|$ ensures
polynomial size even when $k$ is given in binary.
\end{proof}

\subsection{Mixed Discriminants and Determinantal Point Processes}\label{app:determinantal-applications}

{The matrices defining the counting problems below are
rational and positive semidefinite.}
For a matrix $L$ indexed by $\cU$, write $L[S]$ for its principal
submatrix on $S$, with {the convention} $\det L[\varnothing]=1$.
The determinantal point process with matrix $L$ assigns probability
$\det L[S]/\det(I+L)$ to $S\subseteq\cU$.

\begin{corollary}\label{cor:determinantal-applications}
There is an $\mathsf{FPRAS}$ for each of the following quantities:
\begin{enumerate}
 \item For a rational positive semidefinite matrix $L$ indexed by
 $\cU$ and a matroid $\cM$ on $\cU$ given by an independence oracle,
 \[
  Z_{\cM}(L)\coloneqq\sum_{S\in\cB(\cM)}\det L[S].
 \]
 \item For rational positive semidefinite matrices $A,B$ indexed by
 the same set $\cU$,
 \[
  Z(A,B)\coloneqq\sum_{S\subseteq\cU}\det A[S]\det B[S].
 \]
 \item For rational positive semidefinite $d\times d$ matrices
 $A_1,\ldots,A_d$, the mixed discriminant
 \[
  \Delta(A_1,\ldots,A_d)
  \coloneqq[x_1\cdots x_d]\det\left(\sum_{i=1}^d x_iA_i\right).
 \]
\end{enumerate}
\end{corollary}

{
Before proving \cref{cor:determinantal-applications}, we recall the
following standard determinant identity.
\begin{proposition}[Cauchy--Binet formula {\cite[Theorem~9.4]{Stanley2013Topics}}]\label{prop:cauchy-binet}
Let $A\in\mathbb R^{d\times n}$ and $B\in\mathbb R^{n\times d}$,
where $n\ge d$. For $S\subseteq[n]$ with $|S|=d$, let $A_{:,S}$
be the $d\times d$ matrix formed by the columns of $A$ indexed by
$S$, and let $B_{S,:}$ be the $d\times d$ matrix formed by the
rows of $B$ indexed by $S$, both in increasing order. Then
\[
 \det(AB)=\sum_{\substack{S\subseteq[n]\\|S|=d}}
                 \det(A_{:,S})\det(B_{S,:}).
\]
\end{proposition}
}

\begin{proof}[Proof of \cref{cor:determinantal-applications}]
The first two reductions use the principal-minor identities
{(see \cite[Section~3.2]{BorceaBrandenLiggett2009})}
\begin{align*}
 p_L(\bm{x})\coloneqq\sum_S\det L[S]\bm{x}^S
   &={\det(I+LD_{\bm{x}})=}\det(I+L^{1/2}D_{\bm{x}}L^{1/2}),\\
 q_L(\bm{x})\coloneqq\sum_T\det L[\cU\setminus T]\bm{x}^T
   &=\det(L+D_{\bm{x}}),
\end{align*}
where $D_{\bm{x}}\coloneqq\operatorname{diag}(\bm{x})$.
{
To obtain the first expansion, use multilinearity of the determinant
in its columns: choosing columns from $LD_{\bm{x}}$ at indices in $S$ and
from $I$ elsewhere contributes $\det L[S]\bm{x}^S$.
Similarly, choosing columns from $D_{\bm{x}}$ at indices in $T$ and from
$L$ elsewhere contributes $\det L[\cU\setminus T]\bm{x}^T$, giving the
second expansion.
The last equality for $p_L$ follows from Sylvester's identity
$\det(I+AB)=\det(I+BA)$.
}

{
Both polynomials are log-concave. Indeed, the matrices
$I+L^{1/2}D_{\bm{x}}L^{1/2}$ and $L+D_{\bm{x}}$ are affine functions of $\bm{x}$
and positive definite when $\bm{x}>\mathbf0$, even if $L$ is singular.
The function $X\mapsto\log\det X$ is concave on positive definite
matrices \cite[Section~3.1.5, p.~74]{BoydVandenberghe2004}.
Composing a concave function with an affine map preserves concavity,
so $\log p_L(\bm{x})$ and $\log q_L(\bm{x})$ are concave.
\par
}

For Item~1, apply \cref{lem:matroid-constraint} with
$\omega(S)\coloneqq\det L[S]$, initialized at $\varnothing$.

For Item~2, use
\[
 Z(A,B)=[\bm{x}^{\cU}]\bigl(p_A(\bm{x})q_B(\bm{x})\bigr)
\]
and apply \cref{thm:log-concave-convolution}, with initial sets
$\varnothing$ and $\cU$. This is the normalizing constant for the
product of two DPPs \cite{OhsakaMatsuoka2021}; dividing by
$\det(I+A)\det(I+B)$ gives their agreement probability.

For Item~3, use the partition-matroid reduction of
\cite[Section~4.1]{CelisEtAl2017}. Rational $LDL^{\mathsf T}$
decompositions give
$A_i=\sum_{e\in E_i}c_e\bm{v}_e\bm{v}_e^{\mathsf T}$ with $c_e>0$
and $|E_i|\le d$. {Here $E_i$ indexes the rank-one terms
of $A_i$; we take these index sets to be pairwise disjoint.
If $E_i=\varnothing$ for some $i$, then $A_i=0$ and $\Delta=0$.
Otherwise put $E\coloneqq\bigcup_{i=1}^d E_i$, let $V$ have columns
$\bm{v}_e$, and let $\cM$ be the partition matroid on $E$ whose bases contain exactly one
element from each $E_i$. Thus a basis selects one rank-one term
from each matrix $A_i$.}

{
To compute the coefficient defining $\Delta$, let $D(\bm{x})$ be the
diagonal matrix indexed by $E$ with $D(\bm{x})_{e,e}=c_ex_i$ for
$e\in E_i$. Then $\sum_{i=1}^d x_iA_i=VD(\bm{x})V^{\mathsf T}$.
For $S\subseteq E$ with $|S|=d$, write $V_S$ for the $d\times d$
matrix formed by the columns of $V$ indexed by $S$.
Applying \cref{prop:cauchy-binet} with $A=VD(\bm{x})$ and
$B=V^{\mathsf T}$ gives
\[
 \det\left(\sum_{i=1}^d x_iA_i\right)
 =\sum_{\substack{S\subseteq E\\|S|=d}}
   \det(V_S)^2\left(\prod_{e\in S}c_e\right)
   \prod_{i=1}^d x_i^{|S\cap E_i|}.
\]
Here $\det(V_S)^2=\det(V_S^{\mathsf T}V_S)
=\det((V^{\mathsf T}V)[S])$.
The monomial $\prod_{i=1}^d x_i^{|S\cap E_i|}$ equals
$x_1\cdots x_d$ exactly when $S$ contains one element of each
$E_i$, that is, when $S$ is a basis of $\cM$.
Extracting this coefficient therefore gives
}
\[
 \Delta(A_1,\ldots,A_d)
 =\sum_{S\in\cB(\cM)}\det((V^{\mathsf T}V)[S])\prod_{e\in S}c_e.
\]
Apply \cref{lem:matroid-constraint} to this determinant weight:
the factors $c_e$ are external fields, so log-concavity persists,
and $\varnothing$ is again an initial set.
All coefficient queries and decompositions use rational arithmetic
of polynomial bit length. The square roots in $p_L$ serve only to
prove log-concavity.
\end{proof}

\subsection{The Six-Vertex Model}\label{app:six-vertex}

Let $G=(V,E)$
be a loopless $4$-regular multigraph, with the four incident edges at each
vertex labeled $1,2,3,4$. {We consider Eulerian
orientations $D$ of the edges, meaning that exactly two edges are
directed into each vertex. Let $I_v(D)\subseteq\{1,2,3,4\}$ be the
labels of the two edges entering $v$. The weight $w_v(I)$ assigned
to an incoming pair $I$ is $a_v$ for $I=\{1,2\}$ or $\{3,4\}$,
$b_v$ for $I=\{1,3\}$ or $\{2,4\}$, and $c_v$ for $I=\{1,4\}$
or $\{2,3\}$. The partition function is
\[
 \sum_{\substack{D\text{ an Eulerian}\\\text{orientation of }G}}
       \prod_{v\in V(G)}w_v(I_v(D)).
\]
These are joint weights for the incoming pair; no factorization
into individual edge-direction weights is assumed
\cite{CaiLiuLu2019}.
\par
}

{
\begin{remark}
An Eulerian orientation specifies the direction of each edge, with
equal indegree and outdegree at every vertex. If $G$ is connected,
this guarantees the existence of a directed Euler tour; otherwise,
each connected component admits one. The six-vertex partition
function includes one term per orientation, regardless of how many
directed Euler tours it admits.

The Euler-tour weights in \cref{app:euler-tours} depend on the
whole pairing of the four incident half-edges, whereas the
six-vertex weights depend only on the incoming pair. For example,
if edges $1,2$ enter a vertex and $3,4$ leave it, the compatible
local pairings
\[
 \bigl\{\{1,3\},\{2,4\}\bigr\}
 \qquad\text{and}\qquad
 \bigl\{\{1,4\},\{2,3\}\bigr\}
\]
may receive different transition weights in \cref{app:euler-tours}.
Both have incoming pair $\{1,2\}$ under this orientation, so the
six-vertex weight is $a_v$ in either case.
\end{remark}
}

\begin{corollary}\label{cor:six-vertex}
For every loopless $4$-regular multigraph $G$, there is an
$\mathsf{FPRAS}$ for its six-vertex partition function when the local
{pair weights $a_v,b_v,c_v$} are nonnegative rationals satisfying
\[
 a_v\le b_v+c_v,\qquad b_v\le a_v+c_v,\qquad c_v\le a_v+b_v
 \quad\text{for every }v\in V(G).
\]
\end{corollary}

\begin{proof}
Use the incidences $\cU\coloneqq\{(v,e):v\in e\}$ as the ground set,
and let $\cM$ be the partition matroid choosing one incidence of each
edge. Assign to $S\subseteq\cU$ the product of its local incoming-pair
weights, with weight zero unless $S$ has two incidences at every vertex.
Then bases of $\cM$ with positive weight are exactly the six-vertex
orientations, with their prescribed weights. We apply
\cref{lem:matroid-constraint} to this weight.

Its generating polynomial is $\prod_vq_v$, where
\[
 q_v(\bm{x})=a_v(x_1x_2+x_3x_4)+b_v(x_1x_3+x_2x_4)
                  +c_v(x_1x_4+x_2x_3).
\]
An all-zero local triple gives partition function zero.
Otherwise the Hessian of $q_v$ has eigenvalues
$a_v+b_v+c_v$, $a_v-b_v-c_v$, $-a_v+b_v-c_v$, and
$-a_v-b_v+c_v$. The triangle inequalities leave at most one positive
eigenvalue, which is the log-concavity criterion for a nonzero
homogeneous quadratic with nonnegative coefficients
\cite[Section~2]{BrandenHuh2020}. Hence the product is log-concave.
Choosing any positive incoming pair at each vertex gives an initial
set, and coefficient queries are products of the given local weights.
Thus \cref{lem:matroid-constraint} applies, {including equality
in the triangle inequalities}.
\end{proof}

The earlier $\mathsf{FPRAS}$ of \cite{CaiLiuLu2019} requires the stronger
inequalities $a_v^2\le b_v^2+c_v^2$ and their permutations.

\subsection{Bipartite Quantum Max-Cut and the EPR Model}\label{app:quantum}

{
The Pauli matrices are
\[
 X\coloneqq\begin{pmatrix}0&1\\1&0\end{pmatrix},\qquad
 Y\coloneqq\begin{pmatrix}0&-\mathrm{i}\\\mathrm{i}&0\end{pmatrix},\qquad
 Z\coloneqq\begin{pmatrix}1&0\\0&-1\end{pmatrix},
\]
with $\mathrm{i}^2=-1$.}
{
A ket $|\bm{\psi}\rangle$ denotes a column vector, and the bra
$\langle\bm{\psi}|$ is its conjugate transpose. For $a,b\in\{0,1\}$, set
\[
 |0\rangle\coloneqq\begin{pmatrix}1\\0\end{pmatrix},\qquad
 |1\rangle\coloneqq\begin{pmatrix}0\\1\end{pmatrix},\qquad
 |ab\rangle\coloneqq|a\rangle\otimes|b\rangle,
\]
where $\otimes$ denotes the tensor product. Thus
\[
 |\bm{\Phi}^{+}\rangle\coloneqq\frac{|00\rangle+|11\rangle}{\sqrt2}
 =\frac1{\sqrt2}\begin{pmatrix}1\\0\\0\\1\end{pmatrix},
\]
and the local matrix
\[
 P\coloneqq|\bm{\Phi}^{+}\rangle\langle\bm{\Phi}^{+}|\in\mathbb C^{4\times4}
\]
is the orthogonal projector onto its span.

For a vertex set $V=\{1,\ldots,n\}$, the operators act on
$\bigotimes_{v=1}^n\mathbb C^2$, with one factor for each vertex.
The subscript $u$ means that a matrix acts on the factor at vertex
$u$ and as the identity on every other factor; explicitly,
\[
 X_u\coloneqq I_2^{\otimes(u-1)}\otimes X\otimes I_2^{\otimes(n-u)},
\]
where $I_2$ is the $2\times2$ identity matrix and zero tensor powers
are omitted. Define $Y_u,Z_u$ similarly, and let $I=I_2^{\otimes n}$.
Thus $X_uX_v$ applies $X$ at both vertices $u,v$.
For example, when $n=3$, $X_1X_3=X\otimes I_2\otimes X$.

For $a\in\{0,1\}$, the ket $|a\rangle_u$ denotes $|a\rangle$ in
the copy of $\mathbb C^2$ at vertex $u$. For $u<v$, the notation
\[
 |\bm{\Phi}^{+}\rangle_{uv}\coloneqq
 \frac{|0\rangle_u\otimes|0\rangle_v
       +|1\rangle_u\otimes|1\rangle_v}{\sqrt2}
\]
denotes a vector on these two vertex factors. Let
$P_{uv}\in\mathbb C^{2^n\times2^n}$ act as $P$ on the factors at
$u,v$ and as the identity on all remaining factors.
For example, when $n=3$,
\[
 P_{12}=P\otimes I_2,\qquad P_{23}=I_2\otimes P.
\]
\par
}

For a loopless graph $G=(V,E)$ with nonnegative rational edge weights $w_e$,
let
\[
 Q_{\mathrm{cut}}\coloneqq\sum_{e=uv}w_e
       \frac{I-X_uX_v-Y_uY_v-Z_uZ_v}{4},\qquad
 Q_{\mathrm{EPR}}\coloneqq\sum_{e=uv}w_e
       {P_{uv}}.
\]
{
For a Hermitian matrix $M$ and $\beta>0$, define its
\emph{partition function} by
\[
 Z_\beta(M)\coloneqq\operatorname{Tr}\e^{\beta M}.
\]
We consider approximating this quantity and the largest eigenvalue of $M$.
} For background on the two problems, see
\cite{TakahashiSlezakCrosson2024,MarwahaSud2026}.
{Write $W\coloneqq\sum_e w_e$.}

\begin{corollary}\label{cor:quantum-values}
{
Given a loopless graph $G$ with nonnegative rational edge weights,
there is an $\mathsf{FPRAS}$ for $\lambda_{\max}(Q_{\mathrm{EPR}})$.
For rational $\beta>0$, there is also an $\mathsf{FPRAS}$ for
$Z_\beta(Q_{\mathrm{EPR}})$ when $\beta W$ is polynomially bounded
in the input length. {More precisely, a randomized algorithm
approximates this partition function with relative error at most $\epsilon$
in time polynomial in the input length, $\beta W$, and $\epsilon^{-1}$.}

If $G$ is bipartite, the same claims hold for
$\lambda_{\max}(Q_{\mathrm{cut}})$ and $Z_\beta(Q_{\mathrm{cut}})$.
}
\end{corollary}
\begin{proof}
{We first prove the EPR claims by reducing them to the six-vertex
partition function in \cref{app:six-vertex}. We then deduce the Quantum
Max-Cut claims by an explicit change of basis.}

{
\paragraph{The EPR matrix: eigenvalue estimation.}
Fix $0<\eps\le1$. If $W=0$, the largest eigenvalue is zero and
$Z_\beta(Q_{\mathrm{EPR}})=2^n$. Otherwise delete zero-weight edges,
put $p_e\coloneqq w_e/W$, and set
$Q\coloneqq Q_{\mathrm{EPR}}/W=\sum_{e=uv}p_eP_{uv}$.
Each $P_{uv}$ is an orthogonal projector, and
$\sum_e p_e=1$, so $0\preceq Q\preceq I$.
For the unit vector $|\bm{\xi}\rangle=|0\rangle^{\otimes n}$,
we have $\langle\bm{\xi}|P_{uv}|\bm{\xi}\rangle=1/2$ for every edge $uv$.
Consequently,
\[
 \frac12=\frac12\sum_e p_e
 =\langle\bm{\xi}|Q|\bm{\xi}\rangle
 \le\lambda_{\max}(Q)\le1.
\]
\par
}

{
For $\beta>0$, let $\lambda_1=\lambda_{\max}(Q)\ge\cdots\ge\lambda_{2^n}$ be the
eigenvalues of $Q$. The eigenvalues of $\e^{\beta Q}$ are
$\e^{\beta\lambda_i}$, so
\[
 \e^{\beta\lambda_1}
 \le Z_\beta(Q)=\sum_{i=1}^{2^n}\e^{\beta\lambda_i}
 \le 2^n\e^{\beta\lambda_1}.
\]
The lower bound keeps only the $i=1$ term; the upper bound replaces
each term by $\e^{\beta\lambda_1}$.
Taking logarithms and dividing by $\beta$ gives
}
\[
 \lambda_{\max}(Q)\le\frac{\log {Z_\beta(Q)}}{\beta}
 \le\lambda_{\max}(Q)+\frac{n\log2}{\beta}.
\]
Thus $\beta\coloneqq8(n+1)/\eps$ and {an estimate of
$Z_\beta(Q)$ with relative error at most $1/4$} suffice for additive error $\eps/4$ in
{$\lambda_{\max}(Q)$}, including the error from evaluating the logarithm
to sufficiently many bits.
Since {$\lambda_{\max}(Q)\ge1/2$}, this gives relative error at
most $\eps$ after multiplying by $W$.

{
\paragraph{The EPR matrix: a six-vertex instance.}
We next transform $Q$ into a matrix $A$ with the same eigenvalues,
so that $Z_\beta(Q)=\operatorname{Tr}\e^{\beta A}$, and approximate this
trace using \cref{cor:six-vertex}. Put
\[
 K\coloneqq\frac{I_2-\mathrm{i}X}{\sqrt2},\qquad
 |\bm{\Psi}^{+}\rangle\coloneqq\frac{|01\rangle+|10\rangle}{\sqrt2},\qquad
 R\coloneqq|\bm{\Psi}^{+}\rangle\langle\bm{\Psi}^{+}|.
\]
The matrix $K$ is unitary and
$(K\otimes K)|\bm{\Phi}^{+}\rangle=-\mathrm{i}|\bm{\Psi}^{+}\rangle$.
For $e=uv$, let $R_e\coloneqq R_{uv}$ be the extension of $R$ to the
full space by the same convention as $P_{uv}$. Thus
\[
 A\coloneqq K^{\otimes n}Q(K^*)^{\otimes n}=\sum_e p_eR_e,
 \qquad \operatorname{Tr}\e^{\beta A}=Z_\beta(Q).
\]
Here $K^*$ is the conjugate transpose of $K$.
\par
}
{Fix rational $\beta>0$ and a target relative error
$0<\eta\le1$. Set
$m\coloneqq\lceil\max\{1,4\beta,20\beta^2/\eta\}\rceil$ and
$\delta\coloneqq\beta/m$. Label the edges $1,\ldots,|E|$ in any order,
and multiply their matrices in this order and then in reverse order:}
\[
 B\coloneqq
 \left(\prod_{e=1}^{|E|}\left(I+\frac{\delta p_e}{2}{R_e}\right)\right)
 \left(\prod_{e=|E|}^{1}\left(I+\frac{\delta p_e}{2}{R_e}\right)\right).
\]
{
We claim the following bound:
\begin{equation}\label{eq:quantum-product-bound}
 B\succ0,\qquad \|\log B-\delta A\|\le5\delta^2,
\end{equation}
where $\|\cdot\|$ denotes the operator norm. The matrix logarithm
$\log B$ is obtained by applying the scalar logarithm to the positive
eigenvalues of $B$ in an orthonormal eigenbasis.
The intuition comes from the scalar approximation
$\log(1+\delta x)-\delta x=O(\delta^2x^2)$ as $\delta x\to0$:
the linear terms sum to $\delta A$, and we expect an error quadratic
in $\delta$. We will verify \cref{eq:quantum-product-bound} at the
end of the proof, accounting for the order of the matrix factors.
Assuming this estimate, our choice of $m$ gives
\[
 \|m\log B-\beta A\|
 \le5m\delta^2=\frac{5\beta^2}{m}\le\frac\eta4.
\]
Let $\mu_1\ge\cdots\ge\mu_{2^n}$ and
$\nu_1\ge\cdots\ge\nu_{2^n}$ be the eigenvalues of $m\log B$ and
$\beta A$, respectively. These matrices are Hermitian, so the norm
bound implies
\[
 \beta A-\frac\eta4 I\preceq m\log B
 \preceq\beta A+\frac\eta4 I.
\]
The min--max principle for eigenvalues therefore gives
\[
 \nu_i-\frac\eta4\le\mu_i\le\nu_i+\frac\eta4
 \qquad(1\le i\le2^n).
\]
Exponentiating these scalar inequalities and summing over $i$ gives
\[
 \e^{-\eta/4}\sum_i\e^{\nu_i}
 \le\sum_i\e^{\mu_i}
 \le\e^{\eta/4}\sum_i\e^{\nu_i}.
\]
Since
\[
 \sum_i\e^{\mu_i}=\operatorname{Tr}\e^{m\log B}
                  =\operatorname{Tr}(B^m),\qquad
 \sum_i\e^{\nu_i}=\operatorname{Tr}\e^{\beta A}=Z_\beta(Q),
\]
we obtain
\[
 \e^{-\eta/4}Z_\beta(Q)\le\operatorname{Tr}(B^m)
                       \le\e^{\eta/4}Z_\beta(Q).
\]
}

{For each edge $e$, write $s_e\coloneqq\delta p_e/2$ and
$F_e\coloneqq I+s_eR_e$ for its matrix in $B$.
Writing $I_4$ for the $4\times4$ identity matrix, the local table of
$F_e$ in row and column order $00,01,10,11$ is}
\[
 {I_4+s_eR}=
 \begin{pmatrix}
 1&0&0&0\\
 0&1+s_e/2&s_e/2&0\\
 0&s_e/2&1+s_e/2&0\\
 0&0&0&1
 \end{pmatrix}.
\]
{
The product $B^m$ uses each edge $2m$ times. Each use gives
one \emph{node} carrying the displayed $4\times4$ table. At a node for
$e=uv$, two edge ends specify the first and second bits of the row
label, corresponding to $u,v$; another two specify the first and
second bits of the column label. For each original vertex $u$,
connect its column end at each node to its row end at the next node
involving $u$, following the order in $B^m$ and connecting the last
such node back to the first. This defines the tensor network.

A configuration $\sigma$ assigns a bit $0$ or $1$ to every network
edge. The four incident bits select one entry at each node $t$;
call this entry its weight $w_t(\sigma)$. The \emph{partition function}
sums the product of node weights over all configurations:
\[
 Z_{\mathrm{TN}}\coloneqq
 \sum_{\sigma}
 \prod_{\text{nodes }t} w_t(\sigma).
\]
Expanding the trace sums the same products of node weights. Matrices
for edges not incident to $u$ leave its bit unchanged, while each
isolated vertex of $G$ contributes two choices for its bit.
Hence, if $G$ has $h$ isolated vertices,
\[
 \operatorname{Tr}(B^m)=2^h Z_{\mathrm{TN}}.
\]

This tensor network has the same partition function as a six-vertex
model; see \cref{fig:quantum-six-vertex} for an illustrated example.
\par
}

\begin{figure}[!t]
\centering

\begin{tikzpicture}[
  x=1cm,y=1cm,>=stealth,
  every node/.style={font=\footnotesize,text=black},
  tensor/.style={draw=black,fill=black!4,rounded corners=1pt,
    minimum width=10mm,minimum height=7mm,inner sep=2pt,align=center},
  networknode/.style={circle,draw=black,fill=black!4,minimum size=7mm,inner sep=1pt},
  wire/.style={draw=black,line width=.7pt},
  bit/.style={fill=white,inner sep=1.5pt}
]
 \newcommand{\pathtrace}[2]{%
  \begin{scope}[shift={(#1)}]
   \ifnum#2=0
    \node[tensor] (a) at (0,1) {$T_1$\\$vw$};
    \node[tensor] (b) at (4.8,1) {$T_2$\\$uw$};
    \node[tensor] (c) at (4.8,-1) {$T_3$\\$uw$};
    \node[tensor] (d) at (0,-1) {$T_4$\\$vw$};
   \else
    \node[networknode] (a) at (0,1) {$T_1$};
    \node[networknode] (b) at (4.8,1) {$T_2$};
    \node[networknode] (c) at (4.8,-1) {$T_3$};
    \node[networknode] (d) at (0,-1) {$T_4$};
   \fi
   \ifnum#2=2
    \tikzset{zero/.style={wire,<-},one/.style={wire,->}}
   \else
    \tikzset{zero/.style={wire},one/.style={wire}}
   \fi
   \draw[one] (a.east)--node[bit]{$w\ifnum#2=1=1\fi$}(b.west);
   \draw[one] (b.south)--node[bit]{$w\ifnum#2=1=1\fi$}(c.north);
   \draw[one] (c.west)--node[bit]{$w\ifnum#2=1=1\fi$}(d.east);
   \draw[zero] (d.north)--node[bit]{$w\ifnum#2=1=0\fi$}(a.south);
   \draw[zero,rounded corners=3pt] (a.west)--(-.85,1)
     --node[bit]{$v\ifnum#2=1=0\fi$}(-.85,-1)--(d.west);
   \draw[one,rounded corners=3pt] (d.45)--(.85,-.75)
     --node[bit]{$v\ifnum#2=1=1\fi$}(.85,.75)--(a.315);
   \draw[zero,rounded corners=3pt] (b.east)--(5.65,1)
     --node[bit]{$u\ifnum#2=1=0\fi$}(5.65,-1)--(c.east);
   \draw[zero,rounded corners=3pt] (c.135)--(3.95,-.75)
     --node[bit]{$u\ifnum#2=1=0\fi$}(3.95,.75)--(b.225);
  \end{scope}%
 }

 \node[font=\small\bfseries] at (2.65,2.25) {(a) Original path $G$};
 \coordinate (u) at (1.3,-.1);
 \coordinate (v) at (4,-.1);
 \coordinate (w) at (2.65,1.6);
 \draw[wire] (u)--node[above left]{$F_{uw}$}(w);
 \draw[wire] (v)--node[above right]{$F_{vw}$}(w);
 \foreach \p in {u,v,w} {\fill[black] (\p) circle (1.7pt);}
 \node[below left] at (u) {$u$};
 \node[below right] at (v) {$v$};
 \node[above] at (w) {$w$};
 \node[align=center] at (2.65,-1.15)
   {Each edge $e$ specifies a matrix $F_e$.\\Each edge is used twice in $B$.};
 \draw[wire,->] (4.6,.7)--node[above,font=\scriptsize]{expand $B$}(5.85,.7);

 \node[font=\small\bfseries] at (9.65,2.25)
   {(b) Tensor network for $\operatorname{Tr}(B^m)$, $m=1$};
 \pathtrace{7.25,0}{0}

 \node[font=\small\bfseries] at (3.3,-2.25)
   {(c) One assignment of the shared bits};
 \pathtrace{.9,-4.15}{1}
 \node[font=\small\bfseries] at (10.4,-2.25)
   {(d) Its orientation: two in, two out};
 \pathtrace{8,-4.15}{2}

 \node[align=center] at (6.8,-6.3) {$
  \begin{array}{c|cccc}
   \text{node}&T_1&T_2&T_3&T_4\\[2pt]
   (\text{row},\text{column})&(10,01)&(01,01)&(01,01)&(01,10)\\[2pt]
   \text{local entry}&s_{vw}/2&1+s_{uw}/2&1+s_{uw}/2&s_{vw}/2
  \end{array}$};
\end{tikzpicture}
\caption{From a two-edge path to a six-vertex instance.
(a)--(b) The edge sequence $vw,uw,uw,vw$ in $B$ gives nodes $T_1,\ldots,T_4$,
each carrying a $4\times4$ table. Each network edge carries the bit of the
original vertex named by its label. (c)--(d) Keep row bits, flip column
bits, and point each arrow toward its end marked $1$. Every nonzero
table entry has equally many $1$'s in its row and column labels, so
two arrows enter each node. This correspondence is reversible and
preserves the node weights. Label the row ends $1,2$ and the column
ends $4,3$, respectively. The local six-vertex weights are
$(a,b,c)=(1,1+s_e/2,s_e/2)$, satisfying
$b=a+c$ and hence the triangle inequalities. Each node has degree four;
there are no loops because each original edge is used at least twice.
The table lists the four selected local entries;
their product is this assignment's contribution to $\operatorname{Tr}(B)$.
For $\operatorname{Tr}(B^m)$, repeat this edge sequence $m$ times before
closing the connections.}
\label{fig:quantum-six-vertex}
\end{figure}
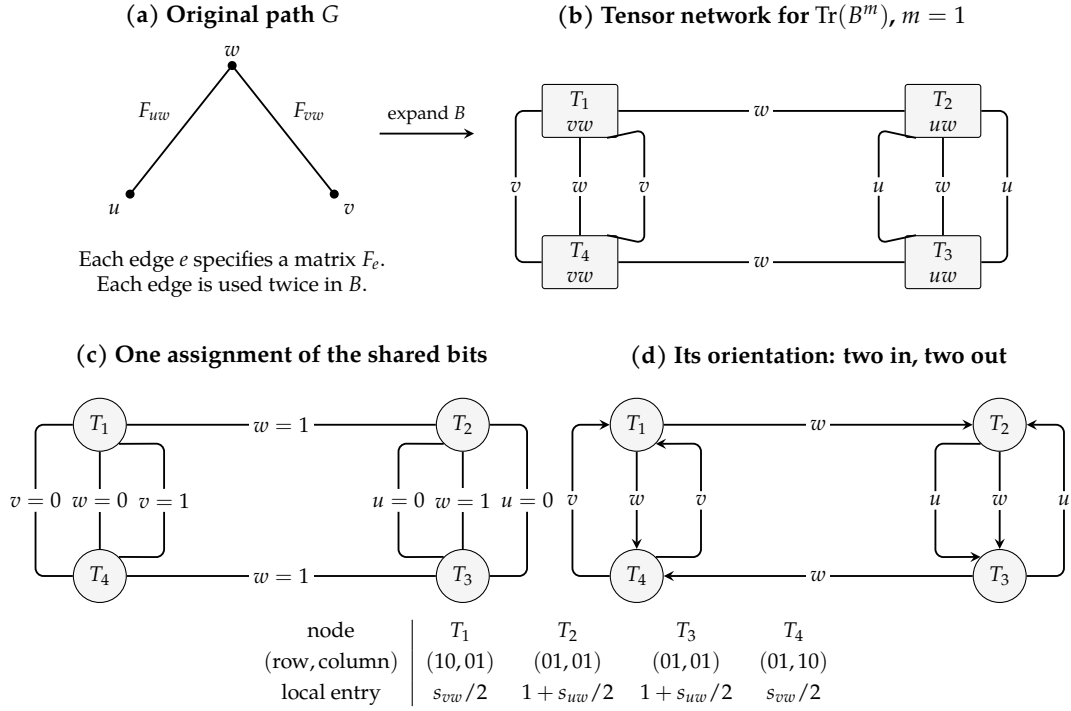

{
There are $2m|E|$ nodes in the network, with rational weights of polynomial encoding
length. {An estimate of $\operatorname{Tr}(B^m)$ with relative
error at most $\eta/4$, obtained from \cref{cor:six-vertex}, together with
the displayed trace bound, gives an estimate of $Z_\beta(Q)$ with relative
error at most $\eta$}: indeed,
$(1+\eta/4)\e^{\eta/4}\le1+\eta$ and
$(1-\eta/4)\e^{-\eta/4}\ge1-\eta$.
The running time is polynomial in the input length, $\beta$, and
$\eta^{-1}$. Since
\[
 Z_\beta(Q_{\mathrm{EPR}})=Z_{\beta W}(Q),
\]
this proves the partition-function claim with running time polynomial
in the input length, $\beta W$, and $\eta^{-1}$.
Taking $\eta=1/4$ and $\beta=8(n+1)/\epsilon$ for the normalized matrix
$Q$ also completes the largest-eigenvalue approximation described above.
\par
}

{
\par\goodbreak
\begin{samepage}
\paragraph{The Quantum Max-Cut matrix.}
Now suppose $G$ is bipartite, with $V=V_0\sqcup V_1$.
We show that $Q_{\mathrm{cut}}$ has the same eigenvalues as
$Q_{\mathrm{EPR}}$ on the same weighted graph. Define the unitary matrix
\[
 U\coloneqq\bigotimes_{v\in V}U_v,\qquad
 U_v\coloneqq
 \begin{cases}
 I_2,&v\in V_0,\\
 Y,&v\in V_1.
 \end{cases}
\]
Since $YXY=-X$, $YYY=Y$, and $YZY=-Z$, conjugation by $U$ gives
\[
 UQ_{\mathrm{cut}}U^*
 =\sum_{e=uv}w_e\frac{I+X_uX_v-Y_uY_v+Z_uZ_v}{4}
 =\sum_{e=uv}w_eP_{uv}
 =Q_{\mathrm{EPR}}.
\]
The second equality follows by multiplying the displayed Pauli matrices.
Therefore
\[
 \lambda_{\max}(Q_{\mathrm{cut}})=\lambda_{\max}(Q_{\mathrm{EPR}}),
 \qquad Z_\beta(Q_{\mathrm{cut}})=Z_\beta(Q_{\mathrm{EPR}}).
\]
The EPR algorithms prove both remaining claims.
\par
\end{samepage}
}

{
\paragraph{Verification of \cref{eq:quantum-product-bound}.}
We return to the EPR product $B$ on an arbitrary graph.
Let $F$ be the first parenthesized product
in the definition of $B$. Each factor is Hermitian and positive
definite, with eigenvalues $1$ and $1+\delta p_e/2$. Thus $F$ is
invertible, the reverse product is $F^*$, and $B=FF^*\succ0$.

To bound $B-I-\delta A$, list the matrices $\delta p_eR_e/2$
in their order of occurrence in $B$ as $H_1,\ldots,H_N$, where
$N=2|E|$, and put $h_j\coloneqq\|H_j\|$. Each edge occurs twice, so
\[
 B=\prod_{j=1}^N(I+H_j),\qquad
 \sum_{j=1}^N H_j=\delta A,\qquad
 \sum_{j=1}^N h_j=\delta\sum_e p_e=\delta.
\]
Expand the product by choosing either $I$ or $H_j$ at each position,
keeping the chosen matrices in their original order. The term with no
chosen $H_j$ is $I$, while the terms with exactly one chosen $H_j$
sum to $\sum_jH_j=\delta A$.
Subtracting these terms therefore leaves
\[
 B-I-\delta A
 =\sum_{k=2}^N\;\sum_{1\le i_1<\cdots<i_k\le N}
 H_{i_1}\cdots H_{i_k}.
\]
The norm of a matrix product is at most the product of the norms of its matrices:
\[
 \|H_{i_1}\cdots H_{i_k}\|
 \le\|H_{i_1}\|\cdots\|H_{i_k}\|
 =h_{i_1}\cdots h_{i_k}.
\]
Applying the triangle inequality to the displayed remainder now yields
\begin{align*}
 \|B-I-\delta A\|
 &\le\sum_{k=2}^N\;\sum_{1\le i_1<\cdots<i_k\le N}
       h_{i_1}\cdots h_{i_k}\\
 &=\prod_{j=1}^N(1+h_j)-1-\sum_{j=1}^N h_j\\
 &\le\e^\delta-1-\delta.
\end{align*}
The equality uses the scalar expansion of $\prod_j(1+h_j)$, with
its constant term $1$ and linear terms $\sum_jh_j$ removed. The last
inequality uses $1+h_j\le\e^{h_j}$ and $\sum_jh_j=\delta$.
Similarly, subtracting only $I$ keeps all terms with $k\ge1$, so
\begin{align*}
 \|B-I\|
 &\le\sum_{k=1}^N\;\sum_{1\le i_1<\cdots<i_k\le N}
       h_{i_1}\cdots h_{i_k}\\
 &=\prod_{j=1}^N(1+h_j)-1\le\e^\delta-1.
\end{align*}

Set $D\coloneqq B-I$. Since $\delta\le1/4$, the scalar exponential
series gives
\[
 \e^\delta-1\le\frac{\delta}{1-\delta}\le2\delta,
 \qquad
 \e^\delta-1-\delta\le\frac{\delta^2}{2(1-\delta)}\le\delta^2.
\]
Hence $\|D\|\le2\delta\le1/2$ and $\|D-\delta A\|\le\delta^2$.
Applying the scalar series for $\log(1+x)$ in an orthonormal
eigenbasis of $D$ gives the absolutely convergent matrix series
\begin{align*}
 \log B&=D+\sum_{k=2}^{\infty}\frac{(-1)^{k+1}}{k}D^k,\\
 \|\log B-D\|
 &\le\sum_{k=2}^{\infty}\frac{\|D\|^k}{k}
 \le\frac{\|D\|^2}{2(1-\|D\|)}
 \le4\delta^2.
\end{align*}
Combining the two remainder bounds proves
\[
 \|\log B-\delta A\|
 \le\|\log B-D\|+\|D-\delta A\|
 \le5\delta^2.
\]
This proves \cref{eq:quantum-product-bound} and completes the proof.
\par
}
\end{proof}

\subsection{Matroid-Congestion Gibbs Distributions}\label{app:congestion}

Let $p,q\ge1$. There are $p$ labeled players, each choosing a basis $B_i$
of a matroid $\cM_i$ on the same ground set $E$ of size $q$, given by an
independence oracle. A \emph{profile} $B=(B_1,\ldots,B_p)$ records
all players' choices. The load of an element $e\in E$ is the number of
bases containing it, $\ell_e(B)\coloneqq|\{i:e\in B_i\}|$.
For each element, supply a table of nonnegative rational weights
$a_e(0),\ldots,a_e(p)$ that are positive exactly for
$l_e\le k\le u_e$, where $0\le l_e\le u_e\le p$ are integers.
Assume these weights are log-concave:
$a_e(k)^2\ge a_e(k-1)a_e(k+1)$ for $1\le k<p$.
The partition function is
\begin{align}\label{eq:congestion-partition}
 Z(a)\coloneqq\sum_{B_i\in\cB(\cM_i),\ 1\le i\le p}
                         \prod_{e\in E}a_e(\ell_e(B)).
\end{align}
Thus a profile has positive weight exactly when
$l_e\le\ell_e(B)\le u_e$ for every element $e$.

\begin{corollary}\label{cor:congestion-tables}
For the matroids and rational log-concave tables above, there is an
$\mathsf{FPRAS}$ for $Z(a)$ with running time polynomial in
$pq$, the table encoding length, and $\epsilon^{-1}$.
\end{corollary}

\begin{proof}
We apply \cref{lem:matroid-constraint} on the set
$\cU\coloneqq[p]\times E$ of ordered pairs $(i,e)$: the first entry
$i\in[p]$ is a player index, and the second entry $e\in E$ is an element.
Including $(i,e)$ records that player $i$ uses element $e$.
Let $\cM$ be the matroid whose bases are precisely
$\{(i,e):1\le i\le p,\ e\in B_i\}$ as each $B_i$ ranges over the bases of $\cM_i$.
Thus its rank is $R\coloneqq\sum_i{\rk}(\cM_i)$.
Put $V_e\coloneqq[p]\times\{e\}$ and define
\[
 \omega(X)\coloneqq\allone[|X|=R]\prod_e a_e(|X\cap V_e|).
\]
Each basis of $\cM$ specifies one profile, so
$Z(a)=\sum_{X\in\cB(\cM)}\omega(X)$.

To show that the generating polynomial of $\omega$ is log-concave,
we use the following sufficient condition
\cite[Corollary~3.16 and Theorem~2.30]{BrandenHuh2020}:
for every $X,Y$ with $\omega(X),\omega(Y)>0$ and every
$a\in X\setminus Y$, some $b\in Y\setminus X$ satisfies
\begin{equation}\label{eq:congestion-exchange}
 \omega(X)\omega(Y)\le\omega(X-a+b)\omega(Y+a-b).
\end{equation}
Here $X-a+b$ means $(X\setminus\{a\})\cup\{b\}$, and similarly for $Y+a-b$.
The sets $X,Y$ in this condition need only have positive $\omega$-weight;
they are not required to be bases of $\cM$.
For each element $g\in E$, write $x_g=|X\cap V_g|$ and $y_g=|Y\cap V_g|$.
Fix $a \in V_e$ for some element $e$.
If $x_e\le y_e$, the fact that $a\in(X\setminus Y)\cap V_e$ gives
\[
 |(Y\setminus X)\cap V_e|
 =y_e-x_e+|(X\setminus Y)\cap V_e|\ge1.
\]
Thus we can choose $b\in(Y\setminus X)\cap V_e$ simply by counting;
no matroid property is needed. Exchanging $a$ and $b$ leaves every
load unchanged, so $\omega(X-a+b)=\omega(X)$ and
$\omega(Y+a-b)=\omega(Y)$.
Otherwise $x_e>y_e$. Since
$\sum_g x_g=|X|=|Y|=\sum_g y_g$, there must be a different element
$f\ne e$ with $x_f<y_f$. Choose $b\in(Y\setminus X)\cap V_f$.
Since $V_f=[p]\times\{f\}$, the first entry of $b$ is a player index
and its second entry is the fixed element $f$. This choice is possible
because $Y$ contains more pairs in $V_f$ than $X$.
Exchanging $a$ and $b$ changes the loads at $e$ from $(x_e,y_e)$ to
$(x_e-1,y_e+1)$, and at $f$ from $(x_f,y_f)$ to $(x_f+1,y_f-1)$.
These loads stay within their allowed intervals, and the log-concavity
of $a_e,a_f$ gives \cref{eq:congestion-exchange}.

To find a set of positive weight, choose integers $k_e\in[l_e,u_e]$ with
$\sum_e k_e=R$ and select $k_e$ elements from each $V_e$. Such integers exist
exactly when $\sum_e l_e\le R\le\sum_e u_e$; otherwise $Z(a)=0$.
We evaluate $\omega(X)$ by checking $|X|=R$, counting $|X\cap V_e|$
for each element, and multiplying the corresponding weights.
The ratio of its largest
to smallest positive weight satisfies
\[
 \cR_\omega\le\prod_e
     \frac{\max_{l_e\le k\le u_e}a_e(k)}
          {\min_{l_e\le k\le u_e}a_e(k)}.
\]
Thus \cref{lem:matroid-constraint} applies.
\end{proof}

{
\paragraph{Weights specified by costs.}
For the same matroids and integers $l_e,u_e$ as in \cref{cor:congestion-tables},
suppose we are given a rational $\beta\ge0$ and \emph{rational costs}
$c_e(1)\le\cdots\le c_e(p)$, where $c_e(k)$ is the cost to each player
using element $e$ at load $k$. Define the tables by
\[
 a_e(k)\coloneqq
 \begin{cases}
  \exp\bigl(-\beta\sum_{j=1}^k c_e(j)\bigr),& l_e\le k\le u_e,\\
  0,&\text{otherwise},
 \end{cases}
\]
with the empty sum equal to zero. We use the same partition function
$Z(a)$ as in \cref{eq:congestion-partition}. When $Z(a)>0$, the
\emph{Gibbs distribution} assigns probability
$\prod_e a_e(\ell_e(B))/Z(a)$ to profile $B$. When $l_e=0$ and $u_e=p$
for every $e$, sampling from this distribution is the question posed in
\cite[Section~1.3]{Kleer2023}.

If each real weight can be evaluated exactly in one operation, this is
already an application of the proof of \cref{cor:congestion-tables}:
the ratios $a_e(k)/a_e(k-1)=e^{-\beta c_e(k)}$, for $l_e<k\le u_e$,
are nonincreasing in $k$, so the tables are log-concave.
Sampling follows by using approximate counts to make the players'
choices one element at a time. The next corollary handles the fact
that these weights need not be rational: we approximate them by
rational log-concave tables and control the resulting error.
To state the running time, put
\[
 \Lambda\coloneqq\sum_e\log\frac{\max_{l_e\le k\le u_e}a_e(k)}
                                    {\min_{l_e\le k\le u_e}a_e(k)}.
\]
This value is rational and can be computed exactly from the costs.
The weight-ratio bound in the preceding proof gives
$\log\cR_\omega\le\Lambda$, so even with exact evaluation the
running-time bound depends polynomially on the numerical value of
$\Lambda$. Maxima below range over $l_e\le k\le u_e$.

\begin{corollary}\label{cor:congestion-gibbs}
For the matroids and rational numbers
$c_e(1)\le\cdots\le c_e(p)$ and $\beta\ge0$ specified above,
restrict the profiles $B$ to those satisfying
$l_e\le\ell_e(B)\le u_e$ for every $e\in E$.
{A randomized algorithm returns an estimate of $Z(a)$ with
relative error at most $\epsilon$}, with probability at least $3/4$, in the form
$\bigl(\prod_e\max_k a_e(k)\bigr)\widehat Z$, where $\widehat Z$ is
rational and the exponential factor is specified by its exact rational
exponent. If $Z(a)>0$, there is also a sampler
within total variation distance $\epsilon$ of the Gibbs distribution.
The running times are polynomial in $pq$, the input encoding length,
$\epsilon^{-1}$, and the \emph{numerical} value of $\Lambda$.
If the costs are nonnegative and every $l_e=0$, then
$\max_k a_e(k)=a_e(0)=1$ for every $e$, so the estimate is rational.
\end{corollary}
\par}

\begin{proof}
{
It suffices to approximate the partition function and sample after
dividing each table by its largest entry. Indeed, put
\[
 \bar a_e(k)\coloneqq\frac{a_e(k)}{\max_j a_e(j)}.
\]
This leaves the Gibbs distribution unchanged, and
$Z(a)=\bigl(\prod_e\max_k a_e(k)\bigr)Z(\bar a)$.
The logarithm of this factor is rational and can be computed exactly
from the given costs. We replace $\bar a_e$ by rational log-concave
tables with nearly the same weights. For $0<\epsilon\le1$, set
$\tau\coloneqq\epsilon/[64q(p+1)^2]$ and $r\coloneqq1+\tau$.
For each $l_e\le k\le u_e$, write $\bar a_e(k)=e^{-x}$, where
$x=-\log\bar a_e(k)\in[0,\Lambda]$ is rational and computable exactly
from the costs. Standard algorithms for evaluating the exponential
function \cite[Sections~4.3--4.4]{BrentZimmermann2010} compute a
positive rational approximation $\widetilde a_e(k)$ satisfying
\[
 \left|\log\frac{\widetilde a_e(k)}{\bar a_e(k)}\right|\le\tau/16
\]
in time polynomial in the input encoding length, $\Lambda$, and
$\log(\tau^{-1})$.
\par}
{The rational approximation may lose log-concavity.
To ensure log-concavity, we make a further small adjustment: define}
$b_e(k)\coloneqq\widetilde a_e(k)r^{-k^2}$ on the allowed interval
and zero outside it. For every $l_e<k<u_e$,
\[
 \log\frac{b_e(k)^2}{b_e(k-1)b_e(k+1)}
 \ge2\log r-\tau/4>0.
\]
Thus $b_e$ is log-concave, and every profile $B$
{with $l_e\le\ell_e(B)\le u_e$ for every $e\in E$} obeys
\[
 {\left|\log\frac{\prod_e b_e(\ell_e(B))}{\prod_e \bar a_e(\ell_e(B))}\right|}
 \le q\tau/16+qp^2\log r\le\epsilon/64.
\]
Consequently the sums $Z(b)$ and {$Z(\bar a)$} differ by a factor at most
$e^{\epsilon/64}$ in either direction. When these sums are positive,
dividing the profile weights by their respective sums gives two
probability distributions whose total variation distance is at most
$e^{\epsilon/32}-1<\epsilon/16$.
Applying \cref{cor:congestion-tables} with accuracy $\epsilon/4$
{{gives an estimate of $Z(\bar a)$ with relative error at most $\epsilon$.}
Multiplying by $\prod_e\max_k a_e(k)$ gives the required estimate of $Z(a)$.}
For sampling, decide for each pair $(i,e)$ whether $e\in B_i$, using
approximate counts to estimate the total weight of each choice
\cite{JVV1986}. Fixing these choices deletes or contracts elements of
the player matroids. For each element, it also shifts the load index
by the number of players already required to use it and removes
infeasible entries. The resulting tables remain log-concave, so
\cref{cor:congestion-tables} applies at every step.
Sampling from the distribution defined by the rational
tables $b_e$ within total variation distance $\epsilon/4$ then gives
the desired approximate Gibbs sample.
The rational tables have encoding length polynomial in $\Lambda$, the
input encoding length, and $\epsilon^{-1}$, which gives the stated
running time.
\end{proof}

\subsection{Integer Polymatroid Intersections, Contingency Tables, and Bounded Flows}
\label{app:integer-polymatroids}

We first reduce weighted intersections of integer families to subset
convolution. Each integer coordinate is encoded by a subset of a block
of labeled elements; factorial weights compensate for the multiple
encodings of the same integer vector. This gives an approximate
counting algorithm for intersections of integral polymatroid bases.
We next obtain an $\mathsf{FPRAS}$ for contingency tables with
log-concave entry weights by separating their row-sum and column-sum
constraints. A weight-preserving bijection then gives the result for
bounded integer flows.
In \cref{app:network-polytopes}, the flow result yields an
$\mathsf{FPRAS}$ for vertices of network-matrix polytopes.

{We use log-$\mathrm M$-concave weights and polarization as defined in \cref{app:weight-background}.}
A nonempty finite family $\mathcal A\subseteq\mathbb Z_{\ge0}^m$
is \emph{$\mathrm M$-convex} if its indicator weight is
log-$\mathrm M$-concave: all vectors have the same coordinate sum,
and both exchanged vectors in \cref{eq:integer-weight-exchange}
belong to $\mathcal A$.
In particular, the integer points of an integral polymatroid base
polytope form such a family \cite{BrandenHuh2020}.

\begin{proposition}[Counting weighted intersections of integer families]
\label{prop:integer-weight-intersection}
Let $\bm{c}\in\mathbb Z_{\ge1}^m$ be given in \emph{unary}, meaning
that each $c_i$ is represented by $c_i$ copies of the symbol $1$.
Let $\mu,\nu:[\bm{0},\bm{c}]\cap\mathbb Z^m\to\mathbb R_{\ge0}$
be log-$\mathrm M$-concave weights, extended by zero outside the box,
given by evaluation oracles. Suppose we are given initial vectors
$\bm{x}_0,\bm{y}_0$ with $\mu(\bm{x}_0)>0$ and $\nu(\bm{y}_0)>0$,
and upper bounds on their coefficient ranges whose logarithms are
polynomial in the input length. Then there is an $\mathsf{FPRAS}$ for
\begin{align*}
 Z\coloneqq\sum_{\bm{k}\in[\bm{0},\bm{c}]\cap\mathbb Z^m}
                   \mu(\bm{k})\nu(\bm{k})
\end{align*}
in the arithmetic model of \cref{thm:log-concave-convolution}.
In particular, membership
oracles and one point in each of two $\mathrm M$-convex families
$\mathcal A,\mathcal D$ in the box suffice to approximate
$|\mathcal A\cap\mathcal D|$.
\end{proposition}

\begin{proof}
We construct two subset weights whose convolution is $Z/\bm{c}!$,
where $\bm{k}!\coloneqq\prod_i k_i!$. We then show that their
generating polynomials are log-concave and apply
\cref{thm:log-concave-convolution}.
If the coordinate sums of the two initial vectors differ, the
supports are disjoint and $Z=0$, so assume these sums agree.

\emph{Encoding the integer vectors.}
Take disjoint sets $\cU_i$ with $|\cU_i|=c_i$ and put
$\cU\coloneqq\bigsqcup_i\cU_i$. For $S\subseteq\cU$, write
$k_i(S)\coloneqq|S\cap\cU_i|$ and define
\begin{align}
 \widehat\mu(S)
 &\coloneqq\mu(\bm{k}(S))\frac{(\bm{c}-\bm{k}(S))!}{\bm{c}!},
 \label{eq:integer-clone-weights}\\
 \widehat\nu(S)
 &\coloneqq\nu(\bm{c}-\bm{k}(S))\frac{(\bm{c}-\bm{k}(S))!}{\bm{c}!}.
 \notag
\end{align}
In the convolution, $\widehat\mu(S)$ and
$\widehat\nu(\cU\setminus S)$ use the same integer vector
$\bm{k}(S)$. Exactly
$\binom{\bm{c}}{\bm{k}}\coloneqq\prod_i\binom{c_i}{k_i}$
subsets encode $\bm{k}$, and the factorial weights cancel this
multiplicity:
\begin{align}\label{eq:integer-convolution-identity}
 (\widehat\mu*\widehat\nu)(\cU)
 &=\sum_{\bm{0}\le\bm{k}\le\bm{c}}
   \binom{\bm{c}}{\bm{k}}\mu(\bm{k})\nu(\bm{k})
   \frac{(\bm{c}-\bm{k})!\,\bm{k}!}{(\bm{c}!)^2}
   =\frac{Z}{\bm{c}!}.
\end{align}

\emph{Log-concavity of the generating polynomials.}
We use the weighted exchange condition in \cref{eq:integer-weight-exchange}
for $\nu$; log-concavity of its generating polynomial alone would not
provide this condition. Put
$\bar\nu(\bm{k})\coloneqq\nu(\bm{c}-\bm{k})$ for integer vectors
$\bm{k}$ satisfying $0\le k_i\le c_i$ for every $i$, and set
$\bar\nu(\bm{k})=0$ otherwise. We first verify that $\bar\nu$ is
also log-$\mathrm M$-concave.
Take $\bm{x},\bm{y}\in\supp\bar\nu$ with $x_i>y_i$.
Then $\bm{c}-\bm{y},\bm{c}-\bm{x}\in\supp\nu$ and
$(\bm{c}-\bm{y})_i>(\bm{c}-\bm{x})_i$. Apply
\cref{eq:integer-weight-exchange} to these two vectors in this order.
It gives an index $j$ with
$(\bm{c}-\bm{y})_j<(\bm{c}-\bm{x})_j$, equivalently $x_j<y_j$, and
\begin{align*}
 \bar\nu(\bm{x})\bar\nu(\bm{y})
 &=\nu(\bm{c}-\bm{y})\nu(\bm{c}-\bm{x})\\
 &\le\nu(\bm{c}-\bm{y}-\bm{e}_i+\bm{e}_j)
       \nu(\bm{c}-\bm{x}+\bm{e}_i-\bm{e}_j)\\
 &=\bar\nu(\bm{x}-\bm{e}_i+\bm{e}_j)
   \bar\nu(\bm{y}+\bm{e}_i-\bm{e}_j).
\end{align*}
Moreover, if the vectors in $\supp\nu$ have coordinate sum $r$,
then those in $\supp\bar\nu$ have coordinate sum $\sum_i c_i-r$.
Thus $\bar\nu$ satisfies both parts of the definition.

By \cite[Corollary~3.16]{BrandenHuh2020}, the exponential generating
polynomial of a log-$\mathrm M$-concave weight is Lorentzian.
Applying this separately to $\mu$ and $\bar\nu$ gives two Lorentzian
polynomials:
\begin{align*}
 g_\mu^{\mathrm{exp}}(\bm{x})
 &\coloneqq\sum_{\bm{0}\le\bm{k}\le\bm{c}}
       \frac{\mu(\bm{k})}{\bm{k}!}\bm{x}^{\bm{k}},\\
 g_{\bar\nu}^{\mathrm{exp}}(\bm{x})
 &\coloneqq\sum_{\bm{0}\le\bm{k}\le\bm{c}}
       \frac{\bar\nu(\bm{k})}{\bm{k}!}\bm{x}^{\bm{k}}
 =\sum_{\bm{0}\le\bm{k}\le\bm{c}}
       \frac{\nu(\bm{k})}{(\bm{c}-\bm{k})!}\bm{x}^{\bm{c}-\bm{k}}.
\end{align*}
{Apply polarization from
\cref{app:weight-background}, with the new variables for coordinate $i$
indexed by $\cU_i$. Since
$1/(k!\binom{c_i}{k})=(c_i-k)!/c_i!$, the resulting polynomials are
exactly $g_{\widehat\mu}$ and $g_{\widehat\nu}$ from
\cref{eq:integer-clone-weights}.} These are multiaffine log-concave
polynomials because polarization preserves the Lorentzian property
\cite[Proposition~3.1 and Theorem~2.30]{BrandenHuh2020}.

\emph{Oracle access and running time.}
Choose $S_0,T_0\subseteq\cU$ with
$|S_0\cap\cU_i|=(x_0)_i$ and
$|T_0\cap\cU_i|=c_i-(y_0)_i$ for every $i$. These satisfy
$\widehat\mu(S_0)>0$ and $\widehat\nu(T_0)>0$.
The input length is at least $\sum_i c_i=|\cU|$ under the stated
encoding. Factorials can be precomputed and
{the weights $\widehat\mu$ and $\widehat\nu$} evaluated in polynomial time. Their coefficient
ranges increase by at most a factor of $\bm{c}!$, whose logarithm
is also polynomial in the input length.
Apply \cref{thm:log-concave-convolution} and multiply the estimate by
$\bm{c}!$, using \cref{eq:integer-convolution-identity}, to obtain the
claimed $\mathsf{FPRAS}$.
\end{proof}

The proposition allows a separate log-concave weight on each coordinate.
Indeed, if $\mathcal A$ is $\mathrm M$-convex and each $a_i$ is
log-concave with interval support, then, whenever nonzero,
\begin{align*}
 \omega(\bm{k})\coloneqq
 \mathbf{1}_{\{\bm{k}\in\mathcal A\}}\prod_i a_i(k_i)
\end{align*}
satisfies \cref{eq:integer-weight-exchange}. For the exchange indices
$i,j$, log-concavity gives
$a_i(x_i)a_i(y_i)\le a_i(x_i-1)a_i(y_i+1)$ and the analogous
inequality with $j$ increased in $\bm{x}$ and decreased in $\bm{y}$.
Thus the proposition counts weighted integer points in the
intersection of two integral polymatroid base polytopes.

A \emph{contingency table} is a nonnegative integer matrix with
specified row and column sums. The next corollary allows both
entrywise upper bounds and a separate log-concave weight on each entry.

\begin{corollary}[Contingency tables with log-concave entry weights]
\label{cor:weighted-contingency-tables}
The input consists of row sums $\bm{r}\in\mathbb Z_{\ge0}^m$,
column sums $\bm{c}\in\mathbb Z_{\ge0}^n$, entrywise upper bounds
$u_{ij}\in\mathbb Z_{\ge0}$, and nonzero nonnegative log-concave
tables $w_{ij}(0),\ldots,w_{ij}(u_{ij})$ with interval support.
There is an $\mathsf{FPRAS}$ for
\begin{align*}
 Z_{\bm{r},\bm{c}}\coloneqq
 \sum_{\substack{X\in\mathbb Z_{\ge0}^{m\times n},\ X_{ij}\le u_{ij}\\
                  \sum_jX_{ij}=r_i\ (i\in[m])\\
                  \sum_iX_{ij}=c_j\ (j\in[n])}}
       \prod_{i=1}^m\prod_{j=1}^n w_{ij}(X_{ij})
\end{align*}
in the arithmetic model of \cref{thm:log-concave-convolution},
provided that $\sum_{i,j}u_{ij}$ and $\log R$ are polynomially
bounded in the input length, where
\begin{align*}
 R\coloneqq\prod_{i=1}^m\prod_{j=1}^n
       \frac{\max_{k:w_{ij}(k)>0}w_{ij}(k)}
            {\min_{k:w_{ij}(k)>0}w_{ij}(k)}.
\end{align*}
Taking $w_{ij}\equiv1$ counts contingency tables with the stated
row sums, column sums, and entrywise upper bounds.
\end{corollary}

\begin{proof}
We apply \cref{prop:integer-weight-intersection} to the intersection
of the row-sum family and the column-sum family, assigning all entry
weights to the first family.

First, find a table with positive weight, or certify that
$Z_{\bm{r},\bm{c}}=0$. This is a standard integral-flow feasibility
problem on the bipartite graph of rows and columns: the edge $(i,j)$
carries $X_{ij}$, its lower and upper bounds are the endpoints of the
positive interval of $w_{ij}$, and the required totals at the row and
column vertices are $\bm{r}$ and $\bm{c}$. It can be solved in
polynomial time. Assume henceforth that such a table $X_0$ exists.
If all $u_{ij}=0$, evaluate the single possible table directly.
Otherwise, regard matrices as vectors of the entries with $u_{ij}>0$;
the omitted entries remain fixed at zero.

Let $\mathcal A$ and $\mathcal D$ be the families satisfying
$0\le X_{ij}\le u_{ij}$ and, respectively, the row sums and the
column sums. Both are $\mathrm M$-convex. Indeed, if $X_{ij}>Y_{ij}$,
the equal sum in that row (respectively, column) gives another entry
where $X$ is smaller than $Y$. Exchanging one unit between these two
entries in opposite directions in $X$ and $Y$ preserves the required
sums and entry bounds. Set
\begin{align*}
 \mu(X)&\coloneqq\mathbf{1}_{\{X\in\mathcal A\}}
                  \prod_{i=1}^m\prod_{j=1}^n w_{ij}(X_{ij}),
 &\nu(X)&\coloneqq\mathbf{1}_{\{X\in\mathcal D\}}.
\end{align*}
{The preceding observation about products of coordinate weights shows}
that these weights
are log-$\mathrm M$-concave, with coefficient ranges at most $R$ and
$1$. Their evaluations take polynomial time, and $X_0$ has positive
weight under both. Their inner product is $Z_{\bm{r},\bm{c}}$.
The coordinate upper bounds have polynomial sum, so they can be
represented in unary. Thus \cref{prop:integer-weight-intersection}
gives the claimed $\mathsf{FPRAS}$.
\end{proof}

\begin{corollary}[Bounded integer flows]\label{cor:bounded-integer-flows}
Let $B$ be the incidence matrix of a directed graph $(V,E)$, with
outgoing entries $+1$ and incoming entries $-1$. The input consists
of a vector $\bm{h}\in\mathbb Z^V$, capacities $d_a\in\mathbb Z_{\ge0}$,
and nonzero nonnegative log-concave tables $w_a(0),\ldots,w_a(d_a)$
with interval support. Each $h_v$ specifies the required outgoing
flow minus incoming flow at vertex $v$.
There is an $\mathsf{FPRAS}$ for
\begin{align*}
 Z_{\bm{h}}\coloneqq\sum_{\substack{\bm{q}\in\mathbb Z^E,\ B\bm{q}=\bm{h}\\
                           0\le q_a\le d_a}}
                   \prod_{a\in E}w_a(q_a)
\end{align*}
in the arithmetic model of \cref{thm:log-concave-convolution},
provided that $\sum_a d_a$ and $\log R$ are polynomially bounded
in the input length, where
\begin{align*}
 R\coloneqq\prod_{a\in E}
       \frac{\max_{k:w_a(k)>0}w_a(k)}{\min_{k:w_a(k)>0}w_a(k)}.
\end{align*}
Taking $w_a\equiv1$ counts all feasible integer flows.
\end{corollary}

\begin{proof}
We give a weight-preserving bijection between flows and contingency
tables, then apply \cref{cor:weighted-contingency-tables}.
Write $D\coloneqq\sum_a d_a$ for the total flow capacity.
If $\sum_v h_v\ne0$, there is no feasible flow, so return zero.

\emph{From flows to square matrices.}
We construct an $N\times N$ nonnegative integer matrix $X$, where
$N=|V|+|E|$ and the rows and columns are indexed by the disjoint union
$V\sqcup E$. Giving each edge its own index keeps parallel edges and
loops distinct. The required row sums $R_i$ and column sums $C_i$ are
\begin{align*}
 R_v&\coloneqq\sum_{a\text{ out of }v}d_a,
 &C_v&\coloneqq R_v-h_v &&(v\in V),\\
 R_a&\coloneqq d_a,
 &C_a&\coloneqq d_a &&(a\in E).
\end{align*}
If any $C_v<0$, return zero: a feasible flow would satisfy
$h_v\le\sum_{a\text{ out of }v}q_a\le R_v$.
The row and column sums are now nonnegative and satisfy
$\sum_iR_i=\sum_iC_i=2D$.
Given a flow $\bm{q}$, start with $X=0$ and set
\begin{align*}
 X_{u,a}=X_{a,v}&\coloneqq q_a,
 &X_{a,a}&\coloneqq d_a-q_a &&(a=u\to v),\\
 X_{v,v}&\coloneqq R_v-\sum_{a\text{ out of }v}q_a
 && &&(v\in V).
\end{align*}
All entries not assigned above remain zero. In particular,
\begin{align*}
 X_{u,v}&=0 &&(u,v\in V,\ u\ne v),\\
 X_{a,b}&=0 &&(a,b\in E,\ a\ne b).
\end{align*}
For an edge $a=u\to v$, the mixed entries also satisfy
$X_{w,a}=0$ for $w\in V\setminus\{u\}$ and
$X_{a,w}=0$ for $w\in V\setminus\{v\}$.
Allow each of
$X_{u,a},X_{a,v},X_{a,a}$ to lie in $[0,d_a]$, and each $X_{v,v}$
to lie in $[0,R_v]$. These bounds hold for the matrix just constructed,
and flow conservation gives
\begin{align*}
 \sum_{j\in V\sqcup E}X_{i,j}=R_i,\qquad
 \sum_{j\in V\sqcup E}X_{j,i}=C_i
 \qquad(i\in V\sqcup E).
\end{align*}
Conversely, consider a matrix with these entry bounds, fixed zero
entries, and row and column sums. For each edge $a=u\to v$, its row
and column equations give
\begin{align*}
 X_{a,v}+X_{a,a}=d_a=X_{u,a}+X_{a,a},
\end{align*}
so $X_{u,a}=X_{a,v}$ defines a unique flow $q_a$.
Subtracting the column equation at vertex $v$ from its row equation
gives outgoing flow minus incoming flow equal to $R_v-C_v=h_v$.
All diagonal entries are then uniquely determined, proving the
bijection.

\emph{Applying the contingency-table result.}
Assign weight $w_a(X_{u,a})$ to the entry $X_{u,a}$ and weight one
to the other entries. The bijection preserves weights, and every
entry-weight table is log-concave with interval support. The product
of their coefficient ranges is $R$, and their upper bounds sum to
\begin{align*}
 3\sum_{a\in E}d_a+\sum_{v\in V}R_v=4D.
\end{align*}
Thus \cref{cor:weighted-contingency-tables} gives an
$\mathsf{FPRAS}$ for $Z_{\bm{h}}$.
\end{proof}

\subsection{Vertices of Network-Matrix Polytopes}\label{app:network-polytopes}

We now apply \cref{cor:bounded-integer-flows} to count vertices of
network-matrix polytopes. The vertices correspond to integer flows;
the key point is that their binary coordinates let us bound the total
flow capacity by a polynomial in the matrix dimensions.

A \emph{network matrix} records signed paths in a directed tree. More
precisely, let $T$ be a directed tree with $m$ edges, and let $E$ consist
of $n$ additional labeled directed edges on the same vertex set.
Write $V$ for this vertex set, so $|V|=m+1$. For each directed edge
$e=(u\to v)$, its incidence entry in row $w$ is
$b_{w,e}=\mathbf{1}_{\{w=u\}}-\mathbf{1}_{\{w=v\}}$.
Choose any vertex $r\in V$ and define
\[
\begin{aligned}
 B_T&=(b_{w,e})_{w\in V\setminus\{r\},\ e\in E(T)}
       \in\mathbb R^{m\times m},\\
 B_E&=(b_{w,e})_{w\in V\setminus\{r\},\ e\in E}
       \in\mathbb R^{m\times n}.
\end{aligned}
\]
Thus we omit the row for vertex $r$ from both incidence matrices and
keep all other vertex rows in the same order. To see that the square
matrix $B_T$ is invertible, let $\bm{z}\in\mathbb R^{V\setminus\{r\}}$
satisfy $B_T^{\mathsf T}\bm{z}=\bm{0}$, and put $z_r=0$.
For every tree edge $e=(u\to v)$,
\[
 0=(B_T^{\mathsf T}\bm{z})_e=z_u-z_v.
\]
Every vertex is joined to $r$ by a path when edge directions are ignored,
so these equalities give $z_w=z_r=0$ for every $w\in V$. Thus
$B_T^{\mathsf T}\bm{z}=\bm{0}$ has only the zero solution, proving
that $B_T$ is invertible. The matrix $A=B_T^{-1}B_E$ is a network matrix.
Given only a network matrix $A$, one can find in polynomial time a
directed tree $T$ and additional directed edges $E$ such that
$A=B_T^{-1}B_E$ \cite{NetworkBixbyCunningham1980}.
Network matrices are totally unimodular, so in particular
their entries lie in $\{0,\pm1\}$. Consequently, for integral
bounds $\bm{\ell},\bm{u}$, the vertices of
\[
 P=\{\bm{x}\in[0,1]^n:\bm{\ell}\leq A\bm{x}\leq\bm{u}\}
\]
are exactly its binary points: integrality gives one direction, and
every binary point of a polytope contained in the cube is a vertex.

The following result answers the network-matrix case posed in
\cite[Problem~4.3]{NetworkGuoJerrum2023}. It follows by combining
their reduction to integer flows \cite[Section~4]{NetworkGuoJerrum2023}
with \cref{cor:bounded-integer-flows}.

\begin{corollary}\label{cor:network-polytope-vertices}
There is an $\mathsf{FPRAS}$ for the number of vertices of $P$ above,
with running time polynomial in the input length and $\epsilon^{-1}$
in the arithmetic model of \cref{thm:log-concave-convolution}.
This includes a polytope $\{\bm{x}:A\bm{x}\leq\bm{b}\}$ promised to have only binary
vertices, for a network matrix $A$ and integral $\bm{b}$.
\end{corollary}

\begin{proof}
We recall the reduction of Guo and Jerrum
\cite[Section~4]{NetworkGuoJerrum2023} and apply
\cref{cor:bounded-integer-flows} with unit weights. We verify the
bijection from vertices to flows and bound the total capacity in
terms of the dimensions of $A$.

The vertices of $P$ are exactly its binary points. For each such point
$\bm{x}$, set $\bm{y}=A\bm{x}$. The identity $A=B_T^{-1}B_E$ gives
\begin{align*}
 \bm{y}=A\bm{x}
 \quad\Longleftrightarrow\quad
 B_T\bm{y}-B_E\bm{x}=\bm{0}.
\end{align*}
For each labeled edge $e=(u\to v)\in E$, let
$\overleftarrow e=(v\to u)$ denote the same edge with its direction
reversed, and set $\overleftarrow E\coloneqq\{\overleftarrow e:e\in E\}$.
On the graph with edge set
$E(T)\sqcup\overleftarrow E$, write $f_a$ for the flow on edge $a$:
set $f_t=y_t$ on each tree edge and $f_{\overleftarrow e}=x_e$
on the reversal of each edge $e\in E$. The equation on the right
expresses \emph{flow conservation}: total outgoing flow equals total
incoming flow at each vertex whose row appears in the matrices,
\begin{align*}
 \sum_{a\text{ out of }v}f_a
 =\sum_{a\text{ into }v}f_a
 \qquad(v\in V\setminus\{r\}).
\end{align*}
Here $r$ is the vertex whose row was deleted when defining $B_T$ and
$B_E$ above. The same equality also holds at $r$: every edge contributes
once to an outgoing sum and once to an incoming sum, so the sum of
outgoing minus incoming flow over all vertices is zero. Since this
difference vanishes at every $v\ne r$, it must vanish at $r$ as well.
The tree flows are uniquely determined by $\bm{x}$, so this is a
bijection between vertices of $P$ and integer flows satisfying
$\ell_t\le y_t\le u_t$ and $0\le x_e\le1$.

We next shift the tree flows to make all flows nonnegative and bound
their capacities. Let $a_t^+,a_t^-$ count the $+1,-1$ entries in row $t$
of $A$. Since $y_t=(A\bm{x})_t$, every binary $\bm{x}$ satisfies
$-a_t^-\le y_t\le a_t^+$. We may therefore replace the original
constraint $\ell_t\le y_t\le u_t$ on the flow through tree edge $t$
by $L_t\le y_t\le U_t$, without changing the feasible flows, where
\begin{align*}
 L_t\coloneqq\max\{\ell_t,-a_t^-\},\qquad
 U_t\coloneqq\min\{u_t,a_t^+\}.
\end{align*}
If any $L_t>U_t$, the count is zero. Otherwise put
$q_t=y_t-L_t$ on tree edges and $q_e=x_e$ on reversed edges.
Let $B$ be the full vertex-edge incidence matrix of this graph, with outgoing
entries $+1$ and incoming entries $-1$. The equivalent constraints are
\begin{align*}
 B\bm{q}=\bm{h},\qquad 0\le q_a\le d_a,\qquad
 \bm{h}=-B(\bm{L},\bm{0}),
\end{align*}
where $d_t=U_t-L_t\le a_t^++a_t^-\le n$ and $d_e=1$.
In particular, $D\coloneqq\sum_a d_a\le mn+n$.
Thus \cref{cor:bounded-integer-flows}, with $w_a\equiv1$ and $R=1$,
gives the required approximation using polynomially many arithmetic
operations in $m+n$ and $\epsilon^{-1}$.

Finally, if $Q=\{\bm{x}:A\bm{x}\le\bm{b}\}$ is a polytope with
only binary vertices, then $Q\subseteq[0,1]^n$, since a polytope is
the convex hull of its vertices. Adding the cube bounds therefore
does not change $Q$. The lower bounds
$\ell_t=-a_t^-$ hold throughout the cube, so taking these bounds
and $\bm{u}=\bm{b}$ puts $Q$ in the form of $P$.
\end{proof}

\subsection{Rainbow Matchings and Degree-Constrained Forests}\label{app:curvature-applications}

Let $G=(V,E)$ be a simple graph with rational edge weights
$0\le w_e\le1$, and write $\bm{w}^S\coloneqq\prod_{e\in S}w_e$.
For $S\subseteq E$, let
$\deg_S(v)\coloneqq|\{e\in S:v\in e\}|$ be the degree of $v$
in the subgraph $(V,S)$.
A matching is \emph{rainbow} if its edges have distinct colors;
the coloring need not be proper.

\begin{corollary}\label{cor:rainbow-matchings}
Given any edge coloring of $G$, there is an $\mathsf{FPRAS}$ for
\[
 \sum_{\substack{M\subseteq E:\\M\text{ a rainbow matching}}}\bm{w}^M.
\]
In particular, taking every $w_e=1$ counts all rainbow matchings.
\end{corollary}

For nonnegative integer bounds $b_v$, a \emph{degree-constrained forest}
is an acyclic edge set $F\subseteq E$ with $\deg_F(v)\le b_v$ for
every vertex $v$.

\begin{corollary}\label{cor:degree-constrained-forests}
There is an $\mathsf{FPRAS}$ for
\[
 \sum_{\substack{F\subseteq E\text{ a forest}:\\
                   \deg_F(v)\le b_v\ \forall v}}
       \bm{w}^F.
\]
In particular, taking $b_v=2$ and $w_e=1$ counts all linear forests,
that is, all unions of vertex-disjoint paths, with isolated vertices
allowed.
\end{corollary}
Both sums range over all sizes, including the empty edge set.

\begin{proof}[Proof of \cref{cor:rainbow-matchings,cor:degree-constrained-forests}]
For the fixed graph $G$ and any matroid $\cN$ on its edge
set $E$, given by an independence oracle, we approximate
\[
 Z_0\coloneqq
 \sum_{\substack{S\in\cI(\cN):\\\deg_S(v)\le b_v\ \forall v}}\bm{w}^S.
\]
For rainbow matchings, take $b_v=1$ and let $\cN$ be the partition
matroid allowing at most one edge of each color. For forests, take
the graphic matroid.
We first replace $Z_0$ by a weighted sum $Z_q$ that gives
a small weight to dependent sets, and bound the resulting error.
We then express $Z_q$ through a convolution coefficient and apply
\cref{thm:curvature-convolution}.

Delete zero-weight edges and edges incident to a vertex
$v$ with $b_v=0$, restricting $\cN$ accordingly. Write $n\coloneqq|E|$; the case
$n=0$ has $Z_0=1$. Define the measure
\[
 \mu(S)\coloneqq
       \allone[\deg_S(v)\le b_v\ \forall v]\bm{w}^S.
\]
Its support is downward closed.

For $0<\epsilon\le1$, put $q\coloneqq\epsilon/(8n)$ and define
\[
 Z_q\coloneqq\sum_{S\in\supp\mu}
                    \mu(S)q^{|S|-{\rk_{\cN}}(S)}.
\]
Independent sets contribute exactly $Z_0$, since their exponent is
zero. To bound the additional contribution, for each $S\in\supp\mu$
choose a basis $I(S)$ of $\cN|S$ by fixed-order greedy.
Then $I(S)\subseteq S$ satisfies the degree bounds,
$|S\setminus I(S)|=|S|-{\rk_{\cN}}(S)$, and
$\mu(S)\le\mu(I(S))$ because $w_e\le1$.
The map $S\mapsto(I(S),S\setminus I(S))$ is injective. Therefore
\[
 Z_0\le Z_q
 \le\sum_{\substack{I\in\supp\mu\cap\cI(\cN)}}
       \mu(I)\sum_{T\subseteq E\setminus I}q^{|T|}
 \le(1+q)^nZ_0\le(1+\epsilon/4)Z_0,
\]
where $(1+q)^n\le e^{\epsilon/8}\le1+\epsilon/4$.
{Thus an estimate of $Z_q$ with relative error at most
$\epsilon/2$ also estimates $Z_0$ with relative error at most $\epsilon$.}

To approximate $Z_q$, define
\[
 \mu_q(S)\coloneqq q^{|S|}\mu(S),\qquad
 \nu_q(T)\coloneqq q^{|T|-{\rk_{\cN^*}}(T)}.
\]
For every $T\subseteq E$, the matroid dual-rank identity is
\[
 {\rk_{\cN^*}}(T)=|T|+{\rk_{\cN}}(E\setminus T)-{\rk_{\cN}}(E).
\]
Taking $T=E\setminus S$ gives
$\nu_q(E\setminus S)=q^{{\rk_{\cN}}(E)-{\rk_{\cN}}(S)}$. Hence
\[
 [\bm{x}^E](g_{\mu_q}g_{\nu_q})
 =\sum_{S\subseteq E}\mu_q(S)\nu_q(E\setminus S)
 =q^{{\rk_{\cN}}(E)}Z_q.
\]

We now verify the hypotheses of \cref{thm:curvature-convolution}.
First, $\mu$ has upward spectral curvature $-1$.
Fix $S\in\supp\mu$ and put
$W_S\coloneqq\{e\in E\setminus S:S+e\in\supp\mu\}$.
Let $B$ be the unsigned incidence matrix of $W_S$ on the
vertices $v$ with $b_v-\deg_S(v)=1$, where at most one further incident
edge can be added. Let $d_e$ be its column sums, and put
$D\coloneqq\operatorname{diag}(w_e)$ on these edges.
Two available edges cannot be added together precisely when they
share such a vertex. Since $G$ is simple, they
share at most one vertex. The principal submatrix indexed by $W_S$ is
\[
 (A_S^{\mu,-1})_{W_S\times W_S}
 =DB^{\mathsf T}BD+
   \operatorname{diag}_e\bigl(w_e-(d_e-1)w_e^2\bigr)
 \succeq0.
\]
Indeed, $d_e\le2$ and $w_e\le1$. Unavailable edges contribute zero
rows and columns. Also, every addition ratio of $\mu$ is at most one.

Consequently, $\mu_q$ has curvature $-q$, since
$A_S^{\mu_q,-q}=q^2A_S^{\mu,-1}$.
The measure $\nu_q$ has full support and curvature $q$. To see this, consider
the contraction $\cN^*/T$: its loops have addition ratio $q$, and
its nonloops have addition ratio one. In $A_T^{\nu_q,q}$ the loop
rows vanish, each nonloop parallel class gives a block $(1-q)J$,
where $J$ is the all-ones matrix, and all other entries vanish.
Thus both supports are downward closed and their curvatures sum to
zero. Their addition ratios are at most $q$ and one, respectively,
so we may supply $\eta=-q$, $\kappa=\gamma=q$, and $\chi=q^{-1}=8n/\epsilon$ to \cref{thm:curvature-convolution}.

Degree checks and greedy rank queries evaluate both
measures using the dual-rank identity above.
Both are positive at the empty set, and
\[
 \log(\cR_{\mu_q}\cR_{\nu_q})
 \le 2n\log(1/q)+\sum_{e\in E}\log(1/w_e).
\]
This bound and $\chi=8n/\epsilon$ are polynomial in the input length
and $\epsilon^{-1}$. Apply \cref{thm:curvature-convolution} with
relative error $\epsilon/2$ and divide its coefficient estimate by
$q^{{\rk_{\cN}}(E)}$. {This gives an estimate of $Z_q$ with
relative error at most $\epsilon/2$, and hence an estimate of $Z_0$
with relative error at most $\epsilon$.}
\end{proof}

\section{Miscellaneous Aspects of \cref{thm:poincare}}

\subsection{A Matroid Intersection Analog of the Monomer-Dimer Model}\label{app:matroid-monomer-dimer}
In this section, we consider a generalization of the classical monomer-dimer model over matchings of a graph to matroid intersection. We use \cref{thm:poincare} to prove a spectral gap for the standard heat-bath Glauber dynamics that decays smoothly as the parameter $\lambda$ grows. In particular, the model does not exhibit a dynamical phase transition in $\lambda$. The motivation for studying this problem was discussed in \cref{subsec:ack}.

Let $\cM,\cN$ be two matroids on a common ground set, and fix a parameter $\lambda \geq 0$. We consider the following probability measure over common independent sets of all sizes:
\begin{align*}
    \mu_{\lambda}(S) \propto \lambda^{\abs{S}}, \qquad \forall S \in \cI(\cM) \cap \cI(\cN).
\end{align*}
The discrete-time heat-bath Glauber dynamics is a Markov chain whose evolution can be described as follows: If the state at time step $t$ is $S_{t}$, then we transition to $S_{t+1}$ by
\begin{enumerate}
    \item sampling a uniformly random element $u \in \cU$, and
    \item resampling the in/out status of $u$ conditioned on all other elements. More specifically, if $u \in S_{t}$, then set $S_{t+1} = S_{t} - u$ with probability $\frac{1}{1 + \lambda}$. Otherwise, if $u \notin S_{t}$ and $S_{t} + u \in \cI(\cM) \cap \cI(\cN)$, then set $S_{t+1} = S_{t} + u$ with probability $\frac{\lambda}{1 + \lambda}$. In all other cases, set $S_{t+1} = S_{t}$.
\end{enumerate}
Let $P_{\lambda}$ denote the corresponding transition probability matrix, and let $\cL_{\lambda} \coloneqq I - P_{\lambda}$. It is straightforward to verify that $P_{\lambda},\cL_{\lambda}$ are reversible with respect to $\mu_{\lambda}$. For $\lambda>0$, the chain is irreducible since every state can reach $\emptyset$ by deletions, and each such path can be reversed.
\begin{theorem}\label{thm:softintersect-heatbath}
For every $\cM,\cN$ with $n = \abs{\cU}\geq1$ and every $\lambda \geq 0$, we have the spectral gap lower bound
\begin{align*}
    \gamma(\cL_{\lambda}) \geq \frac{1}{n \cdot (1 + \lambda) \cdot (1 + \lambda \cdot (n-1))}.
\end{align*}
\end{theorem}

Following the strategy used above for common bases, we first establish a Poincar\'{e} inequality for a softened version of $\mu_{\lambda}$ using \cref{thm:poincare}. We then deduce \cref{thm:softintersect-heatbath} by a direct comparison of Dirichlet forms.
\begin{lemma}\label{lem:softintersect-poincare}
For $0 \leq t \leq 1$, define the distribution
\begin{align*}
    \mu_{\lambda}^{(t)}(S) \propto \lambda^{\abs{S}} \cdot t^{\abs{S} - \rk_{\cN}(S)}, \qquad \forall S \in \cI(\cM).
\end{align*}
Then for every $f : \cI(\cM) \to \R$, we have the Poincar\'{e} Inequality
\begin{align*}
    \Var[\mu_{\lambda}^{(t)}]{f} \leq \cE_{\lambda}^{(t)}(f,f) \coloneqq \lambda \cdot \sum_{S \in \cI(\cM)} \sum_{\substack{u \notin S \\ S + u \in \cI(\cM)}} \mu_{\lambda}^{(t)}(S) \cdot \wrapp{f(S + u) - f(S)}^{2}.
\end{align*}
\end{lemma}
\begin{remark}
As discussed in \cref{rmk:partition-penalty}, the ``softening'' of $\mu_{\lambda}$ into $\mu_{\lambda}^{(t)}$ can be viewed as a generalization of the ``softening'' used in the proof of \cref{thm:homogeneous-convolution-informal}.
\end{remark}
\begin{proof}
For $\lambda>0$ and $0<t\le1$, define the weight functions
\[
 \pi(S)\coloneqq\lambda^{|S|}\allone[S\in\cI(\cM)],
 \qquad
 \nu(S)\coloneqq t^{|S|-\rk_{\cN}(S)}\allone[S\in\cI(\cM)].
\]
Their common support $\cI(\cM)$ is downward closed and connected,
and $\pi\odot\nu=\mu_{\lambda}^{(t)}$. We claim that $\pi$ has
upward $0$-spectral curvature and $\nu^*$ has upward
$1$-spectral curvature. Assuming these claims, \cref{thm:poincare}
with $\eta=0$ and $\kappa=1$ gives the desired inequality:
the dynamics has rate $\pi(S+u)/\pi(S)=\lambda$ for every feasible
addition, so its Dirichlet form is exactly
$\cE_{\lambda}^{(t)}(f,f)$ by reversibility. The case $t=0$
follows by continuity of the finite sums, and $\lambda=0$
gives a point mass at $\emptyset$, for which both sides vanish.
It remains to verify the two claims.

First, for $S\in\cI(\cM)$, the feasible additions are the nonloops of
$\cM/S$. Two distinct feasible additions $u,v$ are called
\emph{parallel} in $\cM/S$ if $S+u+v\notin\cI(\cM)$.
Together with equality, this is an equivalence relation on the
nonloops of $\cM/S$, partitioning them into equivalence classes.
Thus $A_S^{\pi,0}$ is a direct
sum of blocks $\lambda^2J$ on these classes and zero
rows for infeasible additions, where $J$ denotes an all-ones
matrix. Thus $A_S^{\pi,0}\succeq0$, proving the first claim.

Second, fix $S\in\cI(\cM)$, so
$\cU\setminus S\in\supp\nu^*$, and put
$\cK\coloneqq(\cN|_S)^*$.
For every $T\subseteq S$, applying \cref{eq:matroid-dual-rank}
to the matroid $\cN|_S$ on ground set $S$ gives
\[
 \rk_{\cK}(T)
 =|T|+\rk_{\cN|_S}(S\setminus T)-\rk_{\cN|_S}(S)
 =|T|+\rk_{\cN}(S\setminus T)-\rk_{\cN}(S).
\]
Since $S\setminus T\in\cI(\cM)$, the definitions yield
\begin{align*}
 \frac{\nu^*((\cU\setminus S)\cup T)}{\nu^*(\cU\setminus S)}
 &=\frac{\nu(S\setminus T)}{\nu(S)}
 =\frac{t^{|S\setminus T|-\rk_{\cN}(S\setminus T)}}
        {t^{|S|-\rk_{\cN}(S)}}
 =t^{-|T|+\rk_{\cN}(S)-\rk_{\cN}(S\setminus T)}
 =t^{-\rk_{\cK}(T)}.
\end{align*}
Consequently, $A_{\cU\setminus S}^{\nu^*,1}$ is a direct sum of blocks
$\frac{1-t}{t^2}J$ on the equivalence classes under parallelism
among the nonloops of $\cK$,
with zero rows at its loops and outside $S$. Since $0 < t\leq 1$, these blocks are
positive semidefinite.
\end{proof}

\begin{proof}[Proof of \cref{thm:softintersect-heatbath}]
The case $\lambda=0$ is trivial so assume $\lambda>0$.
Let $\Omega \coloneqq \cI(\cM)\cap\cI(\cN)$.
For $S\in\cI(\cM)\setminus\Omega$, write
$D_S\coloneqq\{u\in S:S-u\in\Omega\}$ and $d_S\coloneqq|D_S|$.
Fix a function $f:\Omega \to \R$.
We extend it to the function $f^{\uparrow}:\cI(\cM) \to \R$ as follows:
\[
 f^{\uparrow}(S)\coloneqq
 \begin{cases}
  f(S), & S\in\Omega,\\[2pt]
  \displaystyle\frac{1}{d_S}\sum_{u\in D_S}f(S-u),
    & S\in\cI(\cM)\setminus\Omega,\ d_S>0,\\[6pt]
  0, & S\in\cI(\cM)\setminus\Omega,\ d_S=0.
 \end{cases}
\]
By \cref{lem:softintersect-poincare} at $t=0$,
\[
 \Var[\mu_\lambda]{f} = \Var[\mu_\lambda^{(0)}]{f^{\uparrow}}\le\cE_\lambda^{(0)}\wrapp{f^{\uparrow},f^{\uparrow}}.
\]
We will finish the proof by showing that
\begin{align*}
	\cE_\lambda^{(0)}\wrapp{f^{\uparrow},f^{\uparrow}}
	&\leq n(1+\lambda) (1 + \lambda(n-1)) \cdot \langle f,\cL_\lambda f\rangle_{\mu_\lambda}.
\end{align*}
For convenience, put
\[
 \cE_{\lambda}(f,f)
 \coloneqq
 n(1+\lambda)\langle f,\cL_\lambda f\rangle_{\mu_\lambda}
 =
 \lambda\sum_{\substack{S\in\Omega,\ u\notin S\\S+u\in\Omega}}
 \mu_\lambda(S)\wrapp{f(S+u)-f(S)}^2.
\]
By the definition of $f^{\uparrow}$,
\begin{align*}
 \cE_\lambda^{(0)}\wrapp{f^{\uparrow},f^{\uparrow}}
 &=\cE_\lambda(f,f)+
 \lambda\sum_{S\in\cI(\cM)\setminus\Omega}
 \sum_{u\in D_S}
 \mu_\lambda(S-u)\wrapp{f^{\uparrow}(S)-f^{\uparrow}(S-u)}^2\\
 &=\cE_\lambda(f,f)+
 \underset{(\ast)}{\underbrace{\sum_{\substack{S\in\cI(\cM)\setminus\Omega\\d_S\ge2}}
 \frac{\lambda}{d_S}
 \sum_{\{u,v\}\subseteq D_S}
 \mu_\lambda(S-u)\wrapp{f(S-u)-f(S-v)}^2}}.
\end{align*}
In the last equality, we used that for any probability measure $\zeta$ on a domain $\Omega$ and any test function $\varphi : \Omega \to \R$, we have $\Var[\omega \sim \zeta]{\varphi(\omega)} = \frac{1}{2} \cdot \E[\omega,\omega' \sim \zeta]{\wrapp{\varphi(\omega) - \varphi(\omega')}^{2}}$.
The inner sum is over unordered pairs of distinct elements;
its summand is symmetric because
$\mu_\lambda(S-u)=\mu_\lambda(S-v)$.

We control $(\ast)$ by a multiple of $\cE_{\lambda}(f,f)$. To do this, observe that for each $S \in \cI(\cM) \setminus \Omega$ with $d_{S} \geq 2$ and each pair $\{u,v\} \subseteq D_{S}$, its contribution to $(\ast)$ can be viewed as a ``swap'' move that transitions from $S-u$ to $S-v$. This swap can be simulated using two steps of Glauber dynamics $P_{\lambda}$, namely $S - u \to S - u - v$ followed by $S - u - v \to S - v$. With this intuition in mind, we invoke the numerical inequality $(a-b)^2\le2((a-c)^2+(b-c)^2)$ and
$\mu_\lambda(S-u)=\lambda\mu_\lambda(S-u-v)$ to obtain
\begin{align*}
 (\ast)
 &\le
 \sum_{\substack{S\in\cI(\cM)\setminus\Omega\\d_S\ge2}}
 \frac{2\lambda^2}{d_S}
 \sum_{\substack{u,v\in D_S\\u\ne v}}
 \mu_\lambda(S-u-v)\wrapp{f(S-u)-f(S-u-v)}^2\\
 &\le\lambda^2
 \sum_{\substack{T\in\Omega,\ v\notin T\\T+v\in\Omega}}
 (n-|T|-1)\mu_\lambda(T)\wrapp{f(T+v)-f(T)}^2
 \tag{Setting $T = S - u - v$ and using $d_{S} \geq 2$}\\
 &\le\lambda(n-1)\cE_\lambda(f,f).
 \tag{\(|T|\ge0\)}
\end{align*}
This finishes the proof.
\end{proof}



\end{document}